\documentclass[11pt,letterpaper]{article}
\usepackage{graphicx}
\usepackage{amssymb}
\usepackage{amsthm}
\usepackage{amsmath}
\usepackage[margin=1in]{geometry}
\usepackage{placeins}
\usepackage[font={small,it}]{caption}
\usepackage{subfigure}
\usepackage{authblk}

\numberwithin{equation}{section}
\theoremstyle{plain}
\newtheorem{theorem}{Theorem}[section]
\newtheorem{corollary}[theorem]{Corollary}
\newtheorem{lemma}[theorem]{Lemma}
\newtheorem{proposition}[theorem]{Proposition}
\newtheorem{definition}{Definition}[section]

\newcommand{\R}{\mathbb{R}}
\newcommand{\A}{\mathcal{A}}
\newcommand{\Sp}{\mathbb{S}}
\DeclareMathOperator*{\argmin}{\arg\!\min}

\usepackage{natbib}
\title{Sufficient Dimension Reduction and Modeling Responses Conditioned on Covariates: An Integrated Approach via Convex Optimization}
\author{Armeen Taeb $^\dag$ and Venkat Chandrasekaran $^\dag, ^\ddag$
\thanks{Email: ataeb@caltech.edu, venkatc@caltech.edu} \vspace{.25in} \\ $^\dag$ Department of Electrical Engineering \\  $^\ddag$ Department of Computing and Mathematical Sciences \\ California Institute of Technology \\ Pasadena, Ca 91125}

\begin{document}
\maketitle

\begin{abstract}
Given observations of a collection of covariates and responses $(Y, X) \in \R^p \times \R^q$, sufficient dimension reduction (SDR) techniques aim to identify a mapping $f: \R^q \rightarrow \R^k$ with $k \ll q$ such that $Y|f(X)$ is independent of $X$. The image $f(X)$ summarizes the relevant information in a potentially large number of covariates $X$ that influence the responses $Y$. In many contemporary settings, the number of responses $p$ is also quite large, in addition to a large number $q$ of covariates. This leads to the challenge of fitting a succinctly parameterized statistical model to $Y|f(X)$, which is a problem that is usually not addressed in a traditional SDR framework. In this paper, we present a computationally tractable convex relaxation based estimator for simultaneously (a) identifying a linear dimension reduction $f(X)$ of the covariates that is sufficient with respect to the responses, and (b) fitting several types of structured low-dimensional models -- factor models, graphical models, latent-variable graphical models -- to the conditional distribution of $Y|f(X)$. We analyze the consistency properties of our estimator in a high-dimensional scaling regime. We also illustrate the performance of our approach on a newsgroup dataset and on a dataset consisting of financial asset prices.\end{abstract}


\section{Introduction} \label{sec:intro}
Sufficient dimension reduction (SDR) is a framework for identifying a low-dimensional approximation of a large collection of covariates that is sufficient for predicting a set of responses \citep{LiSDR1991, Duan1991, CookSAVE1991}. Given covariates $X \in \R^q$ and responses $Y \in \R^p$, the objective of SDR is to obtain a mapping $f: \R^q \rightarrow \R^k$ with $k \ll q$ such that the responses $Y$ are independent of the covariates $X$ conditioned on $f(X)$; equivalently, the conditional distribution of $Y | f(X)$ is the same as that of $Y | X$. The image $f(X)$ is called a dimension reduction of the covariates $X$ that is sufficient with respect to the responses $Y$, and it summarizes the relevant information in $X$ that influences $Y$.

In many contemporary settings, the number of responses $Y$ is also quite large (in addition to a potentially large number of covariates $X$). For example, in financial modeling applications the responses are the prices of financial assets (numbering in the several hundreds) and the covariates may be macroeconomic indicators (see the numerical experiment in Section~\ref{section:experiments}). In gene microarray analysis, the responses correspond to the expression levels of a large number of genes (on the order of tens of thousands), and the covariates could be a collection of physiological attributes \citep{Cheung2002,Brem2005}. In these problem domains, the conditional distribution of $Y | f(X)$ is specified using a large number of parameters, and as a result the estimation of this conditional distribution leads to several statistical and computational challenges -- e.g., problems with overfitting when given a modest number o
f observations, and difficulties with developing algorithmic procedures that operate within a reasonable computational budget. These challenges in high-dimensional inference are by now well-recognized \citep{BuhV2011, Wai2014}, and they have been addressed in several settings based on approximations of high-dimensional distributions consisting of many degrees of freedom by elements from structured classes of models specified using a small number of parameters; examples of such structured families include models described by sparse or banded covariance matrices \citep{Bicla2008, Biclb2008, Elk(2008)}, graphical models \citep{YuanLin2006, Friedman2008,RavWRY2008,RotBLZ2008}, and factor models \citep{Fan2008}; see \citet{Wai2014} for a more extensive list of references.
\subsection{Our Contributions}
\label{section:contribution}
Building on this prior literature, we describe a new methodology based on convex optimization that integrates sufficient dimension reduction with techniques to fit succinctly parameterized models to the high-dimensional conditional distribution of the responses given the covariates. In particular, given observations of a set of responses $Y$ and covariates $X$, we fit a linear Gaussian model with the following two properties: $(a)$ there exists a low-dimensional linear dimension reduction\footnote{If the covariates and responses are jointly Gaussian, it suffices to consider dimension reductions specified by linear mappings of the covariates \citep{LiSDR1991,ChechikInform2005}. Even in more general settings, many approaches to SDR construct linear dimension reductions of the covariates for computational reasons \citep{LiSDR1991,Duan1991,CookSAVE1991,LiSDR1992,CookNi2005}} $f(X)$ of the covariates that is sufficient with respect to the responses, and $(b)$ the conditional distribution of $Y | f(X)$ is specified by a concisely parameterized statistical model -- the three concrete examples that we consider are a factor model, a graphical model, and a latent-variable graphical model.

If $(Y,X) \in \R^p \times \R^q$ is a jointly Gaussian random vector with covariance matrix \\\allowbreak{$\Sigma = \begin{pmatrix} \Sigma_Y & \Sigma_{YX} \\ \Sigma_{YX}' & \Sigma_{X} \end{pmatrix}$}, the existence of a dimension-$k$ linear projection of the covariates $X$ that is sufficient with respect to the responses $Y$ is equivalent to the rank of the cross-covariance matrix $\Sigma_{YX}$ being at most $k$ (of interest here is the setting in which $k < \min\{p,q\}$). In particular, the sufficient dimension reduction $f(X)$ in such cases is given by the $k$-dimensional row-space of the matrix $\Sigma_{YX} \cdot \Sigma_X^{-1}$, which is the mapping that specifies the best linear estimator of $Y$ based on $X$. On the other hand, the conditional statistics of the responses given the covariates are specified by the submatrix $[\Sigma^{-1}]_{Y}$ of the joint precision matrix $\Sigma^{-1}$. As a result, the problem of fitting a concisely parameterized model to the conditional distribution of the responses given the covariates is more conveniently specified in terms of structured approximations of a submatrix of the precision matrix $\Sigma^{-1}$; for example, fitting a graphical model to the conditional distribution of the responses given the covariates corresponds to approximating the submatrix $[\Sigma^{-1}]_{Y}$ by a sparse matrix. The different parameterizations in which these two modeling tasks are most naturally described -- SDR in terms of covariance matrices and conditional modeling of the responses given the covariates in terms of precision matrices -- poses an obstruction to their integration into a single framework.

We overcome this difficulty by making the observation that $\mathrm{rank}(\Sigma_{YX}) = \mathrm{rank}([\Sigma^{-1}]_{YX})$ in non-degenerate models (i.e., the joint covariance matrix $\Sigma$ is positive definite) based on the following relation between the two alternative forms -- in terms of the precision matrix and in terms of the covariance matrix -- of the mapping that specified the best linear estimate of $Y$ based on $X$:
\begin{equation}
\Sigma_{YX} \cdot \Sigma_{X}^{-1} = -[\Sigma^{-1}]_{Y} \cdot [\Sigma^{-1}]_{YX}
\label{eqn:relation}
\end{equation}
Hence, our approach for integrating SDR and conditional modeling of the responses given the covariates is to fit a Gaussian model specified by a precision matrix $\Theta = \begin{pmatrix} \Theta_Y & \Theta_{YX} \\ \Theta_{YX}' & \Theta_{X} \end{pmatrix}$ to a collection of observations of covariates and responses such that $(a)$ the submatrix $\Theta_{YX}$ is low-rank, and $(b)$ the submatrix $\Theta_{Y}$ (which specifies the conditional distribution of the responses given the covariates) has a concise parameterization that describes a structured model. This reformulation of our modeling framework in terms of precision matrices leads naturally to the following general form of an estimator, given joint observations $\{Y^{(i)}, X^{(i)}\}_{i = 1}^n \subset \R^{p+q}$ of the responses and covariates:
\begin{equation}
\begin{aligned}
\argmin_{\Theta \in \Sp^{p+q}, ~\Theta \succ 0} & -\ell(\Theta; \{Y^{(i)},X^{(i)}\}_{i=1}^n) + \lambda_n [\gamma \|\Theta_{YX}\|_\star + R(\Theta_Y)]
\end{aligned}
\label{eqn:Generic}
\end{equation}
The set $\Sp^{k}$ denotes the space of $k \times k$ symmetric matrices, the function $\ell(\Theta; \allowbreak \{X^{(i)}, \allowbreak Y^{(i)}\}_{i=1}^n)$ denotes the log-likelihood of the observations $\{Y^{(i)},X^{(i)}\}_{i=1}^n$ with respect to a Gaussian distribution parameterized by the precision matrix $\Theta$, the function $\|\cdot\|_\star$ denotes the nuclear norm (sum of the singular values of a matrix), and the function $R: \Sp^p \to \R$ is suitably chosen to promote a desired structure in the submatrix $\Theta_Y$ (i.e., the conditional distribution of the responses given the covariates). The role of the nuclear norm penalty is to promote low rank structure in the submatrix $\Theta_{YX}$; this regularizer has been successfully employed in many settings for fitting structured low-rank models to high-dimensional data \citep{Faz2002,RecFP2009}. Here $\lambda_n > 0$ and $\gamma > 0$ are regularization parameters. By virtue of the convexity of norms and of the negative of the log-likelihood function $\ell(\Theta;\{Y^{(i)},X^{(i)}\}_{i=1}^n)$, the program \eqref{eqn:Generic} is a convex optimization problem if the function $R(\Theta_Y)$ is convex. In the following discussion, we provide three concrete approaches to fit structured low-dimensional models to the conditional distribution of $Y|f(X)$ via suitable choices of the function $R(\Theta_Y)$. \\

\par{\bf SDR + Factor Modeling (SDR-FM)} Fitting a factor model to the conditional distribution of $Y | f(X)$ corresponds to approximating the covariance matrix of $Y | f(X)$ as the sum of a diagonal matrix and a low-rank matrix. By appealing to the Sherman-Morrison-Woodbury formula \citep{Horn1990}, the covariance matrix of $Y | f(X)$ being decomposable as the sum of a diagonal matrix and a low-rank matrix is equivalent to the precision matrix of $Y | f(X)$ being decomposable as the \emph{difference} between a diagonal matrix and a low-rank matrix. Thus, a natural choice for the regularizer $R(\Theta_Y)$ is:
\begin{eqnarray}
\label{eqn:FMReg}
\begin{aligned}
R(\Theta_Y) = \inf_{D_Y,L_Y \in \Sp^p} & \mathrm{trace}(L_Y) & \text{subject to} \hspace{.1in}& \Theta_Y = D_Y - L_Y, ~ L_Y \succeq 0, ~ D_Y ~\mathrm{is~diagonal.}
\end{aligned}
\end{eqnarray}
Here $D_Y,\Theta_Y$ represent the diagonal and low-rank components of $\Theta_Y$. As before, the role of the nuclear norm penalty (the nuclear norm for positive semidefinite matrices reduces to the trace) is to enforce low-rank structure in the $L_{Y}$ component.\\

\par{\bf SDR + Graphical Modeling (SDR-GM)} Fitting a sparse graphical model to the conditional distribution of $Y | f(X)$ corresponds to approximating the submatrix $\Theta_Y$ by a sparse matrix; the sparsity pattern of $\Theta_Y$ specifies the graphical model structure underlying the conditional distribution of $Y | f(X)$. Based on prior work on the Graphical Lasso \citep{YuanLin2006,Friedman2008}, an appropriate choice for the regularizer $R(\Theta_Y)$ is:
\begin{eqnarray}
\begin{aligned}
R(\Theta_Y) = \|\Theta_Y\|_{\ell_1}.
\end{aligned}
\label{eqn:GMReg}
\end{eqnarray}
The function $\|\cdot\|_{\ell_1}$ denotes the $\ell_1$ norm (sum of the magnitudes of the entries of a matrix), and its role is to induce sparsity in the submatrix $\Theta_{Y}$. Regularizers based on the $\ell_1$ norm have been widely and successfully employed in many settings for fitting sparse models to data in high dimensions \citep{Tib1996,Chen,CanRT2006,Donb2006}. The convex program \eqref{eqn:Generic} with $R(\Theta_Y) = \|\Theta_Y\|_{\ell_1}$ is a natural extension of the Graphical Lasso \citep{YuanLin2006,Friedman2008} in which an $\ell_1$ norm penalty is employed to induce sparsity in the precision matrix (although the Graphical Lasso operates purely on observations of a set of responses and it does not consist of a nuclear norm penalty corresponding to an SDR objective). \\

\par{\bf SDR + Latent-Variable Graphical Modeling (SDR-LVGM)} We also describe a generalization of the SDR-GM approach. In practice, it may be expensive or infeasible for a data analyst to gather observations of all the relevant covariates that may potentially impact the responses. As a result, the responses $Y$ could be affected by unobserved latent variables in addition to being influenced by the covariates $X$. These latent variables can lead to confounding dependencies among the responses, which in turn complicates the task of fitting a graphical model to the conditional distribution of $Y | f(X)$. Motivated by these considerations, it is of interest to fit a graphical model to the distribution of $Y | f(X), \zeta$, where $\zeta \in \R^h$ (here $h \ll p$) represent a small number of latent variables that are statistically independent of $f(X)$. As described by \citet{Chand2012}, fitting such a latent-variable graphical model corresponds to approximating the submatrix $\Theta_{Y}$ by the sum of a sparse matrix and a low-rank matrix, rather than just a sparse matrix as in the case of a pure graphical model approximation. Here the low-rank component of $\Theta_{Y}$ accounts for the effect of the latent variables $\zeta \in \R^h$ on the responses $Y$, and the rank of this component is equal to the dimension $h$ of $\zeta$. The sparse component of $\Theta_{Y}$ specifies the graphical model structure underlying $Y | f(X), \zeta$. Building on the insights in \citet{Chand2012}, the regularizer $R_{\delta}(\Theta_Y)$ in this setting chosen as:
\begin{eqnarray}
\begin{aligned}
R_{\delta}(\Theta_Y) = \inf_{S_Y,L_Y \in \Sp^p} & \mathrm{trace}(L_Y) + \delta \|S_Y\|_{\ell_1} & \text{subject to} && \Theta_Y = S_Y - L_Y, ~ L_Y \succeq 0.
\end{aligned}
\label{eqn:LVGMReg}
\end{eqnarray}
The matrices $S_Y, L_Y \in \Sp^p$ in \eqref{eqn:LVGMReg} correspond to the sparse and low-rank components of $\Theta_{Y}$, respectively, and $\delta > 0$ is a regularization parameter. As before, the $\ell_1$ norm and the nuclear norm penalties promote the type of structure that we desire in our model. Plugging in the regularizer $R_{\delta}(\Theta_Y)$ into \eqref{eqn:Generic}, we obtain the following optimization program for joint SDR and latent-variable graphical modeling given observations $\{Y^{(i)}, X^{(i)}\}$ of the responses and the covariates:
\begin{eqnarray} \label{eqn:originalproblem_NS}
(\hat{\Theta}, \hat{S}_Y, \hat{L}_Y) = \argmin_{\substack{\Theta \in \Sp^{p+q}, ~\Theta \succ 0 \\ S_Y,L_Y \in \Sp^p}} & -\ell(\Theta; \{Y^{(i)},X^{(i)}\}_{i=1}^n) + \lambda_n [\gamma\|\Theta_{YX}\|_{\star} + \mathrm{trace}(L_Y) + \delta \|S_Y\|_{\ell_1}] \nonumber \\ \mathrm{s.t.} & \Theta_{Y} = S_Y - L_Y, ~ L_Y \succeq 0. &
\end{eqnarray}
The convex program \eqref{eqn:originalproblem_NS} is a natural extension of the estimator proposed by \citet{Chand2012} in which a trace penalty and an $\ell_1$ norm penalty are employed to fit a latent-variable graphical model to a set of responses. To be clear, the estimator \eqref{eqn:originalproblem_NS} also incorporates SDR to identify a low-dimensional projection of a set of covariates that are sufficient for predicting the responses, which is in contrast to the estimator proposed by \citet{Chand2012}.

{\par}The regularizers defined in \eqref{eqn:FMReg}, \eqref{eqn:GMReg}, and \eqref{eqn:LVGMReg} are convex with respect to the submatrix $\Theta_Y$. As a result, the estimators corresponding the the SDR-FM, SDR-GM, and SDR-LVGM approaches are convex optimization programs. The estimator \eqref{eqn:originalproblem_NS} corresponding to the SDR-LVGM approach is in some sense a generalization of the estimators corresponding to the SDR-GM and SDR-FM approaches, as latent-variable graphical modeling may be viewed as a blend of factor modeling and graphical modeling. As a result, we focus in Section~\ref{section:consist} on analyzing the consistency properties of the estimator \eqref{eqn:originalproblem_NS} in a high-dimensional scaling regime. Specifically, suppose we observe samples $\{Y^{(i)},X^{(i)}\}_{i=1}^n \allowbreak \subset \R^{p+q}$ of a collection of jointly Gaussian responses and covariates $(Y,X) \in \R^{p+q}$, with population precision matrix $\Theta^\star = \begin{pmatrix} S_Y^\star - L_Y^\star & \Theta_{YX}^\star \\ {\Theta_{YX}^\star}' & \Theta_{X}^\star \end{pmatrix}\in \Sp^{p+q}$. Supplying these observations as input to the convex program \eqref{eqn:originalproblem_NS} and obtaining estimates $(\hat{\Theta},\hat{S}_{Y},\hat{L}_{Y}) \in \Sp^{p+q} \times \Sp^p \times \Sp^p$, we prove in Theorem~\ref{theorem:main} that (under certain conditions on $\Theta^\star$ and with high probability) the rank of the submatrix $\hat{\Theta}_{YX}$ is equal to the rank of $\Theta_{YX}^\star$, the rank of $\hat{L}_Y$ is equal to the rank of $L_Y^\star$, and the sparsity pattern of $\hat{S}_Y$ is the same as that of $S_Y^\star$; these recovery guarantees imply that we obtain the correct dimension of the image $f(X)$ specifying the sufficient dimension reduction of the covariates, the correct number of unobserved latent variables, and the correct conditional graphical model structure underlying the population. Informally, the assumptions on the population precision matrix $\Theta^\star$ are that: $(a)$ the submatrix $\Theta^\star_{YX}$ is sufficiently low-rank; $(b)$ the submatrix $\Theta^\star_{Y} = S_Y^\star - L_Y^\star$ is such that $S_Y^\star$ is sufficiently sparse and $L_Y^\star$ is sufficiently low-rank; and $(c)$ the population Fisher information ${\Theta^\star}^{-1} \otimes {\Theta^\star}^{-1}$ obeys certain irrepresentability-type conditions; see Assumptions 1 and 2 in Section~\ref{section:Fishercond}, Theorem~\ref{theorem:main} and the subsequent discussion in Appendix~\ref{sec:Parameterchoices} for a precise formulation of these conditions. The first assumption above on $\Theta^\star$ states that there exists a low-dimensional linear projection $f(X)$ of the covariates $X$ that is sufficient with respect to $Y$. The second condition requires that $Y | f(X)$ is specified by a latent-variable graphical model with a small number of latent variables and a sparse graphical model. The third assumption is analogous to the irrepresentability conditions that play a role in the analysis of the consistency of the Lasso \citep{MeiB2006,ZhaY2006,Wai2014}, the Graphical Lasso \citep{RavWRY2008}, the convex relaxation proposed by Chandrasekaran et al. (2012) for latent-variable graphical modeling, as well as other estimators in high-dimensional inference problems \citep{BuhV2011,Wai2014}.

In Section~\ref{section:experiments} we illustrate the performance of the estimators corresponding to the SDR-FM, SDR-GM, and SDR-LVGM approaches on two datasets. First, we consider a financial asset modeling problem in which the responses are a collection of stock returns of $67$ companies from the Standard and Poor index, and the covariates are the following $7$ macroeconomic indicators: the industrial production index, the inflation rate, the amount of oil exports, the population growth rate, the unemployment rate, the consumer credit score, and the EUR to USD exchange rate. In the second experiment, we analyze the 20newsgroup dataset that consists of $16,242$ samples in $\mathbb{R}^{100}$, with each observation corresponding to a news document. The coordinates of these observations are indexed by a collection of $100$ words, and each observation is a binary vector specifying whether a word appears in the document. Of those $100$ words, the following $9$ words are chosen to be the covariates as they appear to be useful in categorizing newsgroup documents: government, religion, science, technology, war, medicine, world, food, and games. The remaining $91$ words are used as response variables.\\

\par{\bf Adapting to Alternative Forms of SDR} In many settings, one is interested in identifying a \emph{subset} of the covariates $X$ that is useful for predicting the responses $Y$, rather than a generic dimension reduction of $X$. In such cases, it is more natural to fit a linear Gaussian model with a precision matrix $\Theta \in \mathbb{S}^{p+q}$ in which the submatrix $\Theta_{YX}$ is column-sparse (i.e, only a subset of the columns of this matrix are nonzero) instead of being low-rank. This point follows from the relation \eqref{eqn:relation} by noting that the map $-[\Sigma^{-1}]_{Y} \cdot [\Sigma^{-1}]_{YX}$ specifying the best linear estimator of $Y$ based on $X$ must be column-sparse if a subset of the covariates is sufficient for predicting the responses (the indices of the nonzero columns of $\Theta_{YX}$ correspond to the subset of the covariates that are relevant for predicting the responses). To fit a model with a column-sparse submatrix $\Theta_{YX}$, one can modify the family of estimators \eqref{eqn:Generic} by replacing the nuclear norm penalty $\|\Theta_{YX}\|_\star$ with a group norm penalty $\|\Theta_{YX}\|_{2,1} = \sum_{i=1}^q \|(\Theta_{YX})_{:,i}\|_{\ell_2}$, where $(\Theta_{YX})_{:,i}$ represents the $i$'th column of $\Theta_{YX}$. Such group norm penalties are useful for inducing sparsity in entire columns of $\Theta_{YX}$ \citep{YuanLin2006}. As an illustration, given observations $\{Y^{(i)},X^{(i)}\}_{i=1}^n \subset \R^{p+q}$ the estimator \eqref{eqn:originalproblem_NS} can be modified as follows:
\begin{eqnarray} \label{eqn:CovariateSelection}
(\hat{\Theta}, \hat{S}_Y, \hat{L}_Y) = \argmin_{\substack{\Theta \in \Sp^{p+q}, ~\Theta \succ 0 \\ S_Y,L_Y \in \Sp^p}} & \hspace{-.07in}-\ell(\Theta; \{Y^{(i)},X^{(i)}\}_{i=1}^n) + \lambda_n [\gamma\|\Theta_{YX}\|_{2,1} + \mathrm{trace}(L_Y) + \delta\|S_Y\|_{\ell_1}] \nonumber \\ \mathrm{s.t.} & \Theta_{Y} = S_Y - L_Y, ~ L_Y \succeq 0. &
\end{eqnarray}
This estimator simultaneously identifies a subset of the covariates that are relevant for predicting the responses and also fits a latent-variable graphical model to the conditional distribution of the responses given the covariates. The analysis of the statistical consistency of this estimator is similar in spirit to that of the estimator \eqref{eqn:originalproblem_NS}; see Appendix~\ref{section:Covariate} for more details.

\subsection{Related Work}

Many researchers have developed techniques for computing sufficient dimension reductions (see the survey \citep{AdragniCook2009} and the references therein), with Sliced Inverse Regression being a prominent example \citep{LiSDR1991}. In a jointly Gaussian setting, classical approaches such as Canonical Correlations Analysis (CCA) or Partial Least Squares (PLS) may also be employed to compute linear sufficient dimension reductions \citep{Fung2002}, although the objectives of CCA and PLS are somewhat different than that of SDR. More recently, \citet{NegWain2011} employed a nuclear norm penalty in a multivariate linear regression setup to identify a low-dimensional projection of a set of covariates that best predicts a set of responses. However, none of these papers consider the additional challenge of modeling the conditional distribution of the responses given the covariates.

A number of researchers have developed methods for simultaneously obtaining a concise model of the predictive relationship of the covariates on the responses, while also fitting a sparse graphical model to the conditional distribution of the responses given the covariates. For example, the techniques introduced by \citet{RotBLZ2008}, \citet{Cai2010}, and \citet{Yin2011} may be interpreted as seeking a sparse approximation to the matrix $\Sigma_{YX} \cdot \Sigma_{X}^{-1}$ (or equivalently, $-\Theta_Y^{-1} \cdot \Theta_{YX}$) that specifies the best linear estimator of $Y$ based on $X$, in addition to computing a sparse approximation to the submatrix $\Theta_Y$. The algorithms proposed in these papers are either non-convex or involve multi-step procedures consisting of several convex programs. In a different direction, \citet{SohnKim2012} and \citet{Tong2014} consider a regularized log-likelihood convex program with $\ell_1$ norm penalties on the submatrices $\Theta_Y$ and $\Theta_{YX}$ of the joint precision matrix. In contrast to these papers, our approach for modeling the predictive relationship of the covariates on the responses is based on SDR, where we seek a low-dimensional projection of the covariates that is sufficient for predicting the responses. Further, our framework integrates SDR and conditional modeling of the responses given the covariates via a single convex program.

In the same work referenced above, \citet{Tong2014} also consider the problem of selecting a subset of a collection of covariates that are most relevant for predicting a set of responses, while additionally fitting a sparse graphical model to the conditional distribution of the responses given the selected covariates. They address this problem by proposing the following regularized log-likelihood convex program consisting of an $\|\cdot\|_{2,1}$ norm penalty on $\Theta_{YX}$ and an $\ell_1$ norm penalty on $\Theta_Y$:
\begin{eqnarray} \label{eqn:CovariateSelectionGM}
\hat{\Theta} = \argmin_{\substack{\Theta \in \Sp^{p+q}, ~\Theta \succ 0}} -\ell(\Theta; \{Y^{(i)},X^{(i)}\}_{i=1}^n) + \lambda_n [\gamma\|\Theta_{YX}\|_{2,1} + \|\Theta_Y\|_{\ell_1}]
\end{eqnarray}
In Appendix ~\ref{section:Covariate}, we analyze the high-dimensional consistency of the estimator \eqref{eqn:CovariateSelection} -- which fits a latent-variable graphical model to the conditional distribution of the responses given the selected covariates rather than just a graphical model -- and this analysis can be specialized to obtain consistency results for the estimator \eqref{eqn:CovariateSelectionGM}. More broadly, one of the key distinctions between our work and that of \citet{Tong2014} is that our approach \eqref{eqn:Generic} for simultaneous SDR and conditional modeling of responses given covariates is useful for obtaining general linear sufficient dimension reductions of the covariates rather than just selecting subsets of relevant covariates. Moreover, our framework can be adapted to fit three types of models to the conditional distribution of the responses conditioned on the covariates.

\subsection{Notation}\label{section:notation} Given a matrix $U \in \mathbb{R}^{p_1 \times p_2}$, the norm $\|U\|_{\ell_\infty}$ denotes the largest entry in magnitude of $U$, and the norm $\|U\|_2$ denotes the spectral norm (the largest singular value of $U$). The norm $\|U\|_{2,\infty}$ denotes the maximum of the $\ell_2$ norms of the columns of $U$: $\|U\|_{2,\infty} = \max_{i = 1,2,\dots p_2} \|U_{:,i}\|_{\ell_2}$ . We denote the set of $k \times k$ positive semidefinite and positive-definite matrices by $\Sp^k_{+}$ and $\Sp^k_{++}$, respectively. Finally, the linear operators $\mathcal{A}: \Sp^p \times \Sp^p \times \mathbb{R}^{p{\times}q} \times \Sp^q \rightarrow \Sp^{(p+q)}$ and its adjoint $\mathcal{A}^{\dagger}: \Sp^{(p+q)} \rightarrow \Sp^p \times \Sp^p \times \mathbb{R}^{p{\times}q} \times \Sp^q$ are defined as follows:
\vspace{.1in}
\begin{equation}
\mathcal{A}(M, N, K, O) \triangleq \left( \begin{array}{cc}
M - N & K \\
K^T & O \end{array} \right), \qquad \mathcal{A}^{\dagger}\left( \begin{array}{cc}
Q & K \\
K^{T} & O \end{array} \right)\triangleq (Q,Q,K,O)
\label{eqn:OperatorDefs}
\end{equation}
\section{Model Selection Consistency}
\label{section:consist}

As described in the introduction, we investigate the consistency properties of the SDR-LVGM estimator \eqref{eqn:originalproblem_NS}, which integrates SDR and latent-variable graphical modeling of the conditional distribution of the responses given the covariates. The main result is stated in Section~\ref{sec:Result}, and it is based on assumptions on the population precision matrix (discussed in Section~\ref{sec:techsetup}) and irrepresentability-type conditions on the population Fisher information (discussed in Section~\ref{section:Fishercond}).

%, and we discuss appropriate choices of the tradeoff parameters $\gamma,\delta$ (from the estimator \eqref{eqn:originalproblem_NS}) in Section~\ref{sec:Parameterchoices}

% Finally, in Section~\ref{section:Covariate} we discuss the consistency of the estimator that is obtained by replacing the nuclear norm penalty $\|\Theta_{YX}\|_\star$ by the group-norm penalty $\|\Theta_{YX}\|_{1,2}$; as described at the end of Section~\ref{section:contribution}, such an estimator identifies a subset of covariates -- rather than a general low-dimensional linear projection -- that best predict the responses.

\subsection{Technical Setup}
\label{sec:techsetup}
Let $\Theta^\star = \begin{pmatrix} \Theta_Y^\star & \Theta_{YX}^\star \\ {\Theta_{YX}^\star}' & \Theta_{X}^\star \end{pmatrix} \in \mathbb{S}^{p+q}_{++}$ denote the precision matrix of a jointly Gaussian random vector $(Y,X) \in \R^{p+q}$ of covariates and responses, and let $\Sigma^\star = {\Theta^\star}^{-1}$ denote the corresponding covariance matrix. The submatrix $\Theta_{YX}^\star \in \R^{p \times q}$ is a rank-$k$ matrix, where $k \ll \min\{p,q\}$ is the size of the smallest dimension reduction $f(X)$ of the covariates $X$ that is sufficient with respect to the responses $Y$. The submatrix $\Theta^\star_Y \in \mathbb{S}^p$ is the precision matrix of the conditional distribution of $Y|f(X)$, and it specifies a latent-variable graphical model \citep{Chand2012}. That is, the matrix $\Theta^\star_Y$ is decomposed as $\Theta^\star_Y = S_Y^\star - L_Y^\star$ -- the component $S_{Y}^\star$ is a sparse matrix representing the precision matrix of the distribution of $Y$ conditioned on $f(X)$ as well as a small number of additional unobserved latent variables $\zeta \in \R^h$ (here $h \ll p$), and the component $L_{Y}^\star$ is a low-rank matrix representing the effect of the latent variables $\zeta$ ($\mathrm{rank}(L_Y^\star) = h$). These structural attributes of our model lead to the following definition:

\begin{definition} An estimate $(\hat{\Theta}, \hat{S}_Y, \hat{L}_Y) \in \mathbb{S}^{q+p}_{++} \times \mathbb{S}^p \times \mathbb{S}^p_{+}$ with $\hat{\Theta}_Y = \hat{S}_Y - \hat{L}_Y$ is a \emph{structurally correct} estimate of the model specified by the matrices $(\Theta^\star, S_{Y}^\star, L_{Y}^\star) \in \mathbb{S}^{q+p}_{++} \times \mathbb{S}^p \times \mathbb{S}^p_{+}$ with $\Theta^\star_Y = S^\star_Y - L^\star_Y$ if $(1)$ $\mathrm{rank}(\hat{\Theta}_{YX}) = \mathrm{rank}(\Theta_{YX}^\star)$ and $(2)$ $\mathrm{sign}(\hat{S}_Y) = \mathrm{sign}(S_{Y}^\star)$ $(here\hspace{.05in}\mathrm{sign}(0) = 0)$; $\mathrm{rank}(\hat{L}_Y) = \mathrm{rank}(L_{Y}^\star)$.
\label{definition:1}
\end{definition}
Condition $(1)$ ensures that the size of the smallest dimension reduction $f(X)$ of the covariates $X$ that is sufficient with respect to the responses $Y$ is estimated correctly. Condition $(2)$ ensures that the latent-variable graphical model specifying the conditional distribution of $Y | f(X)$ is estimated correctly, which corresponds to accurately identifying the two components composing the precision matrix $\Theta^\star_Y$ of $Y | f(X)$. In particular, this condition ensures that $(a)$ $\hat{S}_Y$ provides a structurally correct estimate of the graphical model specifying the conditional distribution of $Y | f(X), \zeta$; that is, positive, negative and zero entries in ${S}_Y^\star$ are estimated correctly as positive, negative, and zero entries, respectively, in $\hat{S}_Y$, and $(b)$ the dimension of the latent variables $\zeta$ is estimated correctly.

Following the literature on high-dimensional estimation (see the surveys \citep{BuhV2011,Wai2014} and the references therein), a natural set of conditions for obtaining consistent and structurally correct parameter estimates is to assume that the curvature of the likelihood function at $\Theta^\star$ is bounded in certain directions. This curvature is governed by the Fisher information at $\Theta^\star$:
\begin{eqnarray*}
\mathbb{I}({\Theta}^\star) \triangleq {{{{\Theta}}}^\star}^{-1} \otimes {{\Theta}^\star}^{-1} = \Sigma^\star \otimes \Sigma^\star.
\end{eqnarray*}
Here $ \otimes$ denotes a tensor product between matrices and $\mathbb{I}({\Theta}^\star)$ may be viewed as a map from $\mathbb{S}^{(p+q)}$ to $\mathbb{S}^{(p+q)}$. We impose conditions requiring that $\mathbb{I}({\Theta}^\star)$ is well-behaved when applied to matrices of the form $\Theta - \Theta^\star = \begin{pmatrix} (S_Y-S_Y^\star)-(L_Y-L_Y^\star) & \Theta_{YX} - \Theta_{YX}^\star \\ {\Theta_{YX}}'-{\Theta_{YX}^\star}' & \Theta_X-\Theta_{X}^\star \end{pmatrix}$, where $S_Y$ is in a neighborhood of $S_Y^\star$ restricted to the set of sparse matrices and $(L_Y,\Theta_{YX})$ are in a neighborhood of $(L_Y^\star,\Theta_{YX}^\star)$ restricted to sets of low-rank matrices. As formally described in Section~\ref{section:Fishercond}, these local properties of $I(\Theta^\star)$ around $\Theta^\star$ are conveniently stated in terms of \emph{tangent spaces} to the algebraic varieties of sparse matrices and of low-rank matrices.

Let $M \in \mathbb{S}^{p}$ be a symmetric matrix with $k$ nonzero entries. The tangent space at $M$ with respect to the algebraic variety of matrices in $\mathbb{S}^p$ with at most $k$ nonzero entries is given by:
\begin{eqnarray*}
\Omega(M) \triangleq \{J \in \mathbb{S}^{p}| \hspace{.05in}\text{support}(J) \subseteq \text{support}(M)\}.
\end{eqnarray*}
Here `support' denotes the set of locations of the nonzero entries. Next, consider a rank-$r$ matrix $N \in \mathbb{R}^{p_1\times{p_2}}$ with reduced singular value decomposition (SVD) given by $N = U D V'$, where $U \in \mathbb{R}^{p_1{\times}r}$, $D \in \mathbb{R}^{r{\times}r}$, and $V \in \mathbb{R}^{p_2{\times}r}$. The tangent space at $N$ with respect to the algebraic variety of ${p_1{\times}p_2}$ matrices with rank less than or equal to $r$ is given by\footnote{We also consider the tangent space at a symmetric low-rank matrix with respect to the algebraic variety of symmetric low-rank matrices. We use the same notation `$T$' to denote tangent spaces in both the symmetric and non-symmetric cases, and the appropriate tangent space is clear from the context.}:
\begin{eqnarray*}
{T}(N) \triangleq \{UY_1' + Y_2 V'| Y_1 \in \mathbb{R}^{p_2 \times r}, Y_2 \in \mathbb{R}^{p_1 \times r}\}.
\end{eqnarray*}
In the next section, we describe irrepresentability conditions on the population Fisher information $\mathbb{I}(\Theta^\star)$ in terms of the tangent spaces $\Omega(S_{Y}^\star)$, $T(L_{Y}^\star)$, and $T(\Theta_{YX}^\star)$; under these conditions, we prove in Appendix~\ref{section:proof} that the regularized maximum-likelihood convex program \eqref{eqn:originalproblem_NS} provides structurally correct and consistent estimates.

\subsection{Fisher Information Conditions For Consistency}
\label{section:Fishercond}
Let $\mathbb{I}^\star = \mathbb{I}(\Theta^\star)$ denote the population Fisher information at $\Theta^\star$. Given a norm $\|.\|_\Upsilon$ on $\mathbb{S}^p \times \mathbb{S}^p \times \mathbb{R}^{p \times q} \times \mathbb{S}^{q}$, the first condition we consider is to bound the minimum gain of $\mathbb{I}^\star$ restricted to a subspace $\mathbb{H} \subset \mathbb{S}^p \times \mathbb{S}^p \times \mathbb{R}^{p \times q} \times \mathbb{S}^{q}$ as follows:
\begin{eqnarray}
\chi({\mathbb{H}}, \|.\|_{\Upsilon}) \triangleq \min_{\substack{(S_{Y}, L_{Y}, \Theta_{YX}, \Theta_{X}) \in {\mathbb{H}}\\ \|(S_{Y}, L_{Y}, \Theta_{YX}, \Theta_{X})\|_\Upsilon= 1}} \|\mathcal{P}_{\mathbb{H}} \A^{\dagger} \mathbb{I}^\star \A \mathcal{P}_{\mathbb{H}}(S_{Y}, L_{Y}, \Theta_{YX}, \Theta_{X})\|_{\Upsilon},
\label{eqn:Chieq}
\end{eqnarray}
where $\mathcal{P}_{\mathbb{H}}$ denotes the projection operator onto the subspace $\mathbb{H}$ and the linear maps $\A$ and $\A^{\dagger}$ are defined in \eqref{eqn:OperatorDefs}. The quantity $\chi({\mathbb{H}}, \|.\|_{\Upsilon})$ being large ensures that the Fisher information $\mathbb{I}^\star$ is well-conditioned restricted to image $\A \mathbb{H} \subseteq \mathbb{S}^{p+q}$. The second condition that we impose on $\mathbb{I}^\star$ is in the spirit of irrepresentibility-type conditions \citep{MeiB2006,ZhaY2006,Wai2009,RavWRY2008,Chand2012} that are frequently employed in high-dimensional estimation. Specifically, we require that the inner-product between elements in $\A \mathbb{H}$ and $\A \mathbb{H}^\perp$, as quantified by the metric induced by $\mathbb{I}^\star$, is bounded above:
\begin{equation}
\begin{aligned}
\varphi(\mathbb{H}, \|.\|_{\Upsilon}) \triangleq \max_{\substack{Z \in \mathbb{H}\\ \|Z\|_{\Upsilon} = 1}} \|\mathcal{P}_{{{\mathbb{H}}^{\perp}}}\A^{\dagger} \mathbb{I}^{\star}\A\mathcal{P}_{{{\mathbb{H}}}} (\mathcal{P}_{{{\mathbb{H}}}} \A^{\dagger}\mathbb{I}^{\star}\A\mathcal{P}_{{{\mathbb{H}}}})^{-1}(Z)\|_{\Upsilon}.
\end{aligned}
\label{eqn:varphidef}
\end{equation}
The operator $(\mathcal{P}_{{{\mathbb{H}}}} \A^{\dagger}\mathbb{I}^{\star}\A\mathcal{P}_{{{\mathbb{H}}}})^{-1}$ in \eqref{eqn:varphidef} is well-defined if $\chi(\mathbb{H},\|\cdot\|_\Upsilon) > 0$, since this latter condition implies that $\mathbb{I}^\star$ is injective restricted to $\A \mathbb{H}$. The quantity $\varphi(\mathbb{H}, \|.\|_{\Upsilon})$ being small implies that any element of $\A \mathbb{H}$ and any element of $\A {\mathbb{H}}^{\perp}$ have a small inner-product (in the metric induced by $\mathbb{I}^\star$).

A natural approach to controlling the conditioning of the Fisher information around $\Theta^\star$ is to bound the quantities $\chi(\mathbb{H}^\star,\|\cdot\|_\Upsilon)$ and $\varphi(\mathbb{H}, \|.\|_{\Upsilon})$ for $\mathbb{H}^\star = \Omega(S_Y^\star) \times T(L_Y^\star) \times T(\Theta_{YX}^\star) \times \Sp^q$. However, a complication that arises with this approach is that the varieties of low-rank matrices are locally curved around $L_Y^\star$ and around $\Theta_{YX}^\star$. Consequently, the tangent spaces at points in neighborhoods around $L_Y^\star$ and around $\Theta_{YX}^\star$ are not the same as $T(L_Y^\star)$ and $T(\Theta_{YX}^\star)$. (A similar difficulty does not arise with sparse matrices, as the variety of sparse matrices is locally flat around $S_Y^\star$; hence, the tangent spaces at all points in a neighborhood of $S_Y^\star$ are the same.) In order to account for this curvature underlying the varieties of low-rank matrices, we bound the distance between nearby tangent spaces via the following induced norm:
\begin{eqnarray*}
\rho(T_1,T_2) \triangleq \max_{\|N\|_2 \leq 1} \|(\mathcal{P}_{T_1} - \mathcal{P}_{T_2})(N)\|_2.
\label{eqn:distortion}
\end{eqnarray*}
Using this approach for bounding nearby tangent spaces, we consider subspaces $\mathbb{H}' = \Omega(S_Y^\star) \times T'_Y \times T'_{YX} \times \mathbb{S}^{q}$ for all $T_Y'$ close to $T(L_Y^\star)$ and for all $T_{YX}'$ close to $T(\Theta_{YX}^\star)$, as measured by $\rho$ \citep{Chand2012}. For $\omega_{Y}, \omega_{YX} \in (0,1)$, we bound $\chi({\mathbb{H}}', \|.\|_{\Upsilon})$ and $\varphi({\mathbb{H}}', \|.\|_{\Upsilon})$ in the sequel for all subspaces $\mathbb{H}'$ in the following set:
\begin{equation}
\begin{aligned}
U{(\omega_{Y},\omega_{YX})} \triangleq \Big\{\Omega(S_Y^\star) \times T'_Y \times T'_{YX} \times \mathbb{S}^{q} ~|~ &\rho({{T_{Y}'}}, T({L_{Y}^\star})) \leq \omega_{Y} \\ & \rho({{T_{YX}'}}, T({\Theta_{YX}^\star})) \leq \omega_{YX}\Big\}.
\end{aligned}
\label{eqn:Udef}
\end{equation}

We control the quantities $\chi({\mathbb{H}'}, \|.\|_{\Upsilon})$ and $\varphi({\mathbb{H}}', \|.\|_{\Upsilon})$ using a slight variant of the dual norm of the regularizer $\delta \|S_Y\|_{\ell_1} + \mathrm{trace}(L_Y) + \gamma \|\Theta_{X,Y}\|_\star$ in \eqref{eqn:originalproblem_NS}:
\begin{equation}
\Phi_{\delta,\gamma}(S_Y, L_Y, \Theta_{YX}, \Theta_{X}) \triangleq \max\left\{\frac{\|S_Y\|_{\ell_\infty}}{\delta}, \|L_Y\|_{2}, \frac{\|\Theta_{YX}\|_{2}}{\gamma}, \|\Theta_{X}\|_2 \right\}.
\label{eqn:Phidef}
\end{equation}
As the dual norm $\max\left\{\frac{\|S_Y\|_{\ell_\infty}}{\delta}, \|L_Y\|_{2}, \frac{\|\Theta_{YX}\|_{2}}{\gamma}\right\}$ of the regularizer in \eqref{eqn:originalproblem_NS} plays a central role in the optimality conditions of \eqref{eqn:originalproblem_NS}, controlling the quantities $\chi({\mathbb{H}'},\Phi_{\delta,\gamma})$ and $\varphi({\mathbb{H}}', \Phi_{\delta,\gamma})$ leads to a natural set of conditions that guarantee the structural correctness and consistency of the estimates produced by
\eqref{eqn:originalproblem_NS}. In summary, given a fixed set of parameters $(\delta,\gamma,\omega_{Y}, \omega_{YX}) \in \R_+ \times \R_+ \times (0,1) \times (0,1)$, we assume that $\mathbb{I}^\star$ satisfies the following conditions:
\begin{eqnarray}
\mathrm{Assumption~1}&:& \inf_{\mathbb{H}' \in U{(\omega_{Y}, \omega_{YX})}}\chi({\mathbb{H}}', {{\Phi}_{\delta,\gamma}}) \geq \alpha, ~~~ \mathrm{for~some~} \alpha > 0 \label{eqn:FirstFisherCond} \\[.1in]
\mathrm{Assumption~2}&:& \sup_{\mathbb{H}' \in U{(\omega_{Y}, \omega_{YX})}}\varphi({\mathbb{H}}', {{\Phi}_{\delta,\gamma}}) \leq 1-\nu ~~~ \mathrm{for~some~} \nu \in (0,1/3).
\label{eqn:SecondFisherCond}
\end{eqnarray}
For fixed $(\delta, \allowbreak \gamma, \allowbreak \omega_{Y}, \allowbreak \omega_{YX})$, larger values of $\alpha$ and $\nu$ in these assumptions lead to a better conditioned $\mathbb{I}^\star$.

Assumptions 1 and 2 are analogous to conditions that play an important role in the analysis of the Lasso for sparse linear regression \citep{MeiB2006,ZhaY2006,Wai2009}, graphical model selection via the Graphical Lasso \citep{RavWRY2008}, and in several other approaches for high-dimensional estimation \citep{BuhV2011,Wai2014}. As a point of comparison with respect to analyses of the Lasso, the role of the Fisher information $\mathbb{I}^\star$ in \eqref{eqn:FirstFisherCond} and in \eqref{eqn:SecondFisherCond} is played by $A^TA$, where $A$ is the underlying design matrix \citep{MeiB2006,ZhaY2006,Wai2009}. In analyses of both the Lasso and the Graphical Lasso in the papers referenced above, the analog of the subspace $\mathbb{H}$ is the set of models with support contained inside the support of the underlying sparse population model. Assumptions 1 and 2 are also similar in spirit to conditions employed in the analysis of convex relaxation methods for latent-variable graphical model selection \citep{Chand2012}. However, as emphasized previously in the introduction, an important distinction between the present paper and prior literature on graphical model selection is that the methods and results in previous work are not directly applicable to the problem of simultaneous SDR and (latent-variable) graphical modeling.

%For notational convenience in the statement of the main result in Section~\ref{sec:Result}, we define the set of parameters $\delta,\gamma$ that satisfy Assumptions 1 and 2 for fixed $\alpha >0, \nu \in (0,1/3], \omega_{Y} \in (0,1) , \omega_{YX} \in (0,1)$:
%\begin{equation}
%\begin{aligned}
%V(\alpha,\nu,\omega_Y,\omega_{YX}) \triangleq \Big\{(\delta, \gamma) ~\Big|~ & \inf_{\mathbb{H}' \in U{(\omega_{Y}, \omega_{YX})}}\chi({\mathbb{H}}', \|.\|_{{\Phi}_{\delta,\gamma}}) \geq \alpha, \\ &\sup_{\mathbb{H}' \in U{(\omega_{Y}, \omega_{YX})}}\varphi({\mathbb{H}}', \|.\|_{{\Phi}_{\delta,\gamma}}) \leq 1-\nu\Big\}
%\end{aligned}
%\end{equation}

\subsection{High-Dimensional Consistency Result}
\label{sec:Result}

In this section, we describe the performance of the regularized maximum-likelihood program \eqref{eqn:originalproblem_NS}. Before formally stating our main result, we introduce some notation. Recalling that \\ $\Theta^\star = \begin{pmatrix} S_Y^\star-L_Y^\star & \Theta_{YX}^\star \\ {\Theta_{YX}^\star}' & \Theta_{X}^\star \end{pmatrix}$ is the population precision matrix, let $\tau_{Y}$ denote the minimum nonzero entry in magnitude of ${S_{Y}^{\star}}$, let $\mathrm{deg}(S_{Y}^\star)$ denote the maximal number of nonzeros per row/column of $S_{Y}^\star$ (i.e., the degree of the graphical model underlying the conditional distribution $Y | f(X),\zeta$, which is specified by the precision matrix $S_Y^\star$), let $\sigma_{Y}$ denote the minimum nonzero singular value of ${L_{Y}^{\star}}$, and let $\sigma_{YX}$ denote the minimum nonzero singular value of ${\Theta_{YX}^{\star}}$.

\begin{theorem}
\label{theorem:main}
Suppose we are given i.i.d observations $\{Y^{(i)}, X^{(i)}\}_{i = 1}^n \subset \mathbb{R}^{p+q}$ of a collection of jointly Gaussian covariates/responses with population precision matrix $\Theta^\star \in \mathbb{S}^{p+q}_{++}$. Fix $\alpha > 0, \nu \in (0,1/3), \omega_{Y} \in (0,1), \omega_{YX} \in (0,1)$. Suppose the parameters $\delta$ and $\gamma$ are chosen such that the population Fisher information $\mathbb{I}^\star$ satisfies Assumptions 1 and 2.\\[.05in]

Let $m \triangleq \max\{\frac{1}{\delta}, 1, \frac{1}{\gamma}\}$, $\bar{m} \triangleq \max\{{\delta}, 1, {\gamma}\}$, $\beta \triangleq \frac{3-\nu}{\nu}$, and $\psi \triangleq \|(\Theta^\star)^{-1}\|_2$. Further, $C_1 = \frac{24}{\alpha} + \frac{1}{\psi^2}$, $C_2 = \frac{4}{\alpha} (\frac{1}{3\beta} + 1)$, $C_{\sigma_Y} = C_1^2\psi^2\max \{ 12\beta+1, \frac{2}{{C_2}\psi^2}+1\}$, $C_{\sigma_{YX}} = C_1^2\psi^2\max\{18\beta, \frac{2}{C_2\psi^2} + 6\beta\}$, $C_{samp} = \max\{\frac{1}{48\psi\beta},{48{\beta}{\psi^3}C_1^2}, 8\psi{C_2}, \frac{128{\psi^3}C_2}{\alpha}\}$, and $\lambda_{\text{upper}} = \frac{1}{m\bar{m}^2\mathrm{deg}(S_Y^\star)C_{samp}}$. Suppose that the following conditions hold:
\begin{enumerate}
\item $n \geq \frac{4608\psi^2\beta^2m^2(p+q)}{\lambda_{\text{upper}}^2}$; that is $n \gtrsim \Big[\frac{\beta^4}{\alpha^2} m^4\bar{m}^4\mathrm{deg}(S_{Y}^\star)^2\Big] (p+q) $
\item $\lambda_n \in \Big[\sqrt{\frac{4608\psi^2\beta^2m^2(p+q)}{n}}, \lambda_{\text{upper}}\Big]$; \hspace{.1in} e.g. $\lambda_n \sim \beta{m}\sqrt{\frac{p+q}{n}}$
\item $\tau_{Y} \geq 2C_1\delta\lambda_{n}$; \hspace{.1in} that is $\tau_{Y} \gtrsim \frac{\beta}{\alpha}\delta{m}\sqrt{\frac{p+q}{n}}$ \hspace{.05in} if \hspace{.05in} $\lambda_n \sim \beta{m}\sqrt{\frac{p+q}{n}}$
\item $\sigma_{Y} \geq \frac{m}{\omega_{Y}}C_{\sigma_{Y}} \lambda_n$; \hspace{.1in} that is $\sigma_{Y} \gtrsim \frac{\beta^2}{\alpha^2\omega_{Y}}{m}{}\sqrt{\frac{p+q}{n_{}}}$ \hspace{.05in} if \hspace{.05in} $\lambda_n \sim \beta{m}\sqrt{\frac{p+q}{n}}$
\item $\sigma_{YX} \geq \frac{m^2}{\omega_{YX}} C_{\sigma_{YX}}\gamma^2\lambda_n $; \hspace{.1in} that is $\sigma_{YX} \gtrsim \frac{\beta^2 \gamma^2}{\alpha^2\omega_{YX}}{m^2}{}\sqrt{\frac{p+q}{n}}$ \hspace{.05in} if \hspace{.05in} $\lambda_n \sim \beta{m}\sqrt{\frac{p+q}{n}}$
\end{enumerate}

Then with probability greater than $1-2\exp\{-\frac{n\lambda_n^2}{4608 \beta^2 m^2 \psi^2}\}$, the optimal solution $(\hat{\Theta},\hat{S}_Y,\hat{L}_Y)$ of \eqref{eqn:originalproblem_NS} with the observations $\{Y^{(i)}, X^{(i)}\}_{i = 1}^n$ satisfies the following properties:
\begin{enumerate}
\item sign($\hat{S}_Y$) = sign($S_Y^\star$), rank($\hat{L}_Y$) = rank(${L}_Y^\star$), rank($\hat{\Theta}_{YX}$) = rank(${{\Theta}^\star_{YX}}$)\\[.005in]
\item $\Phi_{\delta,\gamma}(\hat{S}_Y - {S_{Y}^\star}, \hat{L}_Y - {L_{Y}^\star}, \hat{\Theta}_{YX} - {\Theta_{YX}^\star}, \hat{\Theta}_{X} - {\Theta_{X}^\star}) \leq C_1\lambda_{n}$; that is $\|\hat{S}_Y - S_Y^\star\|_{\ell_\infty} \lesssim \frac{\beta}{\alpha}{m}\delta\sqrt{\frac{p+q}{n}}$, $\|\hat{L}_Y - L_Y^\star\|_{2} \lesssim \frac{\beta}{\alpha}{m}\sqrt{\frac{p+q}{n}}$, $\|\hat{\Theta}_{YX} - \Theta_{YX}^\star\|_{2} \lesssim \frac{\beta}{\alpha}\gamma{m}\sqrt{\frac{p+q}{n}}$, $\|\hat{\Theta}_{X} - \Theta_{X}^\star\|_{2} \lesssim \frac{\beta}{\alpha}{m}\sqrt{\frac{p+q}{n}}$ \hspace{.05in} if \hspace{.05in} $\lambda_n \sim \beta{m}\sqrt{\frac{p+q}{n}}$.
\end{enumerate}
\end{theorem}

Notice that condition 1 of Theorem~\ref{theorem:main} ensures that the interval in condition 2 is non-empty. We give a proof of Theorem~\ref{theorem:main} in Appendix~\ref{section:proof}. Under the assumptions of the theorem, we construct appropriate primal feasible variables $(\tilde{\Theta},\tilde{S}_Y,\tilde{L}_Y)$ that satisfy the conclusions of the theorem -- i.e., $\tilde{\Theta}_{YX}, \tilde{L}_Y$ are low-rank (with the same ranks as the underlying population quantities $\Theta_{YX}^\star$ and $L_Y^\star$) and $\tilde{S}_Y$ is sparse (with the same support as the underlying population quantity $S_Y^\star$) -- and for which there exists a corresponding dual variable certifying optimality. This proof technique is sometimes also referred to as a primal-dual witness or certificate approach \citep{Wai2009}. The quantities $\alpha, \beta$ (related to $\nu$, as stated in the theorem, via $\beta=\tfrac{3-\nu}{\nu}$), $\omega_Y, \omega_{YX}, \mathrm{deg}(S_Y^\star)$ as well as the choices of the parameters $\delta, \gamma$ play a prominent role in our result. Indeed, larger values of $\alpha$ and $\nu$ (leading to a better conditioned Fisher information, from Assumptions 1 and 2 in \eqref{eqn:FirstFisherCond} and \eqref{eqn:SecondFisherCond}) lead to less stringent requirements on the sample complexity, on the minimum magnitude nonzero entry $\tau_Y$ of $S_{Y}^\star$, on the minimum nonzero singular value $\sigma_Y$ of $L_{Y}^\star$, and on the minimum nonzero singular value $\sigma_{YX}$ of $\Theta_{YX}^\star$. In a similar vein, larger values of the quantities $\omega_{Y}$ and $\omega_{YX}$ in Assumptions 1 and 2 imply that the Fisher information is well-conditioned even for large distortions of the tangent spaces $T(L_Y^\star)$ and $T(\Theta_{YX}^\star)$ (see \eqref{eqn:Udef}), which in turn lead to less stringent requirements on the minimum nonzero singular values $\sigma_{Y}$ and $\sigma_{YX}$.

{\par}As is clear from Theorem~\ref{theorem:main}, the tradeoff parameters $\delta, \gamma$ must be chosen such that the population Fisher information $\mathbb{I}^\star$ satisfies Assumptions 1 and 2 (see \eqref{eqn:FirstFisherCond} and \eqref{eqn:SecondFisherCond}) for some $\alpha > 0, \nu \in (0, 1/3), \omega_Y > 0, \omega_{YX} > 0$. Recall that these assumptions are stated in terms of the norm $\Phi_{\delta,\gamma}$, as defined in \eqref{eqn:Phidef}. The key complication that arises in characterizing values of $\delta,\gamma$ for which Assumptions 1 and 2 hold is that these parameters appear as multiplicative factors on norms imposed on different sub-blocks of matrices in $\Sp^{p+q}$ (see the definition of $\Phi_{\delta, \gamma}$ in \eqref{eqn:Phidef}), which lead to conditions on gains of the Fisher information that are coupled across the different sub-blocks of $\Sp^{p+q}$. In order to overcome this difficulty, we describe in Appendix~\ref{sec:Parameterchoices} a set of conditions on gains of the Fisher information restricted to the tangent spaces $\Omega(S_Y^\star),T^\star_Y,T^\star_{YX}$ \emph{separately}; under these separable conditions, we explicitly characterize a non-empty polyhedral set $V({\alpha, \nu, \omega_Y, \omega_{YX}}) \subset \mathbb{R}^2$ such that Assumptions 1 and 2 hold for all $(\delta,\gamma) \in V({\alpha, \nu, \omega_Y, \omega_{YX}})$. These conditions are interpretable and are stated in terms of the degree of the graphical model structure underlying the conditional distribution of $S_Y^\star$ (this quantity, $\text{deg}(S_Y^\star)$, makes an appearance in Theorem ~\ref{theorem:main}) and an incoherence parameter associated with the low-rank matrix $L_Y^\star$.

\section{Experiments}
\label{section:experiments}

We illustrate the performance of the estimators corresponding to the SDR-FM, SDR-GM, and SDR-LVGM approaches (recall that these estimators are given by  \eqref{eqn:Generic} with the choices \eqref{eqn:FMReg}, \eqref{eqn:GMReg}, and \eqref{eqn:LVGMReg} for the regularizer $R(\Theta_Y)$) and the estimator \eqref{eqn:CovariateSelection} in statistical modeling tasks involving financial asset data and newsgroup data. We solve these convex programs numerically using the LogdetPPA package developed for log-determinant semidefinite programs \citep{TohTT2002}. Below we discuss the details of each dataset:\\

\par{\bf{financial}}: We consider a financial asset modeling problem in which the responses $Y$ are a collection of monthly stock returns of $67$ companies from the Standard and Poor index from 1990 to 2005 and the covariates $X$ are the following $7$ variables -- $X_1$: EUR to USD exchange rate, $X_2$: inflation rate, $X_3$: oil exports, $X_4$: industrial production index, $X_5$: population growth rate, $X_6$: consumer price index, and $X_7$: unemployment rate. Thus, we observe $n=188$ samples jointly of $(Y,X) \in \R^{67} \times \R^{7}$.\\

\par{\bf{newsgroup}}: This dataset consists of $n = 16242$ observations in $\R^{100}$, with each observation corresponding to a news document. The coordinates of these observations are indexed by a collection of $100$ words, and each observation is a binary vector specifying whether a word appears in a document. Of these $100$ words, the following $9$ words are designated as covariates $X \in \R^9$ as they are useful in categorizing newsgroup documents: $X_1$: government, $X_2$: religion, $X_3$: science, $X_4$: technology, $X_5$: war, $X_6$: medicine, $X_7$: world, $X_8$: food, and $X_9$: games. The remaining $91$ words specify the response $Y \in \R^{91}$. \\

\subsection{Sufficient Dimension Reduction and Conditional Modeling}

{\noindent}{\noindent}
In this section, we investigate the performance of SDR-GM and SDR-LVGM (with the estimator \eqref{eqn:Generic} and choices \eqref{eqn:GMReg} and \eqref{eqn:LVGMReg} for the regularizer $R(\Theta_Y)$) on the financial dataset . To illustrate the utility of incorporating information about the covariates, we also contrast these methods with two modeling approaches that do not account for the impact of the covariates $X$ on the responses $Y$ -- we fit a sparse graphical model (denoted GM) as well as a latent-variable graphical model (denoted LVGM) to the responses $Y$ using the Graphical Lasso technique \citep{YuanLin2006, Friedman2008} and the approach described by \citet{Chand2012}, respectively.

To properly compare the performance of these different techniques, we ensure that the complexity of each of the resulting models (in terms of the number of parameters required for their specification) is approximately the same.\footnote{The number of parameters required to specify a sparse graphical model with a precision matrix $N \in \Sp^{p}$ is equal to $p$ plus one-half the number of nonzero off-diagonal entries of $N$ (as $N$ is symmetric). The number of parameters required to specify a $p \times q$ matrix with rank $r \leq \min\{p,q\}$ is equal to $r(p + q) - r^2$. Finally, the number of parameters required to specify a matrix in $\Sp^{p}$ with rank $r \leq p$ is equal to $rp - r(r-1)/2$.} We begin by choosing a regularization parameter for the Graphical Lasso method such that small changes in the value of this parameter do not lead to substantial structural changes in the estimated graphical model, i.e., the estimated model is stable with respect to small changes in the regularization parameter. Following this approach, we use the GM approach on the financial dataset (without the covariates), and the resulting graphical model consists of $909$ parameters ($842$ edges plus $67$ node parameters). Next, we choose regularization parameters for the estimators corresponding to the LVGM (without covariates), SDR-GM (with both covariates and responses), and SDR-LVGM (with both covariates and responses) approaches such that the resulting models obtained via each of these techniques consist of approximately $909$ parameters. Specifically, the model obtained using LVGM consists of $10$ latent variables and a conditional graphical model (conditioned on the $11$ latent variables) with $221$ edges for a total of $913$ parameters. The model obtained using SDR-GM consists of a $4$-dimensional projection of the $7$ covariates (that is sufficient with respect to the responses) and a conditional graphical model over the responses (conditioned on the dimension-reduced covariates) that consists of $564$ edges for a total number of $908$ parameters. Finally, the model obtained using SDR-LVGM consists of a $3$-dimensional projection of the covariates, $7$ latent variables, and a conditional graphical model (conditioned on the dimension-reduced covariates and the $7$ latent variables) with $180$ edges for a total of $908$ parameters. Let $\hat{\Theta}_{\rm{GM}} \in \Sp_{++}^{67}$, $\hat{\Theta}_{\rm{LVGM}} \in \Sp_{++}^{67}$, $\hat{\Theta}_{\rm{SDR\text{-}GM}} \in \Sp_{++}^{67+7}$, and $\hat{\Theta}_{\rm{SDR\text{-}LVGM}} \in \Sp_{++}^{67+7}$ denote the precision matrices of the models corresponding to GM, LVGM, SDR-GM, and SDR-LVGM respectively. Although these models have similar complexities, they have different predictive performances, as described next.

{\par}We assess the predictive performance of each of these four models on $90$ monthly observations $\{(Y_{\text{test}}^{(j)}, X_{\text{test}}^{(j)}\}_{i=1}^{90} \in \R^{67} \times \R^7$ from $2006$ to $2013$ (recall that the training set based on which the four models were obtained consisted of $188$ monthly observations during the period $1990$ to $2006$). For the models obtained via GM and LVGM, we compute the average log-likelihood over the test samples using the distributions specified by the precision matrices $\hat{\Theta}_{\rm{GM}}$ and $\hat{\Theta}_{\rm{LVGM}}$ respectively. For the model obtained via SDR-GM that accounts for the influence of the covariates on the responses, we compute the log-likelihood of each test sample $Y_{\text{test}}^{(j)}$ (and subsequently average over all test samples) with respect to the distribution of $Y|[f(X) = f(X_{\text{test}}^{(j)})]$ where $f(X)$ is the projection of $X$ into the row-space of the matrix $-(\hat{\Theta}_{\mathrm{SDR\text{-}GM}})_{Y}^{-1}(\hat{\Theta}_{\mathrm{SDR\text{-}GM}})_{YX}$ (recall that $-(\hat{\Theta}_{\mathrm{SDR\text{-}GM}})_{Y}^{-1}(\hat{\Theta}_{\mathrm{SDR\text{-}GM}})_{YX}$ denotes the map of best linear estimator of $Y$ based on $X$). We follow a similar approach to compute the predictive performance of the model obtained via SDR-LVGM. The average predictive log-likelihoods of models obtained via GM, LVGM, SDR-GM, and SDR-LVGM are $-127.55$, $-122.73$, $-121.12$, and $-120.28$ respectively. For comparison, a model obtained via FM (factor modeling) on the training set -- without incorporating the covariates -- consists of $14$ latent factors (for a total of $914$ parameters) and gives an average predictive log-likelihood of $-123.99$ over the test set. On the other hand, a model obtained via SDR-FM (using the estimator \eqref{eqn:Generic} with the regularizer $R(\Theta_Y)$ set as in \eqref{eqn:FMReg}) on the training set -- that incorporates observations of both the covariates and the responses -- consists of a $5$-dimensional projection of the covariates and $8$ latent factors in the conditional model of the responses given the dimension-reduced covariates (the total number of parameters equals $920$), and it provides an average predictive log-likelihood of $-121.35$ on the test set. As larger values of average log-likelihood are indicative of a better fit to the test samples, these results suggest that SDR-LVGM offers the best predictive performance of the different approaches considered in this experiment.

Focussing on the structural aspects of the models obtained via GM, LVGM, SDR-GM, and SDR-LVGM, Fig.~\ref{fig:graphicalS} shows the (conditional) graphical model structures over the responses corresponding to each of these approaches. The ten strongest edges in the conditional graphical model obtained via SDR-LVGM in Fig.~\ref{fig:graphicalS}(d) (in terms of the magnitude of the entries in the precision matrix) are between General Electric - American Express, Target - Bancorp, Hewlett Packard - Oracle, Texas Instruments - General Electric, Occidental Petroleum Corp. - Wells Fargo, American Insurance Group - Bank of New York, Merck $\&$ Co. - Walgreens, JP Morgan - Verizon,
Verizon - CVS Health, Pfizer Inc. - Colgate. The presence of some strong edges between entities in different industries suggests that dependencies between assets belonging to the same industry may be better modeled via the latent variables or the dimension-reduced covariates. Each of these ten edges in the conditional graphical model obtained via SDR-LVGM also appears as a strong edge in the conditional graphical model obtained via LVGM (in which the covariates are not incorporated). Examples of other strong edges in the conditional graphical model obtained via LVGM (which do not appear in the conditional graphical model obtained via SDR-LVGM) include IBM Corp. - American Insurance Group, Hewlett Packard - Southern Company, Hewlett Packard - Walmart, General Electric - Colgate, Fedex - Emerson. 
\vspace{-.14in}
\FloatBarrier
\begin{figure}[!http]
\centering
\subfigure[GM]{
\includegraphics[width=3cm, height = 3cm]{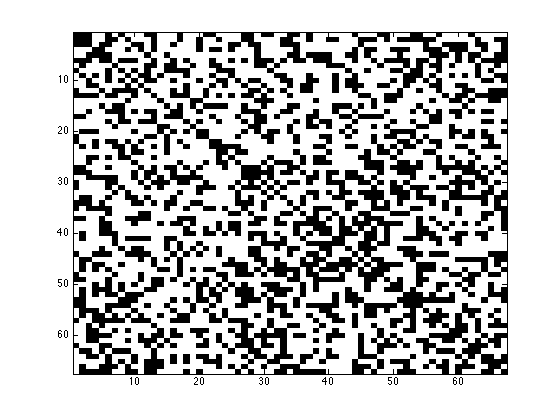}
}
\subfigure[LVGM]{
\includegraphics[width=3cm, height = 3cm]{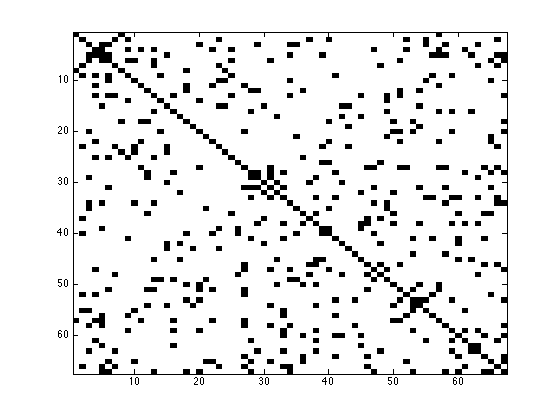}
}
\subfigure[SDR-GM]{
\includegraphics[width=3cm, height = 3cm]{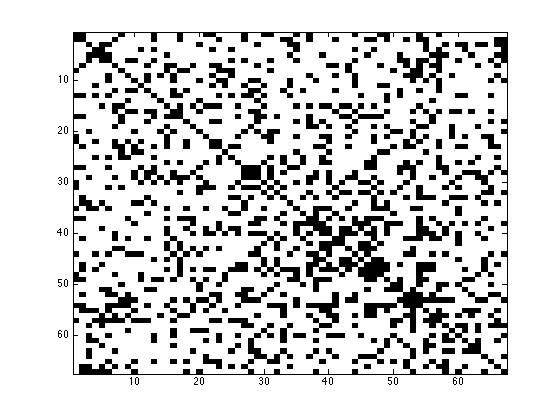}
}
\subfigure[SDR-LVGM]{
\includegraphics[width=3cm, height = 3cm]{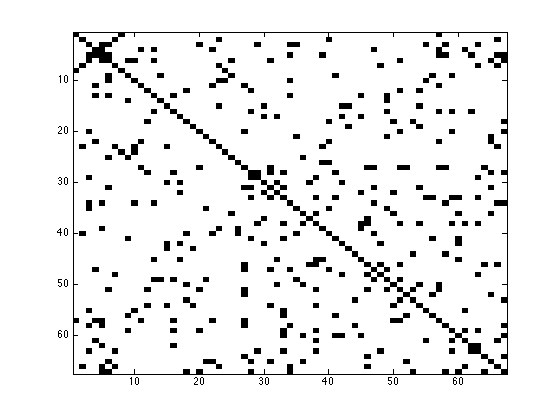}
}
\caption{These figures show the sparsity pattern (black denotes an edge, and white denotes no edge) of the graphical models associated with each modeling paradigm} %The manner in which these results are obtained is discussed in detail in section 3.1.}
\label{fig:graphicalS}
\end{figure}
\FloatBarrier
\vspace{-.2in}
{\par}Turning our attention to the latent components identified in the models obtained via LVGM and SDR-LVGM, we note that the number of latent variables in the model obtained using SDR-LVGM ($7$ variables) is smaller than that in the model obtained via LVGM ($10$ variables). This observation suggests that the $3$-dimensional projection of the covariates in the model obtained via SDR-LVGM accounts for some of the effect of the latent variables in the model obtained using LVGM (which does not incorporate the covariates). To illuminate this point quantitatively, let the matrix ${\Gamma} = -(\hat{\Theta}_{\mathrm{SDR\text{-}LVGM}})_{Y}^{-1}(\hat{\Theta}_{\mathrm{SDR\text{-}LVGM}})_{YX} \in \mathbb{R}^{67 \times 7}$ denote the map of the best linear estimator of $Y$ based on $X$ (specified by the model obtained via SDR-LVGM) and let the matrix $\Lambda \in \mathbb{R}^{67 \times 10}$ denote the map of best linear estimator of $Y$ based on the latent variables $\zeta \in \R^{10}$ (specified by the model obtained via LVGM\footnote{The matrix $\Lambda$ is only known up to right-multiplication by a non-singular linear transformation. As a result, the key invariants are the rank and the column space of $\Lambda$.}). The column space of $\Gamma$ corresponds to the 3-dimensional image of the mapping of the best linear estimator of $Y$ based on $X$, and it represents the component of the response $Y$ that is correlated with the covariates $X$ in the model obtained using the SDR-LVGM approach. Similarly, the column space of  $\Lambda$ represents the 10-dimensional component of Y that is correlated with the latent variables $\zeta$ in the model obtained using the LVGM approach.  As such, the closeness between these column spaces measures the degree to which the sufficient dimension reduction $f(X)$ accounts for some of the influence of the latent variables $\zeta$ on the covariates $Y$. The largest principal angle between the $3$-dimensional column space of $\Gamma$ and the $10$-dimensional column space of $\Lambda$ is $11.06$ degrees. This result indicates that the dimension reduced covariates in the model obtained using SDR-LVGM account for some of the effect of the latent components in the model obtained via LVGM.

\subsection{Combining Covariate Selection and Graphical Modeling}
In this section, we investigate the performance of the estimator \eqref{eqn:CovariateSelection} on the financial and the newsgroup datasets. Recall that this estimator selects a subset of the covariates $X$ that is most useful for predicting the responses $Y$, while simultaneously fitting a latent-variable graphical model to the conditional distribution of the responses given selected covariates. We denote this modeling approach as CS-LVGM.

We apply the estimator \eqref{eqn:CovariateSelection} to the financial and the newsgroup datasets with fixed choices for parameters $\lambda_n, \delta$ and different choices of the parameter $\gamma$. The objective of this experiment is to illustrate the different covariates selected in each dataset as $\gamma$ is varied. For each of the models obtained using this CS-LVGM approach, Table~\ref{table:financial} and Table~\ref{table:newsgroup} list the subset of covariates that were selected (recall that the selected covariates are represented by the indices of the nonzero columns of the submatrix $\hat{\Theta}_{YX} \in \mathbb{R}^{p \times q}$ of the estimated joint precision matrix $\hat{\Theta}_{++}^{p+q}$). As expected, larger values of $\gamma$ yield a smaller subset of relevant covariates as the regularization term $\gamma \|\Theta_{YX}\|_{1,2}$ in \eqref{eqn:CovariateSelection} is enforced more strongly. With the financial dataset, the covariates $X_1$ (population growth rate) and $X_5$ (EUR to USD exchange rate) persist as $\gamma$ increases, which suggests that they are the most relevant for predicting stock returns among the seven covariates considered in the experiment. On the other hand, the covariate $X_3$ (oil exports) appears to be the least influential. For the newsgroup dataset, the covariates $X_1$ (government) and $X_7$ (world) persist as $\gamma$ increases, suggesting that these are the most useful for predicting word occurrences in documents in the 20newsgroup dataset, while the covariates $X_6$ (medicine) and $X_8$ (food) do not seem as relevant.\\
\vspace{-.15in}
\begin{center}
\begin{minipage}{0.4\textwidth}
\begin{tabular}{|c | c| c|}
\hline
$\gamma$ & covariates identified\\
\hline
{1.50} & \{$X_1, X_2, X_4, X_5, X_6, X_7$\}\\
{1.99} & \{$X_1, X_2, X_4, X_5$\}\\
2.14 & \{$X_1, X_5, X_7$\} \\
2.57 & \{$X_1, X_5$\}\\
\hline
\end{tabular}
\captionof{table}{$\gamma$ vs selected covariates for financial dataset with $\lambda_n = 0.58$ and $\delta = 0.29$}
\label{table:financial}
\end{minipage}\qquad
\begin{minipage}{0.4\textwidth}
\begin{tabular}{|c | c| c|}
\hline
$\gamma$ & covariates identified\\
\hline
1.39 & \{$X_1, X_2, X_3, X_4, X_5, X_7, X_9$\} \\
2.32 &\{$X_1, X_2, X_3, X_5, X_9$\} \\
2.80 & \{$X_1, X_2, X_3, X_9$\} \\
3.02 & \{$X_1, X_7$\} \\
\hline
\end{tabular}
\captionof{table}{$\gamma$ vs selected covariates for newsgroup dataset with $\lambda_n = 0.38$ and $\delta = 0.41$}
\label{table:newsgroup}
\end{minipage}
\end{center}

\vspace{-.08in}
We inspect more closely the model obtained using the CS-LVGM approach on the newsgroup dataset with parameters $\lambda_n = 0.38, \delta = 0.41, \gamma = 2.32$ (this corresponds to the second line in Table~\ref{table:newsgroup}). This model consists of $5$ selected covariates, $6$ latent variables, and a conditional graphical model (conditioned on the $5$ covariates and the $6$ latent variables) with $10$ edges for a total number of $1087$ parameters. The $10$ edges in this conditional graphical model are between God-Jesus, Dos-Windows, Bible-God, card-video, email-phone, Christian-God, state-university, computer-university, disk-drive, and Israel-Jews. For comparison, we obtain a model using the LVGM approach (ignoring the covariates) of comparable complexity with a total of $1083$ parameters -- this model consists of $11$ latent variables and a conditional graphical model (conditioned on the $11$ latent variables) with $46$ edges. Each of the ten edges in the conditional graphical model obtained via CS-LVGM appear as stronger edges in the conditional graphical model obtained via LVGM. Examples of additional strong edges in the conditional graphical model obtained via LVGM include patients-disease, hockey-NHL, and players-baseball. The edge between patients-disease presents an interesting illustration as these two words appear together in $64$ documents. In $45$ of these $64$ documents, however, at least one of the $5$ covariates selected by the CS-LVGM model ($X_1, X_2, X_3, X_5, X_9$ from Table~\ref{table:newsgroup}) also appears. Thus, the lack of an edge between patients-disease in the conditional graphical model obtained using the CS-LVGM approach is perhaps explained by the inclusion of the the $5$ covariates. In a similar vein, the absence of the edges hockey-NHL and players-baseball in the conditional graphical model obtained via CS-LVGM may be attributed to the inclusion of the covariate $X_9$ (game). In contrast, a curious point arises when considering the presence of the edges God-Jesus, Bible-God, Christian-God, and Israel-Jews in the conditional graphical model obtained via CS-LVGM. From a preliminary inspection, one might expect that the inclusion of the covariates $X_2,X_7$ (religion, world) would account for these four pairs of interactions, resulting in their absence in the conditional graphical model. However, (as an example) the words God and Jesus appear together in $380$ documents, and only $180$ of these documents include at least one of the $5$ selected covariates. Thus, it is not surprising that the edge God-Jesus remains in the conditional graphical model obtained using the CS-LVGM approach despite this graphical model being conditioned on the $5$ selected covariates.

\section{Further Directions}
%
%Given observations of a collection of responses and covariates $(Y, X) \in \R^{p} \times \R^{q}$, we considered the problem of integrating SDR and modeling the conditional distribution of $Y | f(X)$.  We discuss three concrete approaches for modeling this conditional distribution, namely, factor modeling, graphical modeling, and latent-variable graphical modeling, and we describe a computationally tractable convex program for combining each of these methods with SDR.  We also provide theoretical support (via a consistency analysis in the high-dimensional scaling regime) as well as experimental evidence of the utility of our framework.
{\par} The selection of regularization parameters is a common practical challenge in high-dimensional estimation problems.  Approaches based on cross-validation are widely used, although these techniques optimize for prediction performance and do not always yield concise models.  To address this shortcoming, several methods have been proposed recently for the selection of regularization parameters for the Lasso \citep{MeiB2010}. It is of interest to extend these ideas to our estimator \eqref{eqn:Generic}, which requires the specification of multiple regularization parameters.  
Further, the convex programs proposed in this paper are computationally tractable (i.e., solvable to a desired accuracy in polynomial time), but it would be useful to develop special-purpose numerical schemes -- perhaps building on recent scalable algorithms for the Graphical Lasso \citep{Friedman2008, Ravi2013} -- for efficient solution in massive-scale problems. Finally, it is of interest to extend our techniques to non-Gaussian settings, in which one may be interested in identifying nonlinear dimension reductions of the covariates that are sufficient with respect to the responses.

%\bibliographystyle{biometrika}
%\bibliography{paper-ref}

%\bibliographystyle{biometrika}
%\bibliography{paper-ref}

\begin{thebibliography}{7}
\expandafter\ifx\csname natexlab\endcsname\relax\def\natexlab#1{#1}\fi

\bibitem[{Adragni $\&$ Cook(2009)}]{AdragniCook2009}
\textsc{Adragni, K. P $\&$ Cook, D. R.} (2009).
\newblock {{Sufficient dimension reduction and prediction in regression}}.
\newblock \textit{Philosophical Transactions of Royal Society A}. {\bf 367}, 4385--4400.

%\bibitem[{Allman et al.(2009)}]{AllMR2009}
%\textsc{Allman, E.~S., Matias, C. $\&$ Rhodes, J. A.} (2009).
%\newblock {{Identifiability of parameters in latent structure models with many observed variables}}.
%\newblock \textit{Ann. Statistics}. \textbf{37}, 3099--3132.

%\bibitem[{Bach(2008)}]{Bac2008}
%\textsc{Bach, F.} (2008).
%\newblock {{Consistency of trace norm minimization}}.
%\newblock \textit{J. Mach. Lear. Res.}. \textbf{9}, 1019--1048.
%
%\bibitem[{Bach(2008)}]{Bach2008}
%\textsc{Bach, F.} (2008).
%\newblock {{Consistency of trace norm minimization}}.
%\newblock \textit{Journal of Machine Learning Research}. {\bf 9}, 1019--1048.
\bibitem[{Bach(2008)}]{Bach2008}
\textsc{Bach, F.} (2008).
\newblock {{Consistency of trace norm minimization}}.
\newblock \textit{Journal of Machine Learning Research}. {\bf 9}, 1019--1048.

\bibitem[{Bickel $\&$ Levina(2008a)}]{Bicla2008}
\textsc{Bickel, P. J. $\&$ Levina, E.} (2008).
\newblock {{Regularized estimation of large covariance matrices}}.
\newblock \textit{Annals of Statistics}. {\bf 36}, 199--227.


\bibitem[{Bickel $\&$ Levina(2008b)}]{Biclb2008}
\textsc{Bickel, P. J. $\&$ Levina, E.} (2008).
\newblock {{Covariance regularization by thresholding}}.
\newblock \textit{Annals of Statistics}. {\bf 36}, 2577--2604.

\bibitem[{Brem $\&$ Kruglyak(2005)}]{Brem2005}
\textsc{Brem, R. $\&$ Kruglyak, L.} (2005).
\newblock {{The landscape of genetic complexity across 5,700 gene expression traits in yeast}}.
\newblock \textit{Proceeding of National Academy of Sciences}. {\bf 102}, 1572--1577.



%\bibitem[{Boyd $\&$ Vandenberghe(2004)}]{BoyV2004}
%\textsc{Boyd, S. P. $\&$ Vandenberghe, L.} (2004).
%\newblock \textit{{Convex optimization}}.
%\newblock Cambridge University Press.

\bibitem[{B\"uhlmann $\&$ van de Geer(2011)}]{BuhV2011}
\textsc{B\"{u}hlmann, P $\&$ van de Geer, S.} (2011).
\newblock \textit{{Statistics for high-dimensional data}}.
\newblock New York: Springer.


%\bibitem[{Cand\`es et al.(2011)}]{CanLM2011}
%\textsc{Cand\`es, E. J., Li, X., Ma, Y. $\&$ Wright, J.} (2011).
%\newblock {{Robust
%principal component analysis?}}.
%\newblock \textit{J. of ACM}. \textbf{58}, 1--37.

\bibitem[{Cand\`es et al.(2006)}]{CanRT2006}
\textsc{Cand\`es, E. J., Romberg, J. $\&$ Tao, T.} (2006).
\newblock {{Robust
uncertainty principles: exact signal reconstruction from highly
incomplete frequency information}}.
\newblock \textit{IEEE Transaction on Information Theory}. {\bf 52}, 489--509.

\bibitem[{Cand\`es $\&$ Recht(2009)}]{CanR2009}
\textsc{Cand\`es, E. J. $\&$ Recht, B.} (2009).
\newblock {{Exact matrix completion
via convex optimization}}.
\newblock \textit{Foundation of Computational Mathematics}. \textbf{9}, 717--772.

%\bibitem[{Cand\`es $\&$ Plan(2009)}]{CandPlan2009}
%\textsc{Cand\`es, E. J. $\&$ Plan, Y.} (2009).
%\newblock {{Near-ideal model selection by $\ell_1$ minimization}}.
%\newblock \textit{Annals of Statistics}. \textbf{37}, 2145--2177.


%\bibitem[{Cai et al.(2011)}]{CauLiLiuLuo2011}
%\textsc{Cai, T., Li, H., Liu, W. $\&$ Luo, X.} (2011).
%\newblock {{A constrained $\ell_1$ minimization approach to sparse precision matrix estimation}}.
%\newblock \textit{J. Amer. Statist. Assoc}. \textbf{100}, 1--19.

\bibitem[{Cai et al.(2010)}]{Cai2010}
\textsc{Cai, T., Li, H., Liu, W. $\&$ Xie, J.} (2010).
\newblock {{Covariate adjusted precision matrix estimation with an application in genetical genomics}}.
\newblock \textit{Biometrika}. \textbf{106}, 1--19.


\bibitem[{Chandrasekaran et al.(2012)}]{Chand2012}
\textsc{Chandrasekaran, V., Parrilo, P. A. $\&$ Willsky, A. S.} (2012).
\newblock {{Latent Variable Graphical Model Selection via Convex Optimization}}.
\newblock \textit{Annals of Statistics}. \textbf{40}, 1935--1967.


%\bibitem[{Chandrasekaran et al.(2009)}]{ChaSPW2009}
%\textsc{Chandrasekaran, V., Sanghavi, S., Parrilo, P. A. $\&$ Willsky, A.
%S.} (2009).
%\newblock {{Rank-sparsity incoherence for matrix decomposition}}.
%\newblock \textit{SIAM Jour. Opt.}. \textbf{21}, 572--596.


\bibitem[{Chechik et al.(2005)}]{ChechikInform2005}
\textsc{Chechik, G., Globerson, A., Tishby, N. $\&$ Weiss, Y.} (2005).
\newblock {{Information bottleneck for {G}aussian variables}}.
\newblock \textit{Journal of Machine Learning Research}. \textbf{6}, 165--188.

\bibitem[{Chen et al.(2001)}]{Chen}
\textsc{Chen, S. S., Donoho, D. L. $\&$ Saunders, M. A} (2001).
\newblock {{Atomic decomposition by basis pursuit}}.
\newblock \textit{SIAM Review}. \textbf{43}, 129--159.


%\bibitem[{Chen et al.(2013)}]{ChenCCA2013}
%\textsc{Chen, M., Gao, C., Ren, Z., $\&$ Zhou, H. H.} (2013).
%\newblock {{Sparse CCA via Precision Adjusted Iterative Thresholding}}.
%\newblock \textit{arXiv:1311.6186}.

\bibitem[{Cheung $\&$ Spielman(2002)}]{Cheung2002}
\textsc{Cheung, V. $\&$ Spieldman, R.} (2002).
\newblock {{The genetics of variation in gene expression}}.
\newblock \textit{Nature Genetics}. {\bf 35}, 2313--2351.


\bibitem[{Cook $\&$ Ni(2005)}]{CookNi2005}
\textsc{Cook, R. D. $\&$ Ni, L.} (2005).
\newblock {{Sufficient dimension reduction via inverse regression: A minimum discrepancy approach}}.
\newblock \textit{Journal of American Statistical Association}. {\bf 100}, 410--428.

\bibitem[{Cook $\&$ Weisberg(1991)}]{CookSAVE1991}
\textsc{Cook, R. D. $\&$ Weisberg, S.} (1991).
\newblock {{Discussion of ``Sliced inverse regression for dimension reduction"}}.
\newblock \textit{Journal of American Statistical Association}. {\bf 86}, 328--332.

%\bibitem[{Dempster et al.(1977)}]{DemLR1977}
%\textsc{Dempster, A. P., Laird, N. M. $\&$ Rubin, D. B.} (1977).
%\newblock {{Maximum
%likelihood from incomplete data via the EM algorithm}}.
%\newblock \textit{J. Roy.
%Stat. Soc. B}. \textbf{39}, 1--38.

\bibitem[{Davidson $\&$ Szarek(2001)}]{DavidSzarek2001}
\textsc{Davidson, K.R. $\&$ Szarek, S.J.} (2001).
\newblock {{Local operator theory, random matrices and {B}anach spaces}}.
\newblock \textit{Handbook of the Geometry of Banach Spaces}. {\bf I}, 317--366.


\bibitem[{Donoho(2006)}]{Donb2006}
\textsc{Donoho, D. L.} (2006).
\newblock {{Compressed sensing}}.
\newblock \textit{IEEE Transactions on
Information Theory}. {\bf 52}, 1289--1306.

%\bibitem[{Elidan(2007)}]{EliNN2007}
%\textsc{Elidan, G., Nachman, I. $\&$ Friedman, N.} (2007).
%\newblock {{Ideal
%Parent'' structure learning for continuous variable Bayesian networks}}.
%\newblock \textit{J. Mach. Lear. Res.}. {\bf 8} 1799--1833.


\bibitem[{El Karoui(2008)}]{Elk(2008)}
\textsc{El Karoui, N.} (2008).
\newblock {{Operator norm consistent estimation of
large-dimensional sparse covariance matrices}}.
\newblock \textit{Annals of Statistics}. {\bf 36} 2717--2756.


%\bibitem[{Fan et al.(2008)}]{FanFL2008}
%\textsc{Fan, J., Fan, Y. $\&$ Lv, J.} (2008).
%\newblock {{High dimensional covariance
%matrix estimation using a factor model}}.
%\newblock \textit{J. Econometrics}. {\bf
%147} 186--197.

\bibitem[{Duan $\&$ Li(1991)}]{Duan1991}
\textsc{Duan, N. $\&$ Li, K.C.} (1991).
\newblock {{Slicing regression: A link-free regression method}}.
\newblock \textit{Annals of Statistics}. {\bf19}, 505--530.

\bibitem[{Fan et al.(2008)}]{Fan2008}
\textsc{Fan, J., Fan, Y. $\&$ Lv, J.} (2008).
\newblock {{High-dimensional covariance matrix estimation using a factor model}}.
\newblock \textit{Journal of Econometrics}. {\bf
147} 186--197.



\bibitem[{Fazel(2002)}]{Faz2002}
\textsc{Fazel, M.} (2002).
\newblock {{Matrix rank minimization with applications}}.
\newblock \textit{PhD thesis, Department of Electrical Engineering, Stanford
University}. 2002.


\bibitem[{Friedman et al.(2008)}]{Friedman2008}
\textsc{Friedman, J., Hastie, T $\&$ Tibshirani, R} (2008).
\newblock {{Sparse inverse covariance estimation with the graphical lasso}}.
\newblock \textit{Biostatics}. {\bf
9} 432--441.

\bibitem[{Fung et al.(2002)}]{Fung2002}
\textsc{Fung, W. K., He, X., Liu, L. $\&$ Shi, P.} (2002).
\newblock {{Dimension reduction based on canonical correlation}}.
\newblock \textit{Statistica Sinica}. {\bf
12} 1093--1113.

%
%\bibitem[{Gao et al.(2014)}]{GaoCCA2014}
%\textsc{Gao, C., Ma, Z., Ren, Z. $\&$ Zhou, H. H.} (2014).
%\newblock {{Minimax estimation in sparse canonical correlation analysis}}.
%\newblock \textit{arXiv:1405.1595}


%\bibitem[{Hardoon $\&$ Shawe-Taylor(2011)}]{HS2011}
%\textsc{Hardoon, D. R. $\&$ Shawe-Taylor, J.} (2011).
%\newblock {{Sparse canonical correlation analysis}}.
%\newblock \textit{Machine Learning}. {\bf 83} 331--353.


%\bibitem[{Johnstone(2001)}]{Joh2001}
%\textsc{Johnstone, I. M.} (2001).
%\newblock {{On the distribution of the largest
%eigenvalue in principal components analysis}}.
%\newblock \textit{Ann. Statistics}. {\bf 29} 295--327.


\bibitem[{Horn $\&$ Johnson(1990)}]{Horn1990}
\textsc{Horn, R. A. $\&$ Johnson, C. R.} (1990).
\newblock {Matrix Analysis}.
\newblock \textit{Cambridge University Press, Cambridge}.



\bibitem[{Hsieh et al.(2013)}]{Ravi2013}
\textsc{Hsieh, C-J., Sustik, M. A., Dhillon, I. S., $\&$ Ravikumar, P.} (2013).
\newblock {{B{I}{G} $\&$ Q{U}{I}{C}: Sparse inverse covariance estimation for million variables}}.
\newblock \textit{Neural Information Processing System}.


\bibitem[{Kato(1995)}]{Kato1995}
\textsc{Kato, T.} (1995).
\newblock {{Perturbation theory for linear operators}}.
\newblock \textit{Springer}.


%\bibitem[{Lam $\&$ Fan(2009)}]{LamF2009}
%\textsc{Lam, C. $\&$ Fan, J.} (2009).
%\newblock {{Sparsistency and rates of
%convergence in large covariance matrix estimation}}.
%\newblock \textit{Ann. Statistics}. {\bf 37} 4254--4278.

%\bibitem[{Lauritzen(1996)}]{Lau1996}
%\textsc{Lauritzen, S. L.} (1996).
%\newblock \textit{Graphical models}. Oxford University Press.

%\bibitem[{Ledoit $\&$ Wolf(2003)}]{LedW2003}
%\textsc{Ledoit, O. $\&$ Wolf, M.} (2004).
%\newblock{{A well-conditioned estimator for
%large-dimensional covariance matrices}}.
%\newblock \textit{J. Multivar. Analysis}. {\bf 88} 365--411.

\bibitem[{Li(1991)}]{LiSDR1991}
\textsc{Li, K.-C.} (1991).
\newblock{{Sliced inverse regression for dimension reduction (with discussion)}}.
\newblock \textit{Journal of American Statistical Association}. {\bf 86} 316--342.

\bibitem[{Li(1992)}]{LiSDR1992}
\textsc{Li, K.-C.} (1992).
\newblock{{On principal Hessian directions for data visualization and dimension reduction: Another application of Stein's lemma}}.
\newblock \textit{Journal of American Statistical Association}. {\bf 87} 1025--1039.

%\bibitem[{Li $\&$ Wang(2007)}]{Li2007}
%\textsc{Li, B. $\&$ Wang, S.} (2007).
%\newblock{{On directional regression for dimension reduction}}.
%\newblock \textit{J. Amer. Stat. Ass.}. {\bf 105} 1188--1201.
%
%\bibitem[{Li et al.(2005)}]{Li2005}
%\textsc{Li, B., Zha, H. $\&$ Chiaromonte, F.} (2005).
%\newblock{{Contour regression: a general approach for dimension reduction}}.
%\newblock \textit{Ann. Stat.}. {\bf 33} 1580--1616.
%
%\bibitem[{Li(2007)}]{Li2007}
%\textsc{Li, L.} (2007).
%\newblock{{Sparse sufficient dimension reduction}}.
%\newblock \textit{Biometrika}. {\bf 94} 603--613.

%
%\bibitem[{Li $\&$ Yin(2008)}]{LiYin2008}
%\textsc{Li, L. $\&$ Yin, X.} (2008).
%\newblock{{Sliced inverse regression with regularizations}}.
%\newblock \textit{Biometrics}. {\bf 64} 124--131.

\bibitem[{Meinshausen $\&$ B\"{u}hlmann(2006)}]{MeiB2006}
\textsc{Meinshausen, N. $\&$ B\"{u}hlmann, P.} (2006).
\newblock{{High dimensional
graphs and variable selection with the Lasso}}.
\newblock \textit{Annals of Statistics}. {\bf 34} 1436--1462.

\bibitem[{Meinshausen $\&$ B\"{u}hlmann(2010)}]{MeiB2010}
\textsc{Meinshausen, N. $\&$ B\"{u}hlmann, P.} (2010).
\newblock{{Stability selection}}.
\newblock \textit{Journal of Royal Statistical Society: Series B}. {\bf 34} 1436--1462.

\bibitem[{Negahban $\&$ Wainwright(2011)}]{NegWain2011}
\textsc{Negahban, S. $\&$ Wainwright, M. J.} (2011).
\newblock{{Estimation of (near) low-rank matrices with noise and high-dimensional
scaling}}.
\newblock \textit{Annals of Statistics}. {\bf 39} 1069--1097.



%\bibitem[{Parkhomenko et al.(2009)}]{PTB2009}
%\textsc{Parkhomenko, E., Tritchler, D. $\&$ Beyene, J.} (2009).
%\newblock {{Sparse canonical correlation analysis}}.
%\newblock \textit{Statistical Applications in Genetics and Molecular Biology}. {\bf 8} 1--34.\\

\bibitem[{Ravikumar et al.(2008)}]{RavWRY2008}
\textsc{Ravikumar, P., Wainwright, M. J., Raskutti, G. $\&$ Yu, B.}
(2011).
\newblock{{High-dimensional covariance estimation by minimizing
$\ell_1$-penalized log-determinant divergence}}.
\newblock \textit{Electronic Journal of Statistics} {\bf 4} 935--980.

\bibitem[{Recht et al.(2009)}]{RecFP2009}
\textsc{Recht, B., Fazel, M. $\&$ Parrilo, P. A.} (2010).
\newblock{Guaranteed minimum rank solutions to linear matrix equations via nuclear norm
minimization}.
\newblock \textit{SIAM Review} {\bf 52} 471--501.

\bibitem[{Rothman et al.(2008)}]{RotBLZ2008}
\textsc{Rothman, A. J., Bickel, P. J., Levina, E. $\&$ Zhu, J.} (2008).
\newblock{Sparse permutation invariant covariance estimation}.
\newblock \textit{Electronic Journal of Statistics}. {\bf 2} 494--515.

%\bibitem[{Roth $\&$ Black(2005)}]{RothBlack2005}
%\textsc{Roth, S. $\&$ Black, M. J.} (2005).
%\newblock{Field of experts: A framework for learning image priors}.
%\newblock \textit{In Proc. CVPR}. {\bf 2} 860--867.

%
%\bibitem[{Scott(1966)}]{Scott(1966)}
%\textsc{Scott, J. T.} (1986).
%\newblock{Factor analysis and regression}.
%\newblock \textit{Econometrica}. {\bf 34} 552--562.


%\bibitem[{Speed $\&$ Kiiveri(1986)}]{SpeK1986}
%\textsc{Speed, T. P. $\&$ Kiiveri, H. T.} (1986).
%\newblock{Gaussian Markov
%distributions over finite graphs}.
%\newblock \textit{Annals of Statistics}. {\bf 14} 138--150.


\bibitem[{Sohn $\&$ Kim(2012)}]{SohnKim2012}
\textsc{Sohn, K-A. $\&$ Kim, S.} (2012).
\newblock{Joint estimation of structured sparsity and output structure in multiple-output regression via inverse-covariance regularization}.
\newblock \textit{Proceeding 15th International Conference Artificial Intelligence and Statistics}. {\bf 14} 1081--1089.


%\bibitem[{St\"{a}dler et al.(2010)}]{Stad2010}
%\textsc{St\"{a}dler, N., B\"{u}hlmann, P. $\&$ Geer, S.V.D.} (2010).
%\newblock{$\ell_1$-penalization for mixture regression models}.
%\newblock \textit{TEST}. {\bf 19} 209--256.

\bibitem[{Tibshirani(1996)}]{Tib1996}
\textsc{Tibshirani, R} (1996).
\newblock{Regression shrinkage and selection via the Lasso}.
\newblock \textit{Journal of Royal Statistics Society}. {\bf 58} 267--288.


\bibitem[{Toh et al.(2002)}]{TohTT2002}
\textsc{Toh, K. C, Todd, M. J. $\&$ Tutuncu, R. H.}
\newblock \textit{SDPT3 - a
MATLAB software package for semidefinite-quadratic-linear
programming}. Available from
http://www.math.nus.edu.sg/~mattohkc/sdpt3.html.

%\bibitem[{Waaijenborg $\&$ Zwinderman(2009)}]{WZ2009}
%\textsc{Waaijenborg, S. $\&$ Zwinderman, A. H.} (2009).
%\newblock {{Sparse canonical correlation analysis for identifying,
%connecting and completing gene-expression networks}}.
%\newblock \textit{BMC Bioinformatics}. {\bf 10} 315.


\bibitem[{Wainwright(2009)}]{Wai2009}
\textsc{Wainwright, M. J.} (2009).
\newblock{Sharp thresholds for noisy and
high-dimensional recovery of sparsity using $\ell_1$-constrained
quadratic programming (Lasso)}.
\newblock \textit{IEEE Transactions on Information Theory}. {\bf 55} 2183--2202.

\bibitem[{Wainwright(2014)}]{Wai2014}
\textsc{Wainwright, M. J.} (2014).
\newblock{Structured regularizers for high-dimensional problems: Statistical and computational issues}.
\newblock \textit{Annual Review of Statistics and its Applications}. {\bf 1} 233--253.


%\bibitem[{Watson(1992)}]{Wat(1992)}
%\textsc{Watson, G. A.} (1992).
%\newblock {Characterization of the subdifferential of some matrix norms}.
%\newblock \textit{Linear Algebra and Applications}. {\bf 170} 1039--1053.


%\bibitem[{Wiesel et al.(2008)}]{WKH2008}
%\textsc{Wiesel, A., Kliger, M. $\&$ Hero, A. O.} (2008).
%\newblock {{A greedy approach to sparse canonical correlation analysis}}.
%\newblock \textit{arXiv:0801.2748}.

%
%\bibitem[{Witten et al.(2011)}]{WitFast2011}
%\textsc{Witten, D. M., Friedman, J. H $\&$ Simon, N.} (2011).
%\newblock{New Insights and Faster Computations for the Graphical Lasso}.
%\newblock \textit{J Comp. and Graph. Stat.}. {\bf 20} 892--900.

%\bibitem[{Witten $\&$ Tibshirani(2009)}]{WitT2009}
%\textsc{Witten, D. M. $\&$ Tibshirani, R.} (2009).
%\newblock{Covariance regularized regression and classification for high-dimensional problems}.
%\newblock \textit{J. Roy. Stat. Soc. Ser. B, Stat. Meth}. {\bf 71} 615--636.

%
%\bibitem[{Witten et al.(2009)}]{Witten2009}
%\textsc{Witten, D. M., Tibshirani, R. $\&$ Hastie, T.} (2009).
%\newblock{A penalized matrix decomposition, with applications to sparse principal
%components and canonical correlation analysis}.
%\newblock \textit{Biostat}. {\bf 10} 515--534.

%\bibitem[{WuP(2003)}]{WuP2003}
%\textsc{Wu, W. B. $\&$ Pourahmadi, M.} (2003).
%\newblock{Nonparametric estimation of
%large covariance matrices of longitudinal data}.
%\newblock \textit{Biometrika}. {\bf 90} 831--844.

\bibitem[{Yin $\&$ Li(2011)}]{Yin2011}
\textsc{Yin, J. $\&$ Li, H.} (2011).
\newblock{A sparse conditional Gaussian graphical model for analysis of general genomic data}.
\newblock \textit {Annals of Applied Statistics}. {\bf 5} 2630--2650.


\bibitem[{Yuan(2008)}]{Yuan2008}
\textsc{Yuan, M.}(2008)
\newblock{Efficient computation of $\ell_1$ regularized estimates in Gaussian graphical models}.
\newblock \textit {Journal of Computation and Graphical Statistics}. {\bf 17} 809--826.


\bibitem[{Yuan $\&$ Lin(2006)}]{YuanLin2006}
\textsc{Yuan, M. $\&$ Lin, Y.} (2006).
\newblock{Model selection and estimation in regression with grouped variables}.
\newblock \textit {Journal of Royal Statistical Society: Series B}. {\bf 68} 49--67.

\bibitem[{Yuan $\&$ Lin(2007)}]{Yuan2007}
\textsc{Yuan, M. $\&$ Lin, Y.} (2007).
\newblock{Model selection and estimation in the Gaussian graphical model}.
\newblock \textit {Biometrika}. {\bf 94} 19--35.

\bibitem[{Yuan $\&$ Zhang(2014)}]{Tong2014}
\textsc{Yuan, X. $\&$ Zhang, T.} (2014).
\newblock{Partial Gaussian graphical model estimation}.
\newblock \textit {IEEE Transactions on Information Theory}. {\bf 60} 1673--1687.


%\bibitem[{Xia et al.(2002)}]{Xia2002}
%\textsc{Xia, Y., Tong, H., Li, W. K. $\&$ Zhu, L.-X.} (2002).
%\newblock{An adaptive estimation of dimension reduction space}.
%\newblock \textit{J. Roy. Statist. Soc. Ser. B}. {\bf 64} 363--410.

\bibitem[{Zhao $\&$ Yu(2006)}]{ZhaY2006}
\textsc{Zhao, P. $\&$ Yu, B.} (2006).
\newblock{On model selection consistency of Lasso}.
\newblock \textit{Journal of Machine Learning Research}. {\bf 7} 2541--2567.

\end{thebibliography}

\newpage
\section{Appendix}
\subsection{Proof of Theorem~\ref{theorem:main}}
\label{section:proof}
In this section, we prove the consistency results (stated in Theorem~\ref{theorem:main}) of the estimator \eqref{eqn:originalproblem_NS}. The high-level proof strategy is similar in spirit to the proof of the consistency results for sparse graphical model recovery \citep{RavWRY2008} and latent variable graphical model recovery \citep{Chand2012}. However, the estimator \eqref{eqn:originalproblem_NS} is different than the estimators proposed by \citep{RavWRY2008} and  \citep{Chand2012} due to the nuclear norm penalty corresponding to the SDR objective. \\

We begin by considering the following convex optimization program:

 \begin{eqnarray}
 \label{eqn:ConvexRelaxed4_N}
(\bar{\Theta}, \bar{S}_Y, \bar{L}_Y) = \argmin_{\substack{\Theta \in \Sp^{q+p}, ~\Theta \succ 0 \\ S_Y,{L}_Y \in \Sp^p}} & -\ell(\Theta; \{X^{(i)},Y^{(i)}\}_{i=1}^n) + \lambda_n [\delta \|S_Y\|_{\ell_1} + \|{L_Y}\|_\star + \gamma \|\Theta_{YX}\|_\star] \nonumber \\ \mathrm{s.t.} & \Theta_{Y} = S_Y - {L}_Y
\end{eqnarray}

{\noindent}Comparing \eqref{eqn:ConvexRelaxed4_N} with the convex program \eqref{eqn:originalproblem_NS}, the difference is that we no longer constrain ${L}_Y$ to be a positive semidefinite matrix. In particular, if ${L}_Y \succeq 0$, then the nuclear norm of the matrix ${L}_Y$ in the objective function of \eqref{eqn:ConvexRelaxed4_N} reduces to the trace of $L_Y$. We show that the unique optimum $(\bar{\Theta}, \bar{S}_Y, \bar{L}_Y)$ of \eqref{eqn:ConvexRelaxed4_N} has the property that with high probability, $\tilde{L}_Y$ is positive semidefinite. As a result, with high probability, the variables $(\bar{\Theta}, \bar{S}_Y, \bar{L}_Y)$ are also the optimum of \eqref{eqn:originalproblem_NS}. In the remainder of this section, we show that under the assumptions of Theorem ~\ref{theorem:main}, the primal feasible variables $(\bar{\Theta}, \bar{S}_Y, \bar{L}_Y)$ are structurally correct estimates of $(\Theta^\star, S_Y^\star, {L}_Y^\star)$ (see Definition~\ref{definition:1}). Below, we outline our proof strategy:

\begin{enumerate}
\item We proceed by analyzing \eqref{eqn:ConvexRelaxed4_N} with additional constraints that the variables $S_Y$, ${L}_Y$, and $\Theta_{YX}$ belong to the algebraic varieties of sparse and low-rank matrices (specified by the support of $S_Y^\star$ and rank of $L_Y^\star$ and $\Theta_{YX}^\star$) , and that the tangent spaces $\Omega(S_Y)$, $T(L_Y)$, $T(\Theta_{YX})$ are close to the nominal tangent spaces $\Omega(S_Y^\star)$, $T({L}_Y^\star)$, and $T(\Theta_{YX}^\star)$ respectively. We prove that under suitable conditions on the minimum magnitude nonzero entry of $S_Y^\star$, minimum nonzero singular value of ${L}_Y^\star$, and minimum nonzero singular value of $\Theta_{YX}^\star$, any optimum set of variables $(\Theta, S_Y, L_Y)$ of this non-convex program are smooth points of the underlying varieties; that is $\rm{sign}(S_Y) = \rm{sign}(S_Y^\star), \rm{rank}(L_Y) = \rm{rank}(L_Y^\star)$ and $\rm{rank}(\Theta_{YX}) = \rm{rank}(\Theta_{YX}^\star)$. Further, we show that $L_{Y}$ has the same inertia as $L_Y^\star$ so that $L_Y \succeq 0$. 

\item Conclusions of the previous step imply the the variety constraints can be ``linearized" at the global optimum of the non-convex program to obtain tangent-space constraints. Under suitable conditions on the regularization parameter $\lambda_n$, we prove that with high probability, the unique optimum of this ``linearized" program coincides with the global optimum of the non-convex program. 

\item Finally, we show that the tangent-space constraints of the linearized program are inactive at the optimum. Therefore, the optimal solution of \eqref{eqn:ConvexRelaxed4_N} has the property that with high probability: $\rm{sign}(\bar{S}_Y) = \rm{sign}(S_Y^\star)$, $\rm{rank}(\bar{L}_Y) = \rm{rank}(L_Y^\star)$, and $\rm{rank}(\bar{\Theta}_{YX}) = \rm{rank}(\Theta_{YX}^\star)$. Since $\bar{L}_Y\succeq 0$, we conclude that the variables $(\bar{\Theta}, \bar{S}_Y, \bar{L}_Y)$ are the unique optimum of \eqref{eqn:originalproblem_NS}. 
\end{enumerate}
In Section~\ref{section:proof1}, we prove the results of step 1. In Section~\ref{proof:VtoT}, we prove the results of step 2. Finally, in Section~\ref{sec:backOrig}, we prove the results of step 3.

\subsubsection{Variety Constrained Optimization Program}
\label{section:proof1}
Letting $m \triangleq \max\{ \frac{1}{\delta}, 1, \frac{1}{\gamma}\}$ and $\psi \triangleq \|(\Theta^\star)^{-1}\|_2$, we consider the following variety-constrained optimization program:
\begin{eqnarray}
({\Theta}^{\mathcal{M}}, {S}_Y^{\mathcal{M}}, {L}_Y^{\mathcal{M}}) = \argmin_{\substack{\Theta \in \Sp^{q+p}, ~\Theta \succ 0 \\ S_Y,{L}_Y \in \Sp^p}} & -\ell(\Theta; \{X^{(i)},Y^{(i)}\}_{i=1}^n) + \lambda_n [\delta \|S_Y\|_{\ell_1} + \|{L}_Y\|_{\star} + \gamma \|\Theta_{YX}\|_\star] \nonumber \\ \mathrm{s.t.} & \Theta_{Y} = S_Y - {L}_Y,  (\Theta, S_Y,  {L}_Y) \in \mathcal{M}. \label{eqn:NConveProblem_N}
\end{eqnarray}
{\noindent}Here, the set $\mathcal{M}$ is given by:
\begin{eqnarray*}
\mathcal{M} \triangleq \Big\{(\Theta, S_Y, {L}_Y) \in \Sp^{(p+q)} \times \Sp^p \times \Sp^{p} \Big | S_Y \in \Omega(S_Y^\star), \hspace{.1in} \text{rank}({L}_Y) \leq \text{rank}({L}_Y^\star) \\ \text{rank}(\Theta_{YX}) \leq \text{rank}(\Theta_{YX}^\star); \\
~~\|\mathcal{P}_{T({L}_Y^\star)^{\perp}}({L}_Y - {L}_Y^\star)\|_2 \leq \frac{\omega_Y\lambda_n}{2m\psi^2}, \hspace{.1in} \\
~\hspace{.5in}\|\mathcal{P}_{T(\Theta_{YX}^\star)^{\perp}}(\Theta_{YX} - \Theta_{YX}^\star)\|_2 \leq \frac{\omega_{YX}\lambda_n}{2m\psi^2} \\
\Phi_{\delta, \gamma}[\mathcal{A}^{\dagger}\mathbb{I}^{\star}\mathcal{A}\Delta] \leq 5\lambda_n \Big\}
\end{eqnarray*}

{\noindent}The optimization program \eqref{eqn:NConveProblem_N} is non-convex due to the rank constraints $\text{rank}(L_Y) \leq \text{rank}({L}_Y^\star)$ and $\text{rank}(\Theta_{YX}) \leq \text{rank}(\Theta_{YX}^\star)$ in the set $\mathcal{M}$. These constraints, in addition to the constraint $S_Y \in \Omega(S_Y^\star)$ ensure that the matrices $S_Y, {L}_Y$, and $\Theta_{YX}$ belong to appropriate varieties. The constraints in $\mathcal{M}$ along $T(L_Y^\star)^\perp$ and $T(\Theta_{YX}^\star)^\perp$ ensure that the tangent spaces $T({L}_Y)$ and $T(\Theta_{YX})$ are ``close'' to $T({L}_Y^\star)$ and $T(\Theta_{YX}^\star)$ respectively. Finally, the last condition roughly controls the error. In this section, we will prove that any feasible set of variables $(\Theta, S_Y, L_Y)$ -- and in particular an optimal set of variables $({\Theta}^{\mathcal{M}}, {S}_Y^{\mathcal{M}}, {L}_Y^{\mathcal{M}})$-- is structurally correct estimate of $(\Theta^\star, S_Y^\star, {L}_Y^\star)$. We begin by proving that any feasible set of variables $(\Theta, S_Y, {L}_Y)$ is ``close" in norm to the population quantities $(\Theta^\star, S_{Y}^\star, L_{Y}^\star)$.

\begin{proposition}
\label{prop:FirstResultCor}
Let $(\Theta, S_Y, {L}_Y)$ be a set of feasible variables of \eqref{eqn:NConveProblem_N}. Let $\Delta = (S_Y - S_Y^\star, {L}_Y - {L}_Y^\star, \Theta_{YX} - \Theta_{YX}^\star, \Theta_{X} - \Theta_{X}^\star)$ and $C_1 = \frac{12}{\alpha_{}} + \frac{1}{\psi^2}$. Then, $\Phi_{\delta, \gamma}[\Delta] \leq C_1\lambda_n$

\label{theorem:NonconvexImp}
\end{proposition}
\begin{proof}
Let ${\mathbb{H}}^\star = \Omega(S_Y^\star) \times T({L}_Y^\star) \times T(\Theta_{YX}^\star) \times \Sp^{q}$. Then, 
\begin{eqnarray*}
\Phi_{\delta, \gamma}[\mathcal{A}^{\dagger}\mathbb{I}^{\star}\mathcal{A}\mathcal{P}_{\mathbb{H}^\star}(\Delta)] &\leq& \Phi_{\delta, \gamma}[\mathcal{A}^{\dagger}\mathbb{I}^{\star}\mathcal{A}(\Delta_{})] + \Phi_{\delta, \gamma}[\mathcal{A}^{\dagger}\mathbb{I}^{\star}\mathcal{A}\mathcal{P}_{{\mathbb{H}^\star}^{\perp}}(\Delta)] \\
&\leq& 5\lambda_n + m\psi^2\Big(\frac{\omega_{Y}\lambda_n}{2m\psi^2} + \frac{\omega_{YX}\lambda_n}{2m\psi^2}\Big) \leq 6\lambda_n
\end{eqnarray*} 
Since $\Phi_{\delta, \gamma}[\mathcal{P}_{\mathbb{H}^\star}(\cdot)] \leq 2\Phi_{\delta, \gamma}(\cdot)$, we have that $\Phi_{\delta, \gamma}[\mathcal{P}_{\mathbb{H}^\star}\mathcal{A}^{\dagger}\mathbb{I}^{\star}\mathcal{A}\mathcal{P}_{\mathbb{H}^\star}(\Delta)] \leq 12\lambda_n$. Consequently, we appeal to the Fisher information Assumption 1 in \eqref{eqn:FirstFisherCond} to conclude that $\Phi_{\delta, \gamma}[\mathcal{P}_{\mathbb{H}^\star}(\Delta)] \leq \frac{12\lambda_n}{\alpha}$. Moreover:
\begin{eqnarray*}
\Phi_{\delta, \gamma}[\Delta]& \leq& \Phi_{\delta, \gamma}[\mathcal{P}_{\mathbb{H}^\star}(\Delta_{})] + \Phi_{\delta, \gamma}[\mathcal{P}_{{\mathbb{H}^\star}^{\perp}}(\Delta)] \leq \frac{12\lambda_n}{\alpha} + \frac{\lambda_n}{\psi^2} = C_1\lambda_n 
\end{eqnarray*}
\end{proof}

{\noindent}Proposition~\ref{theorem:NonconvexImp} leads to powerful implications. In particular, under additional conditions on the minimum magnitude nonzero entry of $S_Y^\star$, and minimum nonzero singular values of ${L}_Y^\star$ and $\Theta_{YX}^\star$, any feasible set of variables $(\Theta, S_Y, {L}_Y)$ of \eqref{eqn:NConveProblem_N} has two key properties: $(a)$ The variables $(\Theta_{YX}, S_Y, {L}_Y)$ are smooth points of the underlying varieties, $(b)$ The constraints in $\mathcal{M}$ along $T({L}_Y^\star)^{\perp}$ and $T(\Theta_{YX}^\star)^{\perp}$ are locally inactive at $\Theta_{YX}$ and $L_Y$. These properties, among others, are proved in the following corollary.

\begin{corollary}
\label{eqn:Corollary1}
Consider any feasible variables $( \Theta, S_Y, {L}_Y)$ of \eqref{eqn:NConveProblem_N}. Let $\sigma_{Y}$ be the smallest nonzero singular value of ${L}_Y^\star$, $\sigma_{YX}$ be the smallest nonzero singular value of $\Theta_{YX}^\star$, and $\tau_{Y}$ the minimum magnitude nonzero element of $S_Y^\star$. Let $\mathbb{H}' = \Omega(S_Y) \times T({L}_Y) \times T(\Theta_{YX}) \times \Sp^q$ and $C_{T'} = \mathcal{P}_{\mathbb{H}'^{\perp}} (0, {L}_Y^\star, \Theta_{YX}^\star ,0)$. Furthermore, let $C_1 = \frac{12}{\alpha_{}} + \frac{1}{\psi^2}$, $C_2 = \frac{4}{\alpha_{}} (1+\frac{1}{3\beta})$, $C_{\sigma_Y} = C_1^2\psi^2\max\{12\beta + 1, \frac{2}{C_2\psi^2} + 1\}$ and $C_{\sigma_{YX}}' = C_1^2\psi^2\max\{12\beta + \frac{6\beta}{\gamma}, \frac{2}{C_2\psi^2} + \frac{6\beta}{\gamma}\}$ . Suppose that the following inequalities are met: $\sigma_{Y} \geq  \frac{{m}}{\omega_Y}C_{\sigma_Y}\lambda_n$, $\sigma_{YX} \geq  \frac{m\gamma^2}{\omega_{YX}}C_{\sigma_{YX}}'\lambda_n$, and $\tau_{Y} \geq 2\delta{C}_1\lambda_n$.
Then,
\begin{enumerate}
\item ${L}_Y$ and $\Theta_{YX}$ are smooth points of their underlying varieties, i.e. $\rm{rank}({L}_Y) = \rm{rank}({L}_Y^\star)$, $\rm{rank}(\Theta_{YX}) = \rm{rank}(\Theta_{YX}^\star)$; Moreover ${L}_Y$ has the same inertia as ${L}_Y^\star$.
\item $\|\mathcal{P}_{T({L}_Y^\star)^{\perp}}({L}_Y - {L}_Y^\star)\|_2 \leq \frac{\lambda_n\omega_{Y}}{48m\psi^2}$ and $\|\mathcal{P}_{T(\Theta_{YX}^\star)^{\perp}}(\Theta_{YX} - \Theta_{YX}^\star)\|_2 \leq \frac{\lambda_n\omega_{YX}}{48m\psi^2}$
\item $\rho(T({L}_Y), T({L}_Y^\star)) \leq \omega_{Y}$, and $\rho(T(\Theta_{YX}), T(\Theta_{YX}^\star)) \leq \omega_{YX}$; that is, the tangent spaces at $L_Y$ and $\Theta_{YX}$ are ``close" to the tangent spaces at  $L_{Y}^\star$ and $\Theta_{YX}^\star$ respectively. 
\item $\Phi_{\delta, \gamma}[\mathcal{A}^{\dagger}\mathbb{I}^{\star}\mathcal{A} C_{T_{}'}] \leq \frac{\lambda_n}{6\beta}$
\item $\Phi_{\delta, \gamma}[C_{T_{}'}] \leq C_2\lambda_n$ 
\item $S$ is the smooth point of its underlying variety, i.e. $\rm{sign}(S_Y) = \rm{sign}(S_Y^\star)$\\
\end{enumerate}
\end{corollary}

\begin{proof}
We note the following relations before proving each step: $C_1 \geq \frac{1}{\psi^2} \geq \frac{1}{m\psi^2}$, $\omega_Y, \omega_{YX} \in [0,1]$, and $\beta \triangleq \frac{3-\nu}{\nu} \geq 8$ for $\nu \in (0,1/3]$. We also appeal to the results of \citep{Kato1995,Bach2008,Chand2012} regarding perturbation analysis of the low-rank matrix variety.\\

1. Based on the assumptions regarding the minimum nonzero singular values of ${L}_Y^\star$ and $\Theta_{YX}^\star$, we have:
\begin{eqnarray*}
\sigma_{Y} &\geq& \frac{ C_1^2\lambda_n}{\omega_{Y}} m\psi^2({12\beta+1})
\geq \frac{ C_1\lambda_n}{\omega_{Y}}({12\beta+1})\geq (12\beta+1)C_1\lambda_n \geq 8C_1\lambda_n \geq 8\|L - L_{Y}^\star\|_2\\
\sigma_{YX} &\geq& \frac{C_1^2\lambda_n}{\omega_{YX}} \gamma^2m\psi^2{\Big(\frac{6\beta}{\gamma}+12\beta\Big)} 
\geq { C_1^2\lambda_n^{}} \gamma^2m_{}\psi^2\frac{6\beta}{\gamma} \geq 8\gamma{C}_1^{}\lambda_n^{} \geq 8\|{\Theta}_{YX} - \Theta_{YX}^\star\|_2
\end{eqnarray*}

{\noindent}Combing these results and Proposition~\ref{theorem:NonconvexImp}, we conclude that ${L}_Y$ and $\Theta_{YX}$ are smooth points of their respective varieties, i.e. $\text{rank}({L}_Y)= \text{rank}({L}_Y^\star)$, and $\text{rank}(\Theta_{YX}) = \text{rank}(\Theta_{YX}^\star)$. Furthermore, ${L}_Y$ has the same inertia as ${L}_Y^\star$.

2. Since $\sigma_{Y} \geq 8\|L_Y - L_{Y}^\star\|_2$, and $\sigma_{YX} \geq 8\|\Theta_{YX} - \Theta_{YX}^\star\|_2$, we can appeal to Proposition 2.2 in (Chandrasekaran et al., 2012) to conclude that the constraints in $\mathcal{M}$ along $\mathcal{P}_{T({L}_Y^\star)^{\perp}}$ and $\mathcal{P}_{T(\Theta_{YX}^\star)^{\perp}}$ are strictly feasible: 
\begin{eqnarray*}
\|\mathcal{P}_{T({L}_Y^\star)^{\perp}}({L}_Y - {L}_Y^\star)\|_2 &\leq& \frac{\|{L}_Y - {L}_Y^\star\|_2^2}{\sigma_{Y}} 
= \frac{C_1^2\lambda_n^2\omega_{Y}}{C_1^2\lambda_n{m}\psi^2(12\beta+1)}\leq \frac{\lambda_n}{48m\psi^2} \\
\|\mathcal{P}_{T(\Theta_{YX}^\star)^{\perp}}(\Theta_{YX} - \Theta_{YX}^\star)\|_2 &\leq& \frac{\|\Theta_{YX} - \Theta_{YX}^\star\|_2^2}{\sigma_{YX}} 
= \frac{C_1^2\gamma^2\lambda_n^2\omega_{YX}}{C_1^2\lambda_n{m}\psi^2\gamma^2\Big(\frac{6\beta}{\gamma}+12\beta\Big)} \leq \frac{\lambda_n}{48m\psi^2}
\end{eqnarray*}

3. Appealing to Proposition 2.1 in (Chandrasekaran et al., 2012), we prove that the tangent spaces $T(L_Y)$ and $T(\Theta_{YX})$ are close to $T(L_Y^\star)$ and $T(\Theta_{YX}^\star)$ respectively:
\begin{eqnarray*}
\rho(T({L}_Y), T({L}_Y^\star)) &\leq& \frac{2\| {L}_Y - {L}_Y^\star \|_2}{\sigma_{Y}}
\leq \frac{2C_1\lambda_n \omega_{Y}}{C_1^2\lambda_nm{\psi}^2(12\beta+1)} 
\leq \omega_{Y}\\
\rho(T(\Theta_{YX}), T(\Theta_{YX}^\star)) &\leq& \frac{2\| \Theta_{YX} - \Theta_{YX}^\star\|_2}{\sigma_{YX}}
\leq \frac{2\lambda_nC_1 \gamma\omega_{YX}}{C_1^2\lambda_n{m}\psi^2\gamma^2\frac{6\beta}{\gamma}} 
\leq \omega_{YX}
\end{eqnarray*}

4. Letting $\sigma_{Y}'$ and $\sigma_{YX}'$ be the minimum nonzero singular value of ${L}$ and $\Theta_{YX}$ respectively, we note that:
\begin{eqnarray*}
\sigma_{Y}' \geq \sigma_{Y} - \|{L}_Y - {L}_Y^\star\|_2 &\geq& \frac{C_1^2\lambda_n}{\omega_{Y}} m\psi^2(12\beta+1)-C_1\lambda_n \\ &\geq& \frac{C_1^2\lambda_n}{\omega_{Y}}12 m\psi^2\beta \geq 8C_1\lambda_n \geq 8\|{L}_Y - L_Y^\star\|_2 \\[.1in]
\sigma_{YX}' \geq \sigma_{YX} - \|\Theta_{YX} - \Theta_{YX}^\star\|_2 &\geq& \frac{C_1^2\lambda_n\gamma^2}{\omega_{YX}} m\psi^2\Big(\frac{6\beta}{\gamma}+12\beta\Big)-C_1\lambda_n\gamma \\ &\geq& \frac{6C_1^2\lambda_n\gamma^2m\psi^2\beta}{\omega_{YX}\gamma} - C_1\gamma\lambda_n \geq 8C_1\lambda_n\gamma \geq 8\|\Theta_{YX} - {\Theta}_{YX}^\star\|_2
\end{eqnarray*}
Once again appealing to Proposition 2.2 in (Chandrasekaran et al., 2012), we have:
\begin{eqnarray*}
\Phi_{\delta, \gamma}(C_{T'}) &\leq& m \|\mathcal{P}_{T({L_Y})^{\perp}} ({L}_Y - {L}_Y^\star)\|_2 + m \|\mathcal{P}_{T(\Theta_{YX})^{\perp}} (\Theta_{YX} - \Theta_{YX}^\star)\|_2 \\ &\leq& m\frac{\|{L}_Y - {L}_Y^\star\|_2^2}{\sigma_{Y}'} + m\frac{\|\Theta_{YX} - \Theta_{YX}^\star\|_2^2}{\sigma_{YX}'}\\
&\leq& \frac{mC_1^2\lambda_n^2}{\frac{C_1^2\lambda_nm\psi^2(12\beta+1)}{\omega_Y} - C_1\lambda_n} + \frac{mC_1^2\lambda_n^2\gamma^2}{\frac{C_1^2\lambda_n\gamma^2m\psi^2\Big(\frac{6\beta}{\gamma}+12\beta\Big)}{\omega_{YX}} - C_1\lambda_n\gamma}\\
&\leq& \frac{\lambda_n\omega_{Y}}{12\beta\psi^2} + \frac{\lambda_n\omega_{YX}}{12\beta\psi^2}  \leq\frac{\lambda_n}{6\beta\psi^2}
\end{eqnarray*}

This leads to the result that $\Phi_{\delta, \gamma}[\mathcal{A}^{\dagger}\mathbb{I}^{\star}\mathcal{A}C_{T'}] \leq \frac{\lambda_n}{6\beta}$.\\

5. Following the same reasoning as step 4, we conclude:
\begin{eqnarray*}
\Phi_{\delta, \gamma}(C_{T'}) &\leq& m\frac{\|{L}_Y - {L}_Y^\star\|_2^2}{\sigma_{Y}'} + m\frac{\|\Theta_{YX} - \Theta_{YX}^\star\|_2^2}{\sigma_{YX}'}\\
&\leq& \frac{mC_1^2\lambda_n^2}{\frac{C_1^2\lambda_nm\psi^2(\frac{2}{C_2\psi^2}+1)}{\omega_Y} - C_1\lambda_n} + \frac{mC_1^2\lambda_n^2\gamma^2}{\frac{C_1^2\lambda_n\gamma^2m\psi^2(\frac{6\beta}{\gamma}+\frac{2}{C_2\psi^2})}{\omega_{YX}} - C_1\lambda_n\gamma}\\
&\leq&   \frac{C_2\omega_{Y}}{2}\lambda_n +  \frac{C_2\omega_{YX}}{2}\lambda_n \leq C_2\lambda_n
\end{eqnarray*}

6. This fact follows immediately since $\|S_Y - S_Y^\star\|_{\infty} \leq \delta{C}_1\lambda_n$ and the smallest nonzero magnitude entry of $S_Y^\star$ is greater than $2\delta{C}_1\lambda_n$. 
\end{proof}

\subsubsection{From Variety Constraints to Tangent-Space Constraints}
\label{proof:VtoT}
Consider any optimal solution $(\Theta^{\mathcal{M}}, S_Y^{\mathcal{M}}, {L}_Y^{\mathcal{M}})$ of \eqref{eqn:NConveProblem_N}. In Corollary~\ref{eqn:Corollary1}, we concluded that the variables $(\Theta^{\mathcal{M}}_{YX}$, $S_Y^{\mathcal{M}}, {L}_Y^{\mathcal{M}})$ are smooth points of their respective varieties. As a result, the rank constraints  $\rm{rank}({L}_Y) \leq \rm{rank}({L}_Y^\star)$ and $\rm{rank}(\Theta_{YX}) \leq \rm{rank}(\Theta_{YX}^\star)$ can be ``linearized" to ${L}_Y \in T(L_Y^{\mathcal{M}})$ and $\Theta_{YX} \in T(\Theta^{\mathcal{M}}_{YX})$ respectively. Since all the remaining constraints are convex, the optimum of this linearized program is also the optimum of \eqref{eqn:NConveProblem_N}. Moreover, we once more appeal to Corollary~\ref{eqn:Corollary1} to conclude that the constraints in $\mathcal{M}$ along $\mathcal{P}_{T({L}_Y^\star)^{\perp}}$ and $\mathcal{P}_{T(\Theta_{YX}^\star)^{\perp}}$ are strictly feasible at $(\Theta^{\mathcal{M}}, S_Y^{\mathcal{M}}, {L}_Y^{\mathcal{M}})$. As a result, these constraints are locally inactive and can be removed without changing the optimum. Finally, we claim that the constraint $\Phi_{\delta, \gamma}[\mathcal{A}^{\dagger}\mathbb{I}^{\star}\mathcal{A}\Delta] \leq 5\lambda_n$ in \eqref{eqn:NConveProblem_N} can also removed in this ``linearized" convex program. In particular, letting $\mathbb{H}_{\mathcal{M}} \triangleq \Omega(S_Y^\star) \times T({L}_Y^{\mathcal{M}}) \times T(\Theta^{\mathcal{M}}_{YX}) \times \Sp^q$ (note $\mathbb{H}_{\mathcal{M}} \in U(\omega_Y, \omega_{YX})$ where the set $U(\omega_Y, \omega_{YX})$ is defined in \eqref{eqn:Udef}), consider the following convex optimization program with the constraint $\Phi_{\delta, \gamma}[\mathcal{A}^{\dagger}\mathbb{I}^{\star}\mathcal{A}\Delta] \leq 5\lambda_n$ removed :
\begin{eqnarray}
(\tilde{\Theta}, \tilde{S}_Y, \tilde{L}_Y) = \argmin_{\substack{\Theta \in \Sp^{q+p}, ~\Theta \succ 0 \\ S_Y,{L}_Y \in \Sp^p}} & -\ell(\Theta; \{X^{(i)},Y^{(i)}\}_{i=1}^n) + \lambda_n [\delta \|S_Y\|_{\ell_1} + \|{L_Y}\|_{\star} + \gamma \|\Theta_{YX}\|_\star] \nonumber \\ \mathrm{s.t.} & \Theta_{Y} = S_Y - {L}_Y, \hspace{.1in} (S_Y, {L}_Y, \Theta_{YX}, \Theta_{X}) \in \mathbb{H}_{\mathcal{M}} \label{eqn:ConvexRelaxed2_N}
\end{eqnarray}

{\noindent}In the following theorem, we prove that under additional conditions on the regularization parameter $\lambda_n$ of \eqref{eqn:ConvexRelaxed2_N}
 , the set of variables $(\Theta^{\mathcal{M}}, S_Y^{\mathcal{M}}, {L}_Y^{\mathcal{M}})$ is the unique optimum of \eqref{eqn:ConvexRelaxed2_N}. \begin{proposition} 
\label{proposition:lambdaBound}
Let $C' = (2+\delta\rm{deg}(S_Y^\star)+\gamma)\psi$, $C_1 = \frac{12}{\alpha_{}} + \frac{1}{\psi^2}$, $C_2 = \frac{4}{\alpha_{}} (1+\frac{1}{3\beta})$, and $C_{samp}'  = \max\Big\{ \frac{1}{48\beta\psi}, 4{C_2}{C'}, \frac{32m{\psi}C'^2C_2}{\alpha}, {12{\beta}m{\psi}C'^2{C_1^2} }\Big\}$. Suppose that the number of observed samples obeys $n \geq 4608 \beta^2m^2\psi^2C_{samp}'^2(p+q)$, and the regularization parameter $\lambda_n$ is chosen in the following range: $\lambda_n \in \Big[6\beta\sqrt{\frac{128(p+q)m^2\psi^2}{n}}, \frac{1}{C_{samp}'}\Big]$. Then, with probability greater than $1 - 2\text{exp}\Big\{-\frac{n\lambda_n^2}{4608\beta^2m^2\psi^2}\Big\}$,  $(\tilde{\Theta}, \tilde{S}_Y, \tilde{L}_Y) = ({\Theta}^{\mathcal{M}}, {S}_Y^{\mathcal{M}}, {L}_Y^{\mathcal{M}})$.
\end{proposition}
\begin{proof} 
The high-level proof strategy is to show that the constraint  $\Phi_{\delta, \gamma}[\mathcal{A}^{\dagger}\mathbb{I}^{\star}\mathcal{A}\Delta] \leq 5\lambda_n$ is inactive at the optimum of \eqref{eqn:ConvexRelaxed2_N}. That is, we show that $\Phi_{\delta, \gamma}[\mathcal{A}^{\dagger}\mathbb{I}^{\star}\mathcal{A}(\tilde{S}_Y - S_Y^\star, \tilde{L}_Y - {L}_Y^\star, \tilde{\Theta}_{YX} - \Theta_{YX}, \tilde{\Theta}_{X} - \Theta_{X}^\star)] < 5\lambda_n$. The proof of this fact relies on the results of the following lemmas:
\begin{lemma}
\label{lemma:Remainder}
Let $\Delta = (\tilde{S}_Y - S_Y^\star, \tilde{L}_Y - {L}_Y^\star, \tilde{\Theta}_{YX} - \Theta_{YX}^\star, \tilde{\Theta}_{X} - \Theta_{X}^\star)$. Denote \\$R_{\Sigma^\star}(\Delta) \triangleq \Sigma^\star\Big[\sum_{k = 2}^{\infty} (-\mathcal{A}(\Delta){\Sigma^\star}^{-1})^k \Big]$. If $\Phi_{\delta, \gamma}[\Delta] \leq \frac{1}{2C'}$, then: $\Phi_{\delta, \gamma}[\mathcal{A}^{\dagger}R_{\Sigma^\star}(\Delta)] \leq 2m{\psi}C'^2 \Phi_{\delta,\gamma}[\Delta]^2$.
\end{lemma}
\begin{proof}
We begin by introducing a quantity that plays an important role in the proof of this lemma and was also employed in \citet{Chand2012}. Given a symmetric $p \times p$ matrix $M$, we define the quantity $\mu(\Omega(M))$ with respect to the tangent space $\Omega(M)$ of the variety of $p \times p$ matrices with at most $|\text{support}(M)|$ nonzero entries:
\begin{eqnarray}
\mu(\Omega(M)) \triangleq \max_{N \in \Omega(M), \|N\|_{\ell_\infty} = 1} \|N\|_{2}
\label{eqn:mu}
\end{eqnarray}

{\noindent}
One can show that a sparse matrix $M$ with ``bounded degree" (a small number of non zeros per row/column) has small $\mu(M)$. Specifically, for any $p \times p$ matrix $M$, we have $\mu(\Omega(M)) \leq \rm{deg}(M)$ where $\rm{deg}(M)$ is equal to the maximum number of non zeros in any column/row of $M$ \citep{Chand2012}. We now proceed with the proof:
\begin{eqnarray*}
\Phi_{\delta,\gamma}[\mathcal{A}^{\dagger}R_{\Sigma^\star}(\Delta)] &\leq& m\psi \Big[\sum_{k = 2}^{\infty} (\psi \|\Delta\|_2)^k \Big] \
 \leq m\psi \sum_{k = 2}^{\infty} \psi^k\Big(\delta\mu(\Omega(S_Y^\star)) \frac{\|{\Delta}S_Y\|_{\ell_{\infty}}}{\delta}  + \|{\Delta}{L}_Y\|_2 \\&+&\gamma\frac{\|{\Delta}\Theta_{YX}\|_2}{\gamma} + \|{\Delta}\Theta_{X}\|_2\Big)^k\\
 &\leq&m\psi \sum_{k = 2}^{\infty} \psi^k\Big(\delta\text{deg}(S_Y^\star) \frac{\|{\Delta}S_Y\|_{\ell_{\infty}}}{\delta}  + \|{\Delta}{L_Y}\|_2 \\ &+&\gamma\frac{\|{\Delta}\Theta_{YX}\|_2}{\gamma} + \|{\Delta}\Theta_{X}\|_2\Big)^k\\
&\leq& m\psi^3 \frac{(2+\delta\text{deg}(S_Y^\star)+\gamma)^2\Phi_{\delta,\gamma}[\Delta]^2}{1-(2+\delta\text{deg}(S_Y^\star)+\gamma)\Phi_{\delta,\gamma}[\Delta]\psi}
\leq 2m{\psi}C'^2 \Phi_{\delta,\gamma}[\Delta]^2
\end{eqnarray*}
Note that the second inequality employs the property that $(\tilde{S}_Y - S_Y^\star) \in \Omega(S_Y^\star)$ and the quantity defined in \eqref{eqn:mu}. 

\end{proof}
\begin{lemma}
\label{lemma:Brower}
Let $C_{T_{\mathcal{M}}} = \mathcal{P}_{{\mathbb{H}}_{\mathcal{M}}^{\perp}} (0,{L}_Y^\star, \Theta_{YX}^\star, 0)$. Furthermore, let $E_n = \Sigma_n - \Sigma^\star$, and $\Delta = (\tilde{S}_Y - S_Y^\star, \tilde{L}_Y - {L}_Y^\star, \tilde{\Theta}_{YX} - \Theta_{YX}^\star, \tilde{\Theta}_{X} - \Theta_{X}^\star)$. Finally, define:
\begin{eqnarray}
r = \max\Big\{\frac{4}{\alpha}\Big(\Phi_{\delta,\gamma}[\mathcal{A}^{\dagger}E_n + \mathcal{A}^{\dagger}\mathbb{I}^{\star}\mathcal{A}C_{T_{\mathcal{M}}}] +\lambda_n\Big),\hspace{.1in} \Phi_{\delta, \gamma}[C_{T_{\mathcal{M}}}]\Big\}
\label{eqn:rdef}
\end{eqnarray}

{\noindent}If $r \leq \min\{\frac{1}{4C'}, \frac{\alpha}{32m{\psi}C'^2}\}$, then $\Phi_{\delta, \gamma}[\Delta] \leq 2r$.  
\end{lemma}

\begin{proof}
{\noindent} The proof of this result uses Brouwer's fixed-point theorem, and is inspired by the proof of a similar result in \citep{RavWRY2008,Chand2012}. 
The optimality conditions of \eqref{eqn:ConvexRelaxed2_N} suggest that there exist Lagrange multipliers $Q_{\Omega} \in \Omega(S_Y^\star)^{\perp}$,  $Q_{T_{Y}} \in T(\tilde{L}_Y)^{\perp}$, and $Q_{T_{YX}} \in T(\tilde{\Theta}_{YX})^{\perp}$ such that
\begin{eqnarray*}
[\Sigma_n - {\tilde{\Theta}}^{-1}]_{Y} + Q_{\Omega} \in -\lambda_n\delta\partial \|\tilde{S}_Y\|_{\ell_1}; \hspace{.1in} [\Sigma_n - {\tilde{\Theta}}^{-1}]_{Y} + Q_{T_{Y}}  &\in& \lambda_n\partial\|\tilde{L}_Y\|_\star\\[.01in]
[\Sigma_n - {\tilde{\Theta}}^{-1}]_{YX} + Q_{T_{YX}}  \in -\lambda_n\gamma\partial\|\tilde{\Theta}_{YX}\|_{\star};  \hspace{.1in} [\Sigma_n - {\tilde{\Theta}}^{-1}]_{X}  = 0
\end{eqnarray*}

{\noindent}Letting the SVD decomposition of $\tilde{L}$ and $\tilde{\Theta}_{YX}$ be given by $\tilde{L}_Y = \bar{U}\bar{D}\bar{V}'$ and $\tilde{\Theta}_{YX} = \breve{U}\breve{D}{\breve{V}}'$ respectively, we can restrict the optimality conditions to the space $\mathbb{H}_{\mathcal{M}}$ to obtain:
\begin{equation}
\mathcal{P}_{{\mathbb{H}}_{\mathcal{M}}}\mathcal{A}^{\dagger}(\Sigma_n - \tilde{\Theta}^{-1}) = (-\lambda_n\delta\text{sign}(\tilde{S}_Y), \hspace{.1in} \lambda_n\bar{U}\bar{V}', \hspace{.1in}  -\lambda_n\gamma_{}{\breve{U}}{\breve{V}}',  \hspace{.1in} 0)  
\label{eqn:optcond}
\end{equation}
{\noindent} Based on the Fisher information Assumption \eqref{eqn:FirstFisherCond}, the optimum of \eqref{eqn:ConvexRelaxed2_N} is unique (this is because the Hessian of the negative log-likelihood term is positive definite restricted to the tangent space constraints). Moreover, using standard Lagrangian duality, one can show that the set of variables  $(\tilde{\Theta}, \tilde{S}_Y, \tilde{L}_Y)$ that satisfy \eqref{eqn:optcond} are unique. The matrix inversion lemma allows one to express $\tilde{\Theta}^{-1}$ equivalently by:
\begin{eqnarray*}
\tilde{\Theta}^{-1} = [\Theta^\star + \mathcal{A}(\Delta)]^{-1} = \Sigma^\star  - R_{\Sigma^\star}(\Delta) + \mathbb{I}^{\star}\mathcal{A}(\Delta)
\end{eqnarray*}
Setting $Z \triangleq  (-\lambda_n\delta\text{sign}(\tilde{S}), \hspace{.1in} \lambda_n\bar{U}\bar{V}', \hspace{.1in}  -\lambda_n\gamma_{}{\breve{U}}{\breve{V}}',  \hspace{.1in} 0)$, relation \eqref{eqn:optcond} can be restated as:
\begin{equation}
\mathcal{P}_{{\mathbb{H}}_{\mathcal{M}}}\mathcal{A}^{\dagger}(E_{n} - R_{\Sigma^\star}(\Delta) + \mathbb{I}^{\star}\mathcal{A}(\Delta)) = Z
\label{eqn:Optcond}
\end{equation}
{\noindent}Notice that $\Phi_{\delta, \gamma}(Z) = \lambda_n$. We now appeal to Brouwer's fixed-point theorem to bound $\Phi_{\delta, \gamma}[\Delta]$. Consider the following function $G(\underline{\delta})$ with $\underline{\delta} \in U(\omega_{Y}, \omega_{YX})$:
\begin{eqnarray*}
G(\underline{\delta}) &=& \underline{\delta} - (\mathcal{P}_{\mathbb{H}_{\mathcal{M}}}\mathcal{A}^{\dagger}\mathbb{I}^{\star}\mathcal{A}\mathcal{P}_{\mathbb{H}_{\mathcal{M}}})^{-1}\Big(\mathcal{P}_{\mathbb{H}_{\mathcal{M}}}\mathcal{A}^{\dagger} [E_{n} - R_{\Sigma^\star}\mathcal{A}(\underline{\delta} + C_{T_{\mathcal{M}}}) + \mathbb{I}^{\star}\mathcal{A}(\underline{\delta} + C_{T_{\mathcal{M}}})] - Z\Big) 
\end{eqnarray*}

{\noindent}Note that the function $G(\underline{\delta})$ is well-defined since the operator $\mathcal{P}_{\mathbb{H}_{\mathcal{M}}}\mathcal{A}^{\dagger}\mathbb{I}^{\star}\mathcal{A}\mathcal{P}_{\mathbb{H}_{\mathcal{M}}}$ is bijective due to Fisher information Assumption 1 in \eqref{eqn:FirstFisherCond}. As a result, $\underline{\delta}$ is a fixed point of $G(\underline{\delta})$ if and only if $\mathcal{P}_{\mathbb{H}_{\mathcal{M}}}\mathcal{A}^{\dagger} [E_{n} - R_{\Sigma^\star}\mathcal{A}(\underline{\delta} + C_{T_{\mathcal{M}}}) + \mathbb{I}^{\star}\mathcal{A}(\underline{\delta} + C_{T_{\mathcal{M}}})] =  Z$. 
Since the variables $(\tilde{\Theta}, \tilde{S}_Y, \tilde{L}_Y)$ are the unique solution to \eqref{eqn:ConvexRelaxed2_N}, the only fixed point of $G$ is $\mathcal{P}_{\mathbb{H}_{\mathcal{M}}}[\Delta]$. Next we show that this unique optimum lives inside the ball $\mathbb{B}_r = \{\underline{\delta} \hspace{.1in} | \hspace{.1in} \Phi_{\delta, \gamma}(\underline{\delta}) \leq r,  \hspace{.05in} \underline{\delta} \in \mathbb{H}_{\mathcal{M}}\}$. In particular, if we show that under the map $G$, the image of $\mathbb{B}_r$ lies in $\mathbb{B}_r$, we can appeal to Brouwer's fixed point theorem to conclude that $\mathcal{P}_{\mathbb{H}_{\mathcal{M}}}[\Delta] \in \mathbb{B}_r$. For $\underline{\delta} \in \mathbb{B}_r$, $\Phi_{\delta, \gamma}[G(\underline{\delta})]$ can be bounded as follows:
\begin{eqnarray*}
\Phi_{\delta, \gamma}[G(\underline{\delta})] &=& \Phi_{\delta, \gamma}\Big[ \underline{\delta} -(\mathcal{P}_{\mathbb{H}_{\mathcal{M}}}\mathcal{A}^{\dagger}\mathbb{I}^{\star}\mathcal{A}\mathcal{P}_{\mathbb{H}_{\mathcal{M}}})^{-1}\Big(\mathcal{P}_{\mathbb{H}_{\mathcal{M}}}\mathcal{A}^{\dagger} [E_{n} - R_{\Sigma^\star}\mathcal{A}(\underline{\delta} + C_{T_{\mathcal{M}}}) \\ &+&  \mathbb{I}^{\star}\mathcal{A}(\underline{\delta} + C_{T_{\mathcal{M}}})] - Z\Big)\Big] \\
 &=& \Phi_{\delta, \gamma}\Big[(\mathcal{P}_{\mathbb{H}_{\mathcal{M}}}\mathcal{A}^{\dagger}\mathbb{I}^{\star}\mathcal{A}\mathcal{P}_{\mathbb{H}_{\mathcal{M}}})^{-1}\Big(\mathcal{P}_{\mathbb{H}_{\mathcal{M}}}\mathcal{A}^{\dagger} [E_{n} - R_{\Sigma^\star}A(\underline{\delta} + C_{T_{\mathcal{M}}}) \\&+& \mathbb{I}^{\star}\mathcal{A}C_{T_{\mathcal{M}}}] - Z\Big)\Big] \\
&\leq& \frac{1}{\alpha} \Phi_{\delta, \gamma} \Big[\mathcal{P}_{{\mathbb{H}}_{\mathcal{M}}}\mathcal{A}^{\dagger}(E_n - R_{\Sigma^\star}\mathcal{A}(\underline{\delta} + C_{T_{\mathcal{M}}}) + \mathbb{I}^{\star}\mathcal{A}(C_{T_{\mathcal{M}}})) - Z\Big] \\
&\leq& \frac{2}{\alpha}\Big[\Phi_{\delta, \gamma}[\mathcal{A}^{\dagger}( E_n + \mathbb{I}^{\star}\mathcal{A}(C_{T_{\mathcal{M}}}))] +  \lambda_n\Big] + \frac{2}{\alpha}{\Phi_{\delta, \gamma}[\mathcal{A}^{\dagger}R_{\Sigma^\star}(\underline{\delta})}] \\
&\leq& \frac{r}{2} +  \frac{2}{\alpha}{\Phi_{\delta, \gamma}[\mathcal{A}^{\dagger}R_{\Sigma^\star}(\underline{\delta})}]
\end{eqnarray*}

{\noindent}The first inequality holds because of Fisher information Assumption 1 in \eqref{eqn:FirstFisherCond}. The second inequality uses the property $\Phi_{\delta, \gamma}[\mathcal{P}_{\mathbb{H}_{\mathcal{M}}}(.)] \leq 2\Phi_{\delta, \gamma}(.)$ and $\Phi_{\delta, \gamma}(Z) = \lambda_n$. Moreover, since $r \leq \frac{1}{4C'}$, we have $\Phi_{\delta, \gamma}(\underline{\delta} + C_{T_{\mathcal{M}}}) \leq \Phi_{\delta, \gamma}(\underline{\delta}) + \Phi_{\delta, \gamma}(C_{T_{\mathcal{M}}}) \leq 2r \leq \frac{1}{2C'}$. We can now appeal to Proposition 1 to obtain:
\begin{eqnarray*}
\frac{2}{\alpha}{\Phi_{\delta, \gamma} [\mathcal{A}^{\dagger}R_{\Sigma^\star}(\underline{\delta} + C_{T_{\mathcal{M}}})}] &\leq& \frac{4}{\alpha}m{\psi}C'^2 [\Phi_{\delta, \gamma}(\underline{\delta} + C_{T_{\mathcal{M}}})]^2 \\
 &\leq& \frac{16m{\psi}}{\alpha}C'^2r^2 = \Big[\frac{16m{\psi}}{\alpha}C'^2r_{}\Big]r_{} \leq \frac{r}{2}
\end{eqnarray*}
\par Thus, we conclude that  $\Phi_{\delta, \gamma}
[G(\underline{\delta})] \leq r_{}$ and by Brouwer's fixed-point theorem, $\Phi_{\delta, \gamma} [\mathcal{P}_{\mathbb{H}_{\mathcal{M}}}(\Delta)] \leq r$. Furthermore,
\begin{eqnarray*}
\Phi_{\delta, \gamma}[\Delta] \leq \Phi_{\delta, \gamma}[\mathcal{P}_{\mathbb{H}_{\mathcal{M}}}(\Delta)] + \Phi_{\delta, \gamma}(C_{T_{\mathcal{M}}}) \leq 2r
\end{eqnarray*} 
\end{proof}

\begin{lemma}
\label{prop:Enbound}
Suppose that the number of observed samples obeys $n \geq 4608 \beta^2m^2\psi^2C_{samp}'^2(p+q)$, and the regularization parameter $\lambda_n$ is chosen in the following range: $\lambda_n \in \Big[6\beta\sqrt{\frac{128(p+q)m^2\psi^2}{n}}, \frac{1}{C_{samp}'}\Big]$. Then, with probability greater than $1 - 2\text{exp}\Big\{-\frac{n\lambda_n^2}{4608\beta^2m^2\psi^2}\Big\}$, $\Phi_{\delta, \gamma}[\mathcal{A}^{\dagger}E_n] \leq \frac{\lambda_n}{6\beta}$. \end{lemma}
\begin{proof}
First, note that $\Phi_{\delta, \gamma}[\mathcal{A}^{\dagger}E_n] \leq m \|\Sigma_n - \Sigma^\star\|_2$. Using the results in \citep{DavidSzarek2001} and the fact that $\frac{\lambda_n}{6\beta} \leq 8\psi$ and $n \geq \frac{2304(p+q)m^2\psi^2}{\lambda_n^2}$, the following bound holds: $\text{Pr}[ m\|\Sigma_n - \Sigma^\star\|_2 \geq \frac{\lambda_n}{6\beta}] \leq 2\text{exp} \Big\{\frac{-n\lambda_n^2}{4608m^2\psi^2}\Big\}$. Thus, $\Phi_{\delta, \gamma}[\mathcal{A}^{\dagger}E_n] \leq \frac{\lambda_n}{6\beta}$ with probability greater than  $1 - 2\text{exp}\Big\{-\frac{n\lambda_n^2}{4608\beta^2m^2\psi^2}\Big\}$.
\end{proof}

{\noindent}\underline{\textbf{Proof of Proposition 2}}: We now proceed with completing the proof of this theorem. In particular, we show that $\Phi_{\delta, \gamma}[\mathcal{A}^{\dagger}\mathbb{I}^{\star}\mathcal{A}(\tilde{S}_Y - S_Y^\star, \tilde{L}_Y - {L}_Y^\star, \tilde{\Theta}_{YX} - \Theta_{YX}^\star, \tilde{\Theta}_{X} - \Theta_{X}^\star)] < 5\lambda_n$. Based on the optimality condition \eqref{eqn:Optcond} and the property that $\Phi_{\delta, \gamma}[\mathcal{P}_{\mathbb{H}_{\mathcal{M}}}( . )] \leq 2\Phi_{\delta, \gamma}(.)$, we have:

\begin{eqnarray}
\Phi_{\delta, \gamma}[\mathcal{P}_{\mathbb{H}_{\mathcal{M}}}\mathcal{A}^{\dagger}\mathbb{I}^{\star}\mathcal{A}\mathcal{P}_{\mathbb{H}_{\mathcal{M}}}(\Delta)] &\leq& 2\lambda_n + \Phi_{\delta, \gamma}[\mathcal{P}_{\mathbb{H}_{\mathcal{M}}}\mathcal{A}^{\dagger}R_{\Sigma^\star}(\Delta)] + \Phi_{\delta, \gamma} [\mathcal{P}_{\mathbb{H}_{\mathcal{M}}} \mathcal{A}^{\dagger}\mathbb{I}^{\star}\mathcal{A}C_{T_{\mathcal{M}}}]\nonumber\\ &+& \Phi_{\delta, \gamma}[\mathcal{P}_{\mathbb{H}_{\mathcal{M}}} \mathcal{A}^{\dagger}E_{n}] \nonumber\\
&\leq& 2\lambda_n + 2\Phi_{\delta, \gamma}[\mathcal{A}^{\dagger}R_{\Sigma^\star}(\Delta)] + 2\Phi_{\delta, \gamma}[\mathcal{A}^{\dagger}\mathbb{I}^{\star}\mathcal{A} C_{T_{\mathcal{M}}}] \nonumber\\ &+& 2\Phi_{\delta, \gamma}[\mathcal{A}^{\dagger}E_{n}] \label{eqn:OPTcondIneq} \\ \nonumber
\end{eqnarray}
{\noindent}Appealing to Corollary~\ref{eqn:Corollary1} and Proposition~\ref{prop:Enbound}, we have that $\Phi_{\delta, \gamma}[\mathcal{A}^{\dagger}\mathbb{I}^{\star}\mathcal{A}C_{T_{\mathcal{M}}}]\leq \frac{\lambda_n}{6\beta}$, $\Phi_{\delta, \gamma}[C_{T'}] \leq C_2\lambda_n$ and  (with high probability) $\Phi_{\delta, \gamma}[\mathcal{A}^{\dagger}E_{n}] \leq \frac{\lambda_n}{6\beta}$. Consequently, based on the bound on $\lambda_n$ in assumption of Theorem~\ref{proposition:lambdaBound}, it is straightforward to show that $r \leq \min\{\frac{1}{4C'}, \frac{\alpha}{32m{\psi}C'^2}\}$. Hence by Proposition~\ref{lemma:Brower}, $\Phi_{\delta, \gamma}[\Delta] \leq \frac{1}{2C'}$. Finally, we can appeal to Proposition~\ref{lemma:Remainder} to obtain: 
\begin{eqnarray*}
\Phi_{\delta, \gamma}[\mathcal{A}^{\dagger}R_{\Sigma^\star} (\Delta_{})] \leq {2m{\psi}C'^2 \Phi_{\delta, \gamma}[\Delta]^2} \leq {2m{\psi}C'^2C_1^2\lambda_n^2} 
\leq  \Big[12{\beta}m{\psi}C'^2C_1^2\lambda_n\Big] \frac{\lambda_n}{6\beta} \leq \frac{\lambda_n}{6\beta}
\end{eqnarray*}

{\noindent}where the second to last bound comes from the bound on $\lambda_n$. Thus, the expression in \eqref{eqn:OPTcondIneq} can be further simplified to:
\begin{eqnarray*}
\Phi_{\delta, \gamma}[\mathcal{P}_{\mathbb{H}_{\mathcal{M}}}\mathcal{A}^{\dagger}\mathbb{I}^{\star}\mathcal{A}\mathcal{P}_{\mathbb{H}_{\mathcal{M}}}(\Delta)] \leq 2\lambda_n + 2 \lambda_n\Big(\frac{1}{6\beta} + \frac{1}{6\beta} + \frac{1}{6\beta}\Big) \leq 2\lambda_n + \frac{\lambda_n}{\beta} \leq \frac{17\lambda_n^{}}{8}
\end{eqnarray*}
\vspace{.1in}
{\noindent}The last bound follows since $\beta \geq 8$. Furthermore, 
\begin{eqnarray*}
\Phi_{\delta, \gamma}[\mathcal{A}^{\dagger} \mathbb{I}^{\star}\mathcal{A}(\Delta)] &\leq& \Phi_{\delta, \gamma}[\mathcal{P}_{\mathbb{H}_{\mathcal{M}}}\mathcal{A}^{\dagger}\mathbb{I}^{\star}\mathcal{A}\mathcal{P}_{\mathbb{H}_{\mathcal{M}}}(\Delta)] + \Phi_{\delta, \gamma}[\mathcal{P}_{\mathbb{H}_{\mathcal{M}}^\perp}\mathcal{A}^{\dagger}\mathbb{I}^{\star}\mathcal{A}\mathcal{P}_{\mathbb{H}_{\mathcal{M}}}(\Delta)] \\ &+& \Phi_{\delta, \gamma}[\mathcal{A}^{\dagger}\mathbb{I}^{\star}\mathcal{A}\mathcal{P}_{\mathbb{H}_{\mathcal{M}}^\perp}(\Delta)]\\
&\leq& \Phi_{\delta, \gamma}[\mathcal{P}_{\mathbb{H}_{\mathcal{M}}}\mathcal{A}^{\dagger}\mathbb{I}^{\star}\mathcal{A}\mathcal{P}_{\mathbb{H}_{\mathcal{M}}}(\Delta)] + (1-\nu)\Phi_{\delta, \gamma}[\mathcal{P}_{\mathbb{H}_{\mathcal{M}}}\mathcal{A}^{\dagger}\mathbb{I}^{\star}\mathcal{A}\mathcal{P}_{\mathbb{H}_{\mathcal{M}}}(\Delta)]  \\ &+& \Phi_{\delta, \gamma}[\mathcal{A}^{\dagger}\mathbb{I}^\star\mathcal{A}C_{T_{\mathcal{M}}}] \\
&\leq& \frac{17\lambda_n}{8} +\frac{17\lambda_n}{8}(1-\nu) + \frac{\lambda_n}{6\beta} <  \frac{17\lambda_n}{8} +\frac{17\lambda_n}{8} + \frac{\lambda_n}{48} < 5\lambda_n
\end{eqnarray*}

{\noindent}Note that we appeal to Fisher information Assumption 2 in \eqref{eqn:SecondFisherCond} in the second inequality.
\end{proof}

%%%%%%%%%

\subsubsection{From Tangent Constraints to the Original Problem}
\label{sec:backOrig}
Finally, we show that the tangent-space constraints in \eqref{eqn:ConvexRelaxed2_N} can be removed without altering the optimum value. More concretely,
\begin{proposition}
\label{proposition:ConvertFinal}
Suppose that the number of observed samples obeys $n \geq 4608 \beta^2m^2\psi^2C_{samp}'^2(p+q)$, and the regularization parameter $\lambda_n$ is chosen in the following range: \\\allowbreak$\lambda_n \in \Big[6\beta\sqrt{\frac{128(p+q)m^2\psi^2}{n}}, \frac{1}{C_{samp}'}\Big]$. Then, with probability greater than $1 - 2\text{exp}\Big\{-\frac{n\lambda_n^2}{4608\beta^2m^2\psi^2}\Big\}$, the optimum of \eqref{eqn:ConvexRelaxed2_N} is the same as the optimum of \eqref{eqn:ConvexRelaxed4_N}. \\
\end{proposition}
\begin{proof} 
Note that the bound on $n$ implies that the assumptions of Proposition~\ref{proposition:lambdaBound} are satisfied. Hence, we conclude that the solution of the tangent-space constrained program \eqref{eqn:ConvexRelaxed2_N}  is the same as the global optimum of the variety constrained program \eqref{eqn:NConveProblem_N} . Next, we show that the optimal variables of \eqref{eqn:ConvexRelaxed2_N} remain unchanged once the tangent-space constraints are removed. 
{\noindent}We proceed by proving that the optimum set of variables $(\tilde{\Theta}, \tilde{S}_Y, \tilde{L}_Y)$ satisfy the optimality conditions of \eqref{eqn:ConvexRelaxed4_N} given by:

\begin{eqnarray*}
[\Sigma_n - {\tilde{\Theta}}^{-1}]_{Y}  \in -\lambda_n\delta\partial \|\tilde{S}_Y\|_{\ell_1},  \hspace{.1in} [\Sigma_n - {\tilde{\Theta}}^{-1}]_{Y} \in \lambda_n\partial\|\tilde{L}_Y\|_\star\\[.05in]
[\Sigma_n - {\tilde{\Theta}}^{-1}]_{YX} \in -\lambda_n\partial\|\tilde{\Theta}_{YX}\|_\star, \hspace{.1in}[\Sigma_n - {\tilde{\Theta}}^{-1}]_{X}  = 0
\end{eqnarray*}

{\noindent}Equivalently, we show that $(\tilde{\Theta}, \tilde{S}_Y, \tilde{L}_Y)$ satisfy the following set of equations:

\begin{center}
\begin{enumerate}
\item $\mathcal{P}_{\mathbb{H}_{\mathcal{M}}} \mathcal{A}^{\dagger}(\Sigma_n - \tilde{\Theta}^{-1}) = (-\lambda_n\gamma_{1}\text{sign}({\tilde{S}_Y}), \hspace{.1in} \lambda_nUV' \hspace{.1in}  -\lambda_n\gamma_{2}{\breve{U}}{\breve{V}}',  \hspace{.1in} 0)$
\item $\Phi_{\delta, \gamma} [\mathcal{P}_{\mathbb{H}_\mathcal{M}^{\perp}} \mathcal{A}^{\dagger}(\Sigma_n - \tilde{\Theta}^{-1})] < \lambda_n$
\end{enumerate}
\end{center}

{\noindent}Here, $UDV'$ is the SVD decomposition of $\tilde{L}_Y$ and $\breve{U}\breve{D}\breve{V}'$ is the SVD decomposition of $\tilde{\Theta}_{YX}$. It is clear that the first condition is satisfied since the variables $(\tilde{\Theta}, \tilde{S}_Y, \tilde{L}_Y)$ are optimal with respect to \eqref{eqn:ConvexRelaxed2_N}. To prove that the second condition is met, it suffices to show that:
\begin{eqnarray}
\Phi_{\delta, \gamma}[\mathcal{P}_{\mathbb{H}_{\mathcal{M}}^{\perp}} \mathcal{A}^{\dagger}\mathbb{I}^{\star}\mathcal{A} \mathcal{P}_{\mathbb{H}_{\mathcal{M}}}(\Delta)]  &<& \lambda_n - \Phi_{\delta, \gamma}[\mathcal{P}_{\mathbb{H}_{\mathcal{M}}^{\perp}}\mathcal{A}^{\dagger} E_n] \\ &-& \Phi_{\delta, \gamma}[\mathcal{P}_{\mathbb{H}_{\mathcal{M}}^{\perp}} \mathcal{A}^{\dagger} R_{\Sigma^\star}(\Delta)] - \Phi_{\delta, \gamma}[\mathcal{P}_{\mathbb{H}_{\mathcal{M}}^{\perp}} \mathcal{A}^{\dagger} \mathbb{I}^{\star}\mathcal{A}C_{T_{\mathcal{M}}}] \nonumber
\label{eqn:FinalOpt}
\end{eqnarray}

{\noindent}We first note:
\begin{eqnarray*}
\Phi_{\delta, \gamma}[\mathcal{P}_{\mathbb{H}_{\mathcal{M}}}\mathcal{A}^{\dagger}\mathbb{I}^\star\mathcal{A}\mathcal{P}_{\mathbb{H}_{\mathcal{M}}}(\Delta)]  &\leq& \lambda_n+ 2\Phi_{\delta, \gamma}[\mathcal{A}^{\dagger}R_{\Sigma^\star}(\Delta)] + 2\Phi_{\delta, \gamma}[\mathcal{A}^{\dagger}\mathbb{I}^{\star} \mathcal{A}C_{T_{M}}]\\ &+& 2\Phi_{\delta, \gamma}[\mathcal{A}^{\dagger}E_{n}] \\
&\leq& \lambda_n^{} + \frac{\lambda_n}{\beta} = \frac{(\beta+1)\lambda_n}{\beta}
\end{eqnarray*}

{\noindent}Appealing to the Fisher information Assumption 1 in \eqref{eqn:FirstFisherCond}, we obtain:
\begin{eqnarray*}
\Phi_{\delta, \delta}[\mathcal{P}_{\mathbb{H}_{\mathcal{M}}^{\perp}} \mathcal{A}^{\dagger}\mathbb{I}^{\star}\mathcal{A} \mathcal{P}_{\mathbb{H}_{\mathcal{M}}}(\Delta)] &\leq& \frac{(\beta+1)\lambda_n}{\beta}\Big(1 - \frac{3}{\beta+1}\Big) = \lambda_n - \frac{2\lambda_n}{\beta}  < \lambda_n - \frac{\lambda_n}{2\beta}\\
&\leq& \lambda_n - \Phi_{\delta, \gamma}[\mathcal{A}^{\dagger}R_{\Sigma}(\Delta)] - \Phi_{\delta, \gamma} [\mathcal{A}^{\dagger}\mathbb{I}^{\star}\mathcal{A}C_{T_{\mathcal{M}}}] - \Phi_{\delta, \gamma} [\mathcal{A}^{\dagger}E_{n}] \\
&\leq& \lambda_n - \Phi_{\delta, \gamma}[\mathcal{P}_{\mathbb{H}_{\mathcal{M}}^{\perp}}\mathcal{A}^{\dagger}R_{\Sigma^\star}(\Delta)] -  \Phi_{\delta, \gamma}[\mathcal{P}_{\mathbb{H}_{\mathcal{M}}^{\perp}}\mathcal{A}^{\dagger}\mathbb{I}^{\star} \mathcal{A}C_{T_{\mathcal{M}}}]\\ &-&  \Phi_{\delta, \gamma}[\mathcal{P}_{\mathbb{H}_{\mathcal{M}}^{\perp}}\mathcal{A}^{\dagger}E_n]
\end{eqnarray*}

{\noindent}Here the last inequality holds since $\Phi_{\delta, \gamma}[\mathcal{P}_{\mathbb{H}_{\mathcal{M}}^\perp}( . )] \leq \Phi_{\delta, \gamma}(.)$  \end{proof}

{\noindent}\underline{\textbf{Proof of Theorem 1}}:  Letting $\bar{m} =  \max\{\delta, 1, \gamma\}$, we note that $C'$ in Proposition~\ref{proposition:lambdaBound} can be bounded as follows: $C' \leq 2\psi\text{deg}(S_Y^\star)\bar{m}$. Further, one can check that the constants in Proposition~\ref{proposition:lambdaBound} are related to the constants in Theorem~\ref{theorem:main} as follows: $C_{\sigma_{YX}}' \leq mC_{\sigma_{YX}}$, and subsequently $C_{samp}' \leq m\bar{m}^2\mathrm{deg}(S_Y^\star)C_{samp}$. Under the assumptions of Theorem~\ref{theorem:main}, we can appeal to Corollary~\ref{eqn:Corollary1}, Proposition~\ref{proposition:lambdaBound}, and Proposition~\ref{proposition:ConvertFinal} to conclude that with probability greater than $1 - 2\text{exp}\Big\{-\frac{n\lambda_n^2}{4608\beta^2m^2\psi^2}\Big\}$, the variables $(\hat{\Theta}, \hat{S}_Y, \hat{L}_Y)$ are structurally correct estimates of $({\Theta}^\star, {S}_Y^\star, {L}_Y^\star)$; that is: $\rm{sign}(\hat{S}_Y) = \rm{sign}(\hat{S}_Y^\star), \rm{rank}(\hat{L}_Y) = \rm{rank}(L_Y^\star), \rm{rank}(\hat{\Theta}_{YX}) = \rm{rank}(\Theta_{YX}^\star)$ and the estimates $(\hat{\Theta}, \hat{S}_Y, \hat{L}_Y)$ are ``close" to population quantities in appropriate norms: $ \Phi_{\delta, \gamma}[\hat{S}_Y - S_Y^\star, \hat{L}_Y - L_Y^\star, \hat{\Theta}_{YX} - \Theta_{YX}^\star, \hat{\Theta}_{X} - \Theta_{X}^\star] \leq 2C_1\lambda_n$.

\subsection{Choices of Parameters $\gamma, \delta$ in the estimator \eqref{eqn:originalproblem_NS}}
\label{sec:Parameterchoices}
In this section, we give conditions on the population Fisher information $\mathbb{I}^\star$ such that Assumptions 1 and 2 in \eqref{eqn:FirstFisherCond} and \eqref{eqn:SecondFisherCond} are satisfied by a non-empty set of values of $\delta, \gamma$. In the subsequent discussion in this section, we employ the following notation to denote restrictions of a subspace $\mathbb{H} = H_1 \times H_2 \times H_3 \times H_4  \subset \mathbb{S}^{ p} \times \mathbb{S}^{p} \times \R^{p \times q} \times \mathbb{S}^{q}$ (here $H_1,H_2,H_3,H_4$ are subspaces in $\mathbb{S}^{ p},\mathbb{S}^p,\R^{p \times q},\mathbb{S}^q$, respectively) to its individual components. The restriction to the first component of $\mathbb{H}$ is given by $\mathbb{H}[1] = H_1 \times \{0\} \times \{0\} \times \{0\} \subset \mathbb{S}^{ p} \times \mathbb{S}^{p} \times \R^{p \times q} \times \mathbb{S}^{q}$.  The restrictions $\mathbb{H}[2],\mathbb{H}[3],\mathbb{H}[4]$ to the other components of $\mathbb{H}$ are defined in an analogous manner.
 
As our first quantity, we consider the minimum gain of $\mathbb{I}^\star$ restricted to each of the tangent spaces $\Omega(S_Y^\star)$, and $T_Y^\star$, $T_{YX}^\star$ separately:
\begin{eqnarray}
\eta_1(\mathbb{H}^\star; \omega_Y, \omega_{YX}) = \min_{\substack{\mathbb{H} \in U(\omega_{Y},~ \omega_{YX})\\ i = {1,2,3,4}}} \min_{\substack{M \in \mathbb{H}[i]\\ \|M\|_{\Phi_{1,1}} = 1}} \|\mathcal{P}_{\mathbb{H}[i]}\A^\dagger\mathbb{I}^\star\A\mathcal{P}_{\mathbb{H}[i]}(M)\|_{\Phi_{1,1}}.
\label{eqn:eta1}
\end{eqnarray}
Here the set $U(\omega_{Y}, \omega_{YX})$ is defined in \eqref{eqn:Udef}. Recall that this set denotes the distortions around the population tangent spaces $T^\star_Y,T^\star_{YX}$). Notice also that there is no appearance of $\delta,\gamma$ in the norm $\Phi$.  The quantity $\eta_1(\mathbb{H}^\star; \omega_Y, \omega_{YX})$ being large ensures that $\mathbb{I}^\star$ is well-conditioned when restricted to each of the tangent spaces $\Omega(S^\star_Y), T'_Y, T'_{YX}$ separately.  The second quantity we consider is the maximal inner-product between elements in each of the tangent spaces $\Omega(S_Y^\star),T_Y^\star,T_{YX}^\star$ and those in their respective orthogonal complements (again, in the metric induced by $\mathbb{I}^\star$):
\begin{eqnarray}
\eta_2(\mathbb{H}^\star; \omega_Y, \omega_{YX}) = \max_{\substack{\mathbb{H} \in U(\omega_{Y},~ \omega_{YX})\\ i = {1,2,3,4}}} \max_{\substack{M \in \A(\mathbb{H}[i])\\ \|M\|_2 = 1}} \|\mathcal{P}_{\A(\mathbb{H}[i])^\perp}\mathbb{I}^\star\mathcal{P}_{\A(\mathbb{H}[i])}(M)\|_2
\end{eqnarray}
 
One additional aspect of Assumptions 1 and 2 (in \eqref{eqn:FirstFisherCond} and \eqref{eqn:SecondFisherCond}) that is not addressed via the quantities $\eta_1(\mathbb{H}^\star; \omega_Y, \omega_{YX}),\eta_2(\mathbb{H}^\star; \omega_Y, \omega_{YX})$ is the gain of the population Fisher information $\mathbb{I}^\star$ restricted to $\Omega_Y^\star \oplus T_Y^\star$.  Controlling this gain ensures that the tangent spaces $\Omega^\star_Y$ and $T^\star_Y$ have a transverse intersection in the metric induced by $\mathbb{I}^\star$; as discussed in previous work by \citet{Chand2012}, such a property is critical to ensure the accurate estimation of the latent-variable graphical model specifying the conditional distribution of $Y | f(X)$. Following the approach adopted in that work, we control the gain of $\mathbb{I}^\star$ restricted to $\Omega_Y^\star \oplus T_Y^\star$ via conditions involving three quantities. The first quantity $\rm{deg}(S_Y^\star)$ makes an appearance in Theorem~\ref{theorem:main} -- it is the maximum number of nonzeros per row/column of $S_{Y}^\star$, and it denotes the degree of the graphical model structure underlying the conditional distribution of $Y | f(X), \zeta$.  The degree of the sparse component $S^\star_Y$ being small ensures that the graphical model underlying $Y | f(X), \zeta$ is indeed a sparsely connected structure.  Bounds on the degree of a population graphical model play an important role in the literature in results on consistent graphical model selection \citep{MeiB2006,RavWRY2008}.  The second quantity is an \emph{incoherence} parameter, which played an important role in the literature on low-rank matrix completion \citep{CanR2009}.  Specifically, for a matrix $N \in \Sp^{p_1}$, the incoherence of the row-space / column-space of $N$ is given by:
\begin{equation}
\mathrm{inc}(N) \triangleq \max_{1 \leq i \leq p_1, 1 \leq j \leq p_1} \max\left\{ \|\mathcal{P}_{\mathrm{column}\text{-}\mathrm{space}(N)}(e_i)\|_{\ell_2}, \|\mathcal{P}_{\mathrm{row}\text{-}\mathrm{space}(N)}(e_j)\|_{\ell_2} \right\},
\label{eqref:incoherencedef}
\end{equation}
%as well as on sparse and low-rank matrix decomposition \citep{ChaSPW2009, CanLM2011}
where $\mathcal{P}$ denotes the projection operation and $e_i \in \mathbb{R}^{p_1}$ denotes the $i$'th standard basis vector. The incoherence parameter of the low-rank matrix $L_Y^\star$ being small ensures that the latent variables $\zeta$ affect most of the observed responses $Y$.  As developed by \citet{Chand2012}, the quantities $\rm{deg}(S_Y^\star)$ and $\rm{inc}(L_Y^\star)$ being small simultaneously ensures that the tangent spaces $\Omega_Y^\star \triangleq \Omega(S_Y^\star)$ and $T_Y^\star \triangleq T(L_Y^\star)$ are sufficiently transverse in the standard Euclidean inner-product.  To further ensure that the minimum gain of $\mathbb{I}^\star$ restricted to $\Omega_Y^\star \oplus T_Y^\star$ is bounded below (i.e., to certify transversality of $\Omega_Y^\star$ and $T_Y^\star$ in the metric induced by $\mathbb{I}^\star$), \citet{Chand2012} introduce the following quantity for $\mathbb{W} = \Sp^p \times \{0\} \times \{0\} \times \{0\} \subset \Sp^p \times \Sp^p \times \R^{p \times q} \times \Sp^q$:
\begin{eqnarray}
\eta_3(\Omega_Y^\star, T_Y^\star; \omega_Y) = \max\Bigg\{ \max_{\rho(T_Y', T_Y^\star) \leq \omega_Y} \max_{\substack{M \in T_Y' \\ \|M\|_{\ell_\infty} = 1}}  \|\mathcal{P}_{\A(\mathbb{W})}\mathbb{I}^\star M\|_{\ell_\infty}, \max_{\substack{M \in \Omega_Y^\star \\ \|M\|_{2} = 1}}  \|\mathcal{P}_{\A(\mathbb{W})}\mathbb{I}^\star M\|_{2} \Bigg\} \nonumber\\
\label{eqn:eta_3}
\end{eqnarray}
The reason for the statement of this definition in terms of the $\ell_\infty$ and spectral norms is that these are the dual norms of the regularizers employed in \eqref{eqn:originalproblem_NS}(recall the discussion in Section~\ref{section:Fishercond}). As shown by \citet{Chand2012} and as described in the following proposition, suitably controlling the quantities $\rm{deg}(S_Y^\star), \rm{inc}(L_Y^\star), \eta_3(\Omega_Y^\star, T_Y^\star; \omega_Y)$ leads to lower bounds on the minimum gain of $\mathbb{I}^\star$ restricted to $\Omega_Y^\star \oplus T_Y^\star$, which enables the accurate estimation of the latent-variable graphical model underlying $Y | f(X)$. \\
 
In the following proposition, we describe a set of conditions on the quantities $\eta_1(\mathbb{H}^\star; [\omega_Y, \omega_{YX}])$, $\eta_2(\mathbb{H}^\star;  [\omega_Y, \omega_{YX}])$, $\eta_3(\Omega_Y^\star, T_Y^\star; \omega_Y)$, $\rm{deg}(S_Y^\star)$, and $\rm{inc}(L_Y^\star)$, which lead to Assumptions 1 and 2 (14) and (15) (main theorem) being satisfied for $(\delta,\gamma)$ inside a polyhedral set.  We explicitly characterize this set and show that it is non-empty. For notational convenience, we denote $\eta_1^\star \triangleq \eta_1^\star(\mathbb{H}^\star; \omega_Y, \omega_{YX})$, $\eta_2^\star \triangleq \eta_2^\star(\mathbb{H}^\star; \omega_Y, \omega_{YX})$, and $\eta_3^\star \triangleq \eta_3^\star(\Omega_Y^\star, T_Y^\star; \omega_Y)$.
 
\begin{proposition}
\label{proposition:1}
Fix $\alpha > 0, \nu \in (0,1/3), \omega_Y > 0, \omega_{YX} > 0$. Let $\beta \triangleq \tfrac{3-\nu}{\nu}$.  Suppose that $(i)~\eta_1^\star \geq 2\alpha$, $(ii)~\eta_2^\star \leq \min \Big\{ \alpha(1-\frac{3}{1+\beta}), \sqrt{\frac{\alpha}{\beta}}\frac{[2~\rm{inc}(L_Y^\star) + \omega_{Y}]}{4(1-\omega_{Y})}, \frac{\alpha}{\beta\sqrt{2~\rm{deg}(S_{Y}^\star)}}, \frac{1}{2}(\frac{\alpha}{\beta})^{3/2}\Big\}$, $(iii)~\eta_3^\star \leq \sqrt{\frac{\alpha}{\beta}}$, and $(iv)~\frac{2~\text{inc}(L_Y^\star) + \omega_Y}{1-\omega_Y}\rm{deg}(S_Y^\star) \leq \frac{8\alpha}{\beta}$. Then Assumptions 1 and 2 in \eqref{eqn:FirstFisherCond} and \eqref{eqn:SecondFisherCond} are satisfied for all $(\delta,\gamma)$ in the following non-empty polyhedral set:
\begin{eqnarray*}
V(\alpha,\nu,\omega_Y,\omega_{YX}) = \Bigg\{ (\delta, \gamma)~ \Big{|} ~ \frac{[2~\rm{inc}(L_Y^\star) + \omega_{Y}]}{4(1-\omega_{Y})}\sqrt{\frac{\beta}{\alpha}} &\leq& \delta \leq \frac{2}
{\rm{deg}(S_Y^\star)}\sqrt{\frac{\alpha}{\beta}}; \\
\max\Big\{1, \eta_2^\star~\rm{deg}(S_Y^\star)\delta\frac{2\beta}{\alpha}\Big\} &\leq& \gamma \leq \frac{\min\{\delta,1\}}{\eta_2^\star} \frac{\alpha}{\beta}\Bigg\}.
\end{eqnarray*}
\end{proposition}
{\noindent}
Conditions analogous to $(iii),(iv)$ appear in previous work on latent-variable graphical model selection \citep{Chand2012}, and in our context they are useful for ensuring structurally correct estimates of the latent-variable graphical model corresponding to the conditional distribution of $Y | f(X)$. Conditions $(i), (ii)$ are relevant for simultaneously obtaining structurally correct estimates of the smallest dimension reduction $f(X)$ and of the latent-variable graphical model specifying $Y | f(X)$ via the convex program \eqref{eqn:originalproblem_NS}. See Definition~\ref{definition:1} for more details.

\begin{proof}
The proof of this proposition relies on two quantities. Given a matrix $M \in \Sp^{p_1}$, we defined the first quantity $\mu(\Omega(M))$ in \eqref{eqn:mu} with respect to the tangent space $\Omega(M)$. Additionally we define the quantity $\xi(T(M))$ with respect to the the tangent space $T(M)$:
$$
\xi(T(M)) \triangleq \max_{N \in T(M), \|N\|_{2} = 1} \|N\|_{\infty}.  \hspace{.1in}
$$
For extensive discussion regarding the properties of these quantities, we refer the reader to \citep{Chand2012}. Here, we highlight a few important facts. In particular, one can check that $\xi(T(M)) \in [\text{inc}(M),  2\text{inc}(M)]$  and $\mu(\Omega(M)) \leq \rm{deg}(M)$ (recall that $\text{inc}(M)$ measures the incoherence of $M$ and $\rm{deg}(M)$ denotes the maximum number of nonzero elements in any column or row of $M$). Furthermore, given two linear subspaces $T_1$ and $T_2$, the quantity $\rho(T_1,T_2)$ that measures the distortion of tangent spaces (see the definition in \eqref{eqn:Udef}) allows us to bound the variation in $\xi(T_2)$ as follows \citep{Chand2012}:
\begin{eqnarray*}
\xi(T_2) \leq \frac{1}{1-\rho(T_1, T_2)}[\xi(T_1) + \rho(T_1, T_2)]
\end{eqnarray*}
Returning to the proof, we show that for any $(\delta, \gamma)$ inside the polyhedron set $V(\alpha, \nu, \omega_Y, \omega_{YX})$, Fisher Assumption 1 and 2 in \eqref{eqn:FirstFisherCond} and \eqref{eqn:SecondFisherCond} are satisfied. First, using conditions $(ii)$ and $(iv)$, one can check that the polyhedron set  $V(\alpha, \nu, \omega_Y, \omega_{YX})$ is non-empty. Now, let $\mathbb{H} = \Omega_Y^\star \times T_Y' \times T_{YX}' \times \Sp^q$ be any subspace inside $U(\omega_{Y}, \omega_{YX})$ and the tradeoff parameters be chosen so that $(\delta, \gamma) \in V(\alpha, \nu, \omega_Y, \omega_{YX})$. Further let $\mathbb{Z} = \Sp^p \times \Sp^p \times \R^{p \times q} \times \Sp^q$, and let $(S_Y, {L}_Y, \Theta_{YX}, \Theta_{X}) \in \mathbb{H}$ with $\|S_Y\|_{\ell_\infty} \leq \delta$,  $\|{L_Y} \|_{2} \leq 1$, $\|\Theta_{YX} \|_{2} \leq \gamma$, $\|\Theta_{X} \|_{2} \leq 1$. Suppose equality holds in at least one of these set of inequalities so that $\Phi_{\delta, \gamma}(S_Y, L_Y, \Theta_{YX}, \Theta_{X}) = 1$. Then, at least one of the following cases is active:
\begin{enumerate}
\item Suppose $\|S_Y\|_{\ell_{\infty}} = \delta$. Then using conditions $(i) - (iii)$ of Proposition~\ref{proposition:1}, we have:
\begin{eqnarray*}
\frac{1}{\delta}\|\mathcal{P}_{\mathbb{H}[1]}\mathcal{A}^{\dagger}\mathbb{I}^{\star}\mathcal{A}(S_Y, {L}_Y, \Theta_{YX}, \Theta_{X})]\|_{\Phi_{1,1}} &\geq& \frac{1}{\delta}\Big[\|\mathcal{P}_{\mathbb{H}[1]}\mathcal{A}^{\dagger}\mathbb{I}^{\star}\mathcal{A}(S_Y, 0, 0, 0)\|_{\Phi_{1,1}}\\ &-& \|\mathcal{P}_{\mathcal{A}(\mathbb{H}[1])}\mathbb{I}^{\star}\mathcal{A}(0, {L}_Y, 0, 0)\|_{\ell_{\infty}} \\ &-& \|\mathcal{P}_{\mathcal{A}(\mathbb{H}[1])}\mathbb{I}^{\star}\mathcal{A}(0, 0, \Theta_{YX}, \Theta_X)\|_{\ell_{\infty}}\Big] \\  &\geq& \frac{1}{\delta}\Big[\|\mathcal{P}_{\mathbb{H}[1]}\mathcal{A}^{\dagger}\mathbb{I}^{\star}\mathcal{A}(S_Y, 0, 0, 0)\|_{\Phi_{1,1}}\\ &-&  \|\mathcal{P}_{\mathcal{A}(\mathbb{Z}[1])}\mathbb{I}^{\star}\mathcal{A}(0, {L}_Y, 0, 0)\|_{\ell_{\infty}} \\ &-& \|\mathcal{P}_{\mathcal{A}(\mathbb{Z}[1])}\mathbb{I}^{\star}\mathcal{A}(0, 0, \Theta_{YX}, \Theta_X)\|_{\ell_{\infty}}\Big] \\ &\geq& 2\alpha-\frac{\eta_3^\star\xi(T_{Y}')}{\delta} - \frac{2\eta_2^\star{\max\{\gamma,1\}}}{\delta}\\
&\geq& 2\alpha - \frac{(2~\rm{inc}(L_Y^\star) + \omega_Y)\eta_3^\star}{(1-\omega_Y)\delta} - \frac{2{\max\{\gamma,1\}}\eta_2^\star}{\delta} \\
&\geq& 2\alpha - \frac{4\alpha}{\beta} - \frac{2\alpha}{\beta} \geq 2\alpha - \frac{8\alpha}{\beta}
\end{eqnarray*}

\item Suppose $\|L_Y\|_2 = 1$. Then using conditions $(i) - (iii)$ of Proposition~\ref{proposition:1}, we have:
\begin{eqnarray*}
\|\mathcal{P}_{\mathbb{H}[2]}\mathcal{A}^{\dagger}\mathbb{I}^{\star}\mathcal{A}(S_Y, {L}_Y, \Theta_{YX}, \Theta_{X})\|_{\Phi_{1,1}} &\geq& \|\mathcal{P}_{\mathbb{H}[2]}\mathcal{A}^{\dagger}\mathbb{I}^{\star}\mathcal{A}(0, {L}_Y, 0, 0)\|_{\Phi_{1,1}}\\ &-& \|\mathcal{P}_{\mathcal{A}(\mathbb{H}[2])}\mathbb{I}^\star\mathcal{A}(S_Y, 0, 0, 0)\|_{2} \\ &-& \|\mathcal{P}_{\mathcal{A}(\mathbb{H}[2])}\mathbb{I}^\star\mathcal{A}(0, 0, \Theta_{YX}, \Theta_X)\|_{2}\\ &\geq&\|\mathcal{P}_{\mathbb{H}[2]}\mathcal{A}^{\dagger}\mathbb{I}^{\star}\mathcal{A}(0, {L}_Y, 0, 0)\|_{\Phi_{1,1}}\\ &-& 2\|\mathcal{P}_{\mathcal{A}(\mathbb{Z}[2])}\mathbb{I}^\star\mathcal{A}(S_Y, 0, 0, 0)\|_{2} \\ &-& 2\|\mathcal{P}_{\mathcal{A}(\mathbb{Z}[2])}\mathbb{I}^\star\mathcal{A}(0, 0, \Theta_{YX}, \Theta_X)\|_{2}\\
&\geq&  2{\alpha_{}} - 2\eta_3^\star\mu(\Omega_Y^\star)\delta - 4\eta_2^\star\max\{\gamma,1\} \\
&\geq&  2{\alpha_{}} - 2\eta_3^\star\rm{deg}(S_Y^\star)\delta - 4\eta_2^\star\max\{\gamma,1\} \\
&\geq& 2\alpha - \frac{4\alpha}{\beta} - \frac{4\alpha}{\beta} \geq 2\alpha - \frac{8\alpha}{\beta}
\end{eqnarray*}

\item Suppose $\|\Theta_{YX}\|_2 = \gamma$. Then using conditions $(i)$ and $(ii)$ of Proposition~\ref{proposition:1}, we have:
\begin{eqnarray*}
\frac{1}{\gamma}\|\mathcal{P}_{\mathbb{H}[3]}\mathcal{A}^{\dagger}\mathbb{I}^{\star}\mathcal{A}(S_Y, {L}_Y, \Theta_{YX}, \Theta_{X})]\|_{\Phi_{1,1}} &\geq& \frac{1}{\gamma}\Big[\|\mathcal{P}_{\mathbb{H}[3]}\mathcal{A}^{\dagger}\mathbb{I}^{\star}\mathcal{A}(0, 0, \Theta_{YX}, 0)]\|_{\Phi_{1,1}}\\ &-& \|\mathcal{P}_{\A(\mathbb{H}[3])}\mathbb{I}^{\star}\mathcal{A}(S_Y, L_Y, 0, \Theta_{YX})]\|_{2} \\ &\geq& \frac{1}{\gamma}\Big[\|\mathcal{P}_{\mathbb{H}[3]}\mathcal{A}^{\dagger}\mathbb{I}^{\star}\mathcal{A}(0, 0, \Theta_{YX}, 0)]\|_{\Phi_{1,1}}\\ &-& 2\|\mathcal{P}_{\A(\mathbb{Z}[3])}\mathbb{I}^{\star}\mathcal{A}(S_Y, L_Y, 0, \Theta_{YX})]\|_{2} \Big]\\
&\geq&  2{\alpha_{}} - \frac{2\eta_2^\star\mu(\Omega_Y^\star)\delta}{\gamma} - \frac{4\eta_2^\star}{\gamma} \\
&\geq& 2{\alpha_{}} - \frac{2\eta_2^\star\rm{deg}(S_Y^\star)\delta}{\gamma} - {4\eta_2^\star}\\
&\geq& 2{\alpha_{}} - \frac{\alpha}{\beta} -  \frac{\alpha}{\beta} \geq 2\alpha - \frac{8\alpha}{\beta}
\end{eqnarray*}
\item Suppose $\|\Theta_{X}\|_2 = 1$. Then using conditions $(i)$ and $(ii)$ of Proposition~\ref{proposition:1}, we have:
\begin{eqnarray*}
\|\mathcal{P}_{\mathbb{H}[4]}\mathcal{A}^{\dagger}\mathbb{I}^{\star}\mathcal{A}(S_Y, {L}_Y, \Theta_{YX}, \Theta_{X})]\|_{\Phi_{1,1}} &\geq& \|\mathcal{P}_{\mathbb{H}[4]}\mathcal{A}^{\dagger}\mathbb{I}^{\star}\mathcal{A}(0, 0, 0, \Theta_{X})\|_{\Phi_{1,1}}\\ &-& \|\mathcal{P}_{\A(\mathbb{H}[4])}\mathbb{I}^{\star}\mathcal{A}(S_Y, L_Y, \Theta_{YX}, 0)\|_{2}  \\ &\geq& 2\alpha_{}-\eta_2^\star\mu(\Omega_Y^\star)\delta - 2\eta_2^\star\max\{\gamma,1\} \\ &\geq& 2\alpha_{}-\eta_2^\star\rm{deg}(S_Y^\star)\delta - 2\eta_2^\star\gamma \\
&\geq& 2\alpha - {\frac{4\alpha}{\beta}} - \frac{2\alpha}{\beta} \geq 2\alpha - \frac{8\alpha}{\beta}\\
\end{eqnarray*}
\end{enumerate}
From these results, we conclude that $\Phi_{\delta, \gamma}[\mathcal{P}_{\mathbb{H}}\A^{\dagger}\mathbb{I}^\star\A(S_Y, L_Y, \Theta_{YX}, \Theta_X)] \geq 2\alpha - \frac{8\alpha}{\beta}$. Further, we can bound the quantity $\chi(\mathbb{H}, \|.\|_{\Phi_{\delta, \gamma}})$ in \eqref{eqn:Chieq} as follows
\begin{eqnarray}
\chi(\mathbb{H}, \|.\|_{\Phi_{\delta, \gamma}}) \geq 2\alpha - \frac{8\alpha}{\beta} \geq \alpha
\label{eqn:Re}
\end{eqnarray}
Using a similar decoupling technique, one can show:
\begin{eqnarray*}
\Phi_{\delta,\gamma} \Big[ \mathcal{P}_{\mathbb{H}^{\perp}} [\mathcal{A}^{\dagger}\mathbb{I}^{\star}\mathcal{A}(S_Y, {L}_Y, \Theta_{YX}, \Theta_{X})] \Big] \leq \eta_2^\star + \frac{8\alpha}{\beta} \leq \alpha(1 - \frac{3}{1+\beta}) + \frac{8\alpha}{\beta}
\end{eqnarray*}

{\noindent}Using this bound and the bound on $\chi(\mathbb{H}, \|.\|_{\Phi_{\delta, \gamma}})$, we control the quantity $\varphi(\mathbb{H}, \|.\|_{\Phi_{\delta, \gamma}})$ in \eqref{eqn:varphidef}:
\begin{equation}
\varphi(\mathbb{H}, \|.\|_{\Phi_{\delta, \gamma}}) \leq \frac{\Big(1-\frac{3}{1+\beta}\Big)\alpha + \frac{8\alpha_{}}{\beta}}{2\alpha_{} - \frac{8\alpha_{}}{\beta}} \leq 1-\frac{3}{1+\beta}
\label{eqn:perpRe}
\end{equation}

{\noindent}Since the bounds  \eqref{eqn:Re} and \eqref{eqn:perpRe} are valid for all $\mathbb{H} \in U(\omega_{Y}, \omega_{YX})$, Fisher information Assumptions 1 and 2 (in \eqref{eqn:FirstFisherCond} and \eqref{eqn:SecondFisherCond}) are satisfied for ($\delta , \gamma$) inside the polyhedron set $V(\alpha,\nu, \omega_Y, \omega_{YX})$. 
\end{proof}

{\subsection{High-Dimensional Consistency of the Estimator \eqref{eqn:CovariateSelection}}
\label{section:Covariate}

In this section, we discuss the consistency properties of the estimator \eqref{eqn:CovariateSelection} in a high-dimensional scaling regime. Specifically, suppose we observe samples $\{Y^{(i)}, X^{(i)}\}_{i = 1}^n \subset \R^{p+q}$ of a collection of jointly Gaussian responses and covariates $(Y, X)$ with joint population precision matrix $\Theta^\star = \begin{pmatrix} S_Y^\star-L_Y^\star & \Theta_{YX}^\star \\ {\Theta_{YX}^\star}' & \Theta_{X}^\star \end{pmatrix} \subset \mathbb{S}^{p+q}$, where $S_Y^\star$ is sparse, $L_Y^\star$ is low-rank, and $\Theta_{YX}^\star$ is column-sparse.  Supplying these observations into the program \eqref{eqn:CovariateSelection} and obtaining estimates $(\hat{\Theta}, \hat{S}_Y, \hat{L}_Y) \subset \Sp^{p + q} \times \Sp^p \times \Sp^p$, we prove in Theorem~\ref{theorem:2} that (under certain conditions on $\Theta^\star$ and with high probability) $(a)$ the column support of $\hat{\Theta}_{YX}$ is equal to the column support of ${\Theta}_{YX}^\star$, $(b)$ $\rm{rank}(\hat{L}_Y) = \rm{rank}(L_Y^\star)$, and $(c)$ $\rm{sign}(\hat{S}_Y) = \rm{sign}(S_Y^\star)$.  Thus, the subset of covariates that are sufficient for predicting the responses and the latent-variable graphical model specifying the conditional distribution of the responses given the covariates are both correctly identified.

Proceeding in a similar manner as in Section~\ref{section:Fishercond}, we prove that the estimator \eqref{eqn:CovariateSelection} is consistent under assumptions on the conditioning of the population Fisher information $\mathbb{I}^\star$.  These assumptions are stated in terms of tangent spaces of the algebraic variety of \emph{column-sparse} matrices.  Letting $M \in \mathbb{R}^{p{\times}q}$ be a matrix with $k$ nonzero columns, the tangent space at $M$ with respect to the variety of $p \times q$ matrices with at most $k$ nonzero columns is given by:
\begin{eqnarray*}
{F}(M) \triangleq \{J \in \mathbb{R}^{p \times q} ~|~ \text{columnsupport}(J) \subseteq \text{columnsupport}(M)\}.
\end{eqnarray*}
{\noindent}Here `columnsupport' denotes the indices of the nonzero columns. As in Section~\ref{section:Fishercond}, we control the conditioning of $\mathbb{I}(\Theta^\star)$ for all subspaces $\mathbb{H}'$ in the following set: \footnote{The variety of column-sparse matrices is locally flat around $\Theta_{YX}^\star$ so that the tangent spaces at all points in a neighborhood of $\Theta_{YX}^\star$ are all equal to $F(\Theta_{YX}^\star)$.}
\begin{equation}
\begin{aligned}
\tilde{U}{(\omega_{Y})} \triangleq \Big\{\Omega(S_Y^\star) \times T'_Y \times F(\Theta_{YX}^\star) \times \mathbb{S}^{q} ~|~ &\rho({{T_{Y}'}}, T({L_{Y}^\star})) \leq \omega_{Y} \Big\}.
\end{aligned}
\end{equation}
We control the quantities $\chi({\mathbb{H}}', \tilde{\Phi}_{\delta,\gamma})$ and $\varphi(\mathbb{H}', \tilde{\Phi}_{\delta,\gamma})$ (defined in \eqref{eqn:Chieq} and \eqref{eqn:varphidef}) for all $\mathbb{H} \in \tilde{U}{(\omega_{Y})}$ and for $\tilde{\Phi}_{\delta,\gamma}(S_Y, L_Y, \Theta_{YX}, \Theta_{X})$ defined as:
\begin{equation}\tilde{\Phi}_{\delta,\gamma}(S_Y, L_Y, \Theta_{YX}, \Theta_{X}) \triangleq \max\left\{\frac{\|S_Y\|_{\ell_\infty}}{\delta}, \|L_Y\|_{2}, \frac{\|\Theta_{YX}\|_{2,\infty}}{\gamma}, \|\Theta_{X}\|_2 \right\}.
\end{equation}
As with $\Phi_{\delta,\gamma}$ in Section~\ref{section:Fishercond}, the norm $\tilde{\Phi}_{\delta,\gamma}$ is a slight variant of the dual norm of the regularizer $\delta \|S_Y\|_{\ell_1} + \mathrm{trace}(L_Y) + \gamma \|\Theta_{YX}\|_{1,2}$ in \eqref{eqn:CovariateSelection}.

{\noindent}In summary, given $(\delta, \gamma ,\omega_Y) \in \mathbb{R}_+ \times \mathbb{R}_+ \times (0,1)$ we assume that the population Fisher information $\mathbb{I}^\star$ satisfies the following conditions:
 \begin{eqnarray*}
\mathrm{Assumption~3}&:& \inf_{\mathbb{H}' \in {\tilde{U}}{(\omega_{Y})}}\chi({\mathbb{H}}', {{\tilde{\Phi}}_{\delta,\gamma}})  \geq \alpha, ~~~ \mathrm{for~some~} \alpha > 0 \\
\mathrm{Assumption~4}&:& \sup_{\mathbb{H}' \in {\tilde{U}}{(\omega_{Y})}}\varphi({\mathbb{H}}', {\tilde{\Phi}_{\delta,\gamma}})  \leq 1-\nu, ~~~ \mathrm{for~some~} \nu \in (0,1/3).
\end{eqnarray*}

{\noindent}As with the notation preceding the statement of Theorem ~\ref{theorem:main}, let $\tau_{Y}$ denote the minimum nonzero entry in magnitude of ${S_{Y}^{\star}}$, let $\sigma_{Y}$ denote the minimum nonzero singular value of ${L_{Y}^{\star}}$, and let $\mathrm{deg}(S_{Y}^\star)$ denote the maximal number of nonzeros per row/column of $S_{Y}^\star$.  Further, let $\zeta_{YX}$ denote the minimum $\ell_2$ norm over nonzero columns of ${\Theta_{YX}^{\star}}$ and let $\kappa$ be the number of nonzero columns of $\Theta_{YX}^\star$.

\begin{theorem}
Suppose we are given i.i.d observations $\{Y^{(i)}, X^{(i)}\}_{i = 1}^n \subset \mathbb{R}^{p+q}$ of a collection of jointly Gaussian covariates/responses with population precision matrix $\Theta^\star \in \mathbb{S}^{p+q}_{++}$.  Fix $\alpha > 0, \nu \in (0,1/3), \omega_{Y} \in (0,1)$.  Suppose the trade-off parameters $\delta$ and $\gamma$ are chosen such that the population Fisher information $\mathbb{I}(\Theta^\star)$ satisfies Assumptions 3 and 4.\\

Let $m_{} \triangleq \max\{\frac{1}{\delta{}}, 1, \frac{1}{\gamma{}}\}$, $\bar{m} \triangleq \max\{\delta, 1, \gamma\}$, $\beta \triangleq \frac{3-\nu}{\nu}$, and $\psi \triangleq \|(\Theta^\star)^{-1}\|_2$. Further, let$C_1{} = \frac{24}{\alpha_{}} + \frac{1}{\psi^2}$, $C_2 = \frac{8}{\alpha_{}} (\frac{1}{3\beta} + 1)$, $C_{\sigma} = C_1^2\psi^2\max\{ 12\beta + 1, \frac{1}{{C_2}\psi^2}+1\}$, $C_{samp} = \max\{\frac{1}{48\psi\beta},{48{\beta}{\psi^3}C_1^2}, 8\psi{C_2}, \frac{128{\psi^3}C_2}{\alpha}\}$, and $\lambda_{\text{upper}} = \frac{1}{m\bar{m}^2\max\{\mathrm{deg}(S_Y^\star), \kappa\}C_{samp}}$. Suppose that the following conditions hold:
\begin{enumerate}
\item $n \geq \frac{4608\psi^2\beta^2m^2(p+q)}{\lambda_{\text{upper}}^2}$; that is $n \gtrsim \Big[\frac{\beta^4}{\alpha^2} m^4\bar{m}^4\max\{\mathrm{deg}(S_{Y}^\star)^2, \kappa^2\}\Big] (p+q) $
\item $\lambda_n \in \Big[\sqrt{\frac{4608\psi^2\beta^2m^2(p+q)}{n}}, \lambda_{\text{upper}}\Big]$; \hspace{.1in} e.g. $\lambda_n \sim \beta{m}\sqrt{\frac{p+q}{n}}$
\item $\tau{}_{Y} \geq 2\gamma_{1}{}C_1{}\lambda_{n}{}$; that is  $\tau{}_{Y}  \gtrsim \frac{\beta}{\alpha}{}m\bar{m}\sqrt{\frac{p+q}{n_{}}}$ \hspace{.05in} if \hspace{.05in} $\lambda_n \sim \beta{m}\sqrt{\frac{p+q}{n}}$
\item $\sigma_{Y}{} \geq {m_{}} C_\sigma\lambda_n{}$; that is $\sigma_{Y}{} \gtrsim \frac{\beta^2}{\alpha{w_Y}}m^2\sqrt{\frac{p+q}{n_{}}}$ \hspace{.05in} if \hspace{.05in} $\lambda_n \sim \beta{m}\sqrt{\frac{p+q}{n}}$
\item $\zeta_{YX}{} \geq 2\gamma{}C_1{}\lambda_{n}{}$; that is $\zeta_{YX}{} \gtrsim \frac{\beta}{\alpha}m\bar{m}\sqrt{\frac{p+q}{n_{}}}$ \hspace{.05in} if \hspace{.05in} $\lambda_n \sim \beta{m}\sqrt{\frac{p+q}{n}}$
\end{enumerate}
\vspace{.2in}

{\noindent}Then with probability greater than $1-2\exp\{-\frac{n\lambda_n^2}{4608 \beta^2 m^2 \psi^2}\}$, the optimal solution $(\hat{\Theta},\hat{S}_Y,\hat{L}_Y)$ of \eqref{eqn:CovariateSelection} with the observations $\{Y^{(i)}, X^{(i)}\}_{i = 1}^n$ satisfies the following properties:

\begin{enumerate}
\item sign($\hat{S}_Y$) = sign(${{S}^\star_{Y}}{}$), rank($\hat{L}_Y$) = rank(${{L}_Y^\star}{}$), and $\text{columnsupport}(\hat{\Theta}_{YX}) = \text{columnsupport}(\Theta_{YX}^\star)$. \\[.005in]
\item $\Phi_{\delta,\gamma}(\hat{S}_Y - {S_{Y}^\star}{}, \hat{L}_Y - {L_{Y}^\star}{}, \hat{\Theta}_{YX}{} - {\Theta_{YX}^\star}{}, \hat{\Theta}_{X}{} - {\Theta_{X}^\star}{}) \leq C_1{}\lambda_{n}{}$;  that is $\|\hat{S}_Y - S_Y^\star\|_{\ell_\infty} \lesssim \frac{\beta}{\alpha}{m}\delta\sqrt{\frac{p+q}{n}}$,  $\|\hat{L}_Y - L_Y^\star\|_{2} \lesssim \frac{\beta}{\alpha}{m}\sqrt{\frac{p+q}{n}}$, $\|\hat{\Theta}_{YX} - \Theta_{YX}^\star\|_{2,\infty} \lesssim \frac{\beta}{\alpha}\gamma{m}\sqrt{\frac{p+q}{n}}$, $\|\hat{\Theta}_{X} - \Theta_{X}^\star\|_{2} \lesssim \frac{\beta}{\alpha}{m}\sqrt{\frac{p+q}{n}}$ \hspace{.05in} if \hspace{.05in} $\lambda_n \sim \beta{m}\sqrt{\frac{p+q}{n}}$.
\end{enumerate}
\label{theorem:2}
\end{theorem}

{\noindent}The strategy for the proof of Theorem~\ref{theorem:2} is analogous to that of  Theorem ~\ref{theorem:main}. 

\end{document}